%% file: paper_arxiv.tex
\documentclass[11pt]{article}

\usepackage[a4paper,margin=1in]{geometry}

\usepackage{amsmath}
\usepackage{amsthm}

\usepackage{newtxtext}
\usepackage[subscriptcorrection]{newtxmath}

\usepackage{graphicx}
\usepackage{subcaption}
\usepackage{booktabs} 
\usepackage{array} 

\usepackage{xcolor} 
\makeatletter
\@ifundefined{XC@showfalse}{}{%
  \XC@showfalse 
}
\makeatother

\usepackage{hyperref}

\usepackage{orcidlink}
\usepackage{tabularray} 
\usepackage{multirow}

\usepackage{amsmath}
\usepackage{mathtools}
\usepackage{array}
\usepackage{natbib}
\usepackage[capitalize,noabbrev]{cleveref}
\usepackage{bm}
\usepackage{verbatim}
\usepackage{standalone}
\usepackage{bibentry}
\usepackage{dsfont}
\usepackage{setspace}
\usepackage{esint}
\usepackage{enumitem}
\usepackage[most]{tcolorbox}

\usepackage{float}
\usepackage{tikz}
\usetikzlibrary{matrix}
\usepackage{listings}
\usepackage{mathrsfs}
\newtheorem{theorem}{Theorem}

\newtheorem{lemma}{Lemma}
\newtheorem{corollary}{Corollary}
\newtheorem{proposition}{Proposition}

\newtheorem{definition}{Definition}

\newtheorem{remark}{Remark}

\graphicspath{{./art/}}

\usepackage[plain,noend]{algorithm2e}
\usetikzlibrary{arrows.meta}

\newenvironment{keywords}
{\small\textit{\hspace{8pt} Keywords: }}
{\ignorespacesafterend}

\title{Total robustness in Bayesian Nonlinear Regression}

\author{%
Mengqi Chen$^{1}$\quad
Charita Dellaporta$^{2}$\quad
Thomas B.~Berrett$^{1}$\quad
Theodoros Damoulas$^{1,3}$
\vspace{1.6ex}\\
\small $^{1}$Department of Statistics, University of Warwick, UK\\
\small $^{2}$Department of Statistical Science, University College London, UK\\
\small $^{3}$Department of Computer Science, University of Warwick, UK\\
\small Corresponding author: \texttt{mengqi.chen.2@warwick.ac.uk}
}

\date{}
\input{preamble_defs_arXiv}

\begin{document}

\maketitle
\begin{abstract}
Modern regression analyses are often undermined by covariate measurement error, misspecification of the regression model, and misspecification of the measurement error distribution. We present, to the best of our knowledge, the first Bayesian nonparametric learning framework targeting total robustness to all three challenges in general nonlinear regression. Our framework places a joint Dirichlet process prior on the latent covariate--response distribution and updates it with posterior pseudo-samples of the latent covariates, so that inference is calibrated to the joint law. This yields estimators defined by minimizing the discrepancy between posterior realizations of the joint Dirichlet process and the model-implied joint distribution. We establish generalization bounds and provide a first proof of convergence and consistency of the resulting estimators under non-degenerate measurement error. A gradient-based implementation enables efficient computation; simulations and two real-data studies show improved stability to misspecification under increasing measurement error relative to recent Bayesian and frequentist alternatives.
\end{abstract}


\section{Introduction}
\label{sec:Intro_v2} 
\subsection{Background and motivation}
Contemporary robust regression often faces three ubiquitous threats: covariate measurement error (ME), regression model misspecification, and misspecification of the ME distribution. However, existing robust approaches typically target only one of them, failing when also confronted by another. In regression, these problems are linked because inference depends on how the latent covariate relates to the observed response. ME obscures that relationship, while misspecification can distort how it is modelled and assessed. We introduce a Bayesian framework that simultaneously protects against all three sources of bias, delivering \emph{total robustness}. To motivate this goal, and to clarify the limits of current methodology, we begin by examining the distinct challenges posed by ME and the many forms of misspecification.

\emph{ME in covariates} arises in numerous fields, including economic, biomedical, and environmental studies, where recorded values deviate from the true unknown signal \citep{hausman1995nonlinear,brakenhoff2018measurement,haber2021bias,curley2022nonlinear}. This discrepancy can be classical (when the observation is a noisy version of the true covariate) or Berkson (when the true covariate has a random offset from a nominal target), as established by \citet{berkson1950there} and summarized by \citet{carroll2006measurement}. These two canonical forms of ME, which we focus on in this paper, are part of a broader typology that also includes differential, multiplicative, and systematic ME \citep[see][Chapter 1 for a full taxonomy]{buonaccorsi2010measurement}. If ignored, ME commonly biases estimates and can degrade inference on quantities of interest \citep{gustafson2003measurement}. A substantial literature addresses ME \citep{deming1943statistical, stefanski1985effects, cook1994simulation,wang2004estimation, schennach2013regressions,hu2020measurementerrormodelsnonparametric}, yet many rely on restrictive assumptions, such as known error distributions or replicate measurements \citep{delaigle2006nonparametric, mcintyre2011density}. Nonparametric approaches, including deconvolution and Bayesian frameworks, were developed to alleviate such assumptions \citep{delaigle2016methodology,dellaporta2023robust}. For regression, however, the aim is not simply to de-noise the covariate, but to recover the part of latent-covariate information that remains relevant to the response.

\emph{Model misspecification} arises whenever the model family cannot explain the true data-generating process (DGP). In regression, this may occur because the true regression function lies outside the assumed parametric family, the noise distribution is misspecified, or the data contain a fraction of outliers generated by a different law.  Classical robust frequentist approaches such as Huber’s M-estimators and Hampel’s influence-curve framework aim to limit the impact of atypical observations \citep{huber1992robust,hampel1974influence}.  In the Bayesian paradigm, generalized posteriors replace the likelihood with a loss or divergence to maintain coherent inference under model misspecification \citep{bissiri2016general,grunwald2017inconsistency,jewson2018principles,knoblauch2022optimization} and extend to intractable likelihood problems \citep{matsubara2024generalized}.  Divergence-based updates built on the maximum mean discrepancy (MMD) \citep{gretton2012kernel} further reduce sensitivity to contamination without requiring precise knowledge of the misspecification form \citep{briol2019statistical,alquier2024universal}. Although these techniques mitigate bias from regression-model misspecification, they assume accurately measured covariates. Extensions to ME settings are nontrivial because influence-function calculations and divergence minimization typically rely on unbiased estimating equations in the covariates. Replacing latent covariates by error-prone measurements perturbs the covariate--outcome joint distribution, and can invalidate the unbiasedness and regularity conditions that are required for influence-function and estimating-equation arguments.

\emph{Misspecification in the ME mechanism} is also damaging: \citet{yi2021estimation} show that even advanced corrections become biased when the assumed error law is wrong in parameter estimation and hypothesis tests. \citet{roy2006flexible} develop an expectation-maximization (EM) algorithm for generalized linear models (GLM) with heavy-tailed ME (Student-t) and potentially multimodal covariate distributions, but they rely on external validation data to identify the ME variance. Later, \citet{cabral2014multivariate} relax this requirement by imposing a Student-t family for the ME distribution and estimating its parameters via EM, though their framework is restricted to linear regression models. More flexible estimators, such as the phase-function technique of \citet{delaigle2016methodology}, avoid the need for a known ME distribution, but do not deal with regression model misspecification.

Existing literature has presented several strands of work that mitigate either ME bias or misspecification, and a smaller but growing body attempts to address them \emph{simultaneously}. Corrected-score methods, which adjust the score equations so that their expectation remains zero in the presence of classical ME \citep{nakamura1990corrected}, were extended by \citet{huang2014corrected}, who studied their pathology under sizeable ME and proposed trend-constrained corrected scores. \citet{huang2016dual} analysed maximum likelihood estimation under the coexistence of ME and model misspecification in GLMs. These approaches handle ME and misspecification through corrected estimating equations or likelihood-based adjustments, but they still require accurate knowledge of error moments or strong parametric assumptions. \citet{zhang2018robust} handle outliers and ME in longitudinal data using robust estimating equations, but their framework requires replicate measurements and is restricted to linear models. Recent Bayesian frameworks \citep{dellaporta2023robust} place priors on latent covariate distributions to target ME uncertainty, but they do not address the combined regime of non-degenerate ME and regression-model misspecification. Although there has been incremental progress in this joint regime, existing methods still rely on auxiliary data or restrictive assumptions. Researchers who recognize this gap have called for unified frameworks to tackle ME and misspecification simultaneously \citep{gustafson2002simultaneous,hu2020measurementerrormodelsnonparametric,zhou2023gaussian}.

To answer these calls, we formalize \emph{total robustness} as simultaneous robustness to covariate ME, regression-model misspecification, and misspecification of the ME distribution in general nonlinear regression. Robustness in this regime is naturally expressed through the unobserved covariate--outcome joint distribution, since both ME and regression misspecification perturb it. The observed response therefore remains informative about the latent covariate through the regression relation, and we take this joint law as the nonparametric learning target. The basis of our framework is \emph{Bayesian nonparametrics}, which enables us to model unknown distributions and incorporate prior information via Dirichlet processes (DPs), providing the flexibility required for Bayesian total robustness. To explain this claim and situate our contribution, we now survey the existing Bayesian nonparametric work on robustness.

\emph{Bayesian nonparametric} methods are widely used for flexible modelling and, among other benefits, can reduce sensitivity to model misspecification. DP mixtures flexibly describe unknown distributions, such as heavy-tailed residuals or random effects, thus reducing the dependence on distributional assumptions \citep{muller1997bayesian,neal2000markov,lee2020bayesian}. Gaussian process (GP) priors, meanwhile, allow decision makers to place nonparametric priors over functions, facilitating complex regression relations to be captured without fixing a functional form \citep[Section~5]{gramacy2020surrogates}; see also \citet{zhou2023gaussian}. The Bayesian semiparametric regression of \citet{sarkar2014bayesian} combines B-spline mixtures with a DP mixture prior for the covariate density to target some classes of heteroscedastic ME but requires replicated measurements and does not deliver theoretical guarantees. A related line of work \citep{dellaporta2023robust} addresses ME by placing a DP prior on the latent covariate distribution \emph{alone}, and pairing the response-agnostic DP samples with the observed outcomes. In their approach, the DP posterior updates ignore ME, so any correction for ME enters only through the prior centring measure, where latent covariate draws are then generated only given their noisy observations, breaking the regression-induced dependence structure, which is crucial for regression. Consequently, their error bounds scale with the variance of the ME. By contrast, our framework places the DP prior directly on the \emph{joint} distribution of the latent covariate and response and updates it using latent-variable pseudo-samples informed by the observed outcomes. This joint, response-informed formulation yields posterior summaries that separate ME uncertainty from regression-model misspecification, and it enables convergence and consistency guarantees for the resulting estimators under non-degenerate ME, overcoming their key limitations. We will further clarify this distinction in Section~\ref{sec:Methodology_v2}.

We therefore present a \emph{unified framework for total robustness} based on Bayesian nonparametric learning (NPL) \citep{lyddon2018nonparametric, fong2019scalable} that simultaneously handles covariate ME, regression-model misspecification, and misspecification of the ME distribution in general nonlinear models. Our framework is designed to be flexible, enabling decision makers to tune prior strength and choose whether to sample latent covariates or to work directly with ME-prone observations. At the same time, our theory isolates the effect of each decision through generalization bounds and convergence properties of the resulting estimators. This addresses the long-standing robustness gap and offers a blueprint for trustworthy regression in complex, error-prone, and data-driven settings.


\subsection{Problem setting}
\label{sec:problem-setting}
Let \((\Omega,\mathscr F,\P)\) be a complete probability space on which all random variables are defined. For a fixed dimension \(d\ge 1\), let \(\mathcal X,\mathcal W\subseteq\R^{d}\) denote the spaces of latent covariates \(X\) and their noisy observations \(W\), and let \(\mathcal Y\subseteq\R\) denote the outcome space. We observe i.i.d. pairs \((W_i,Y_i)\sim \P_{WY}^0\ (i=1,\ldots,n)\).

The covariates are generated by a mechanism involving a latent \(X\) and one of two standard forms of ME (shown in Fig.~\ref{fig:intro-graphs}; application examples are listed in Table~\ref{tab:mot_examples}):
\[
\begin{aligned}
\textit{Classical ME:}\quad & X_i \sim \P_X^0,\quad N_i \sim F_N^{0},\quad N_i \ind X_i,\quad W_i = X_i + N_i,\\
\textit{Berkson ME:}\quad& W_i \sim \P_W^{0},\quad N_i \sim F_N^{0},\quad N_i \ind W_i,\quad X_i = W_i + N_i.
\end{aligned}
\]

The response relates to the latent covariate via a nonlinear regression function,
\[
Y_i = g^0(X_i) + E_i,\qquad E_i \sim F_E^0,\quad E_i \ind (X_i,N_i).
\]
Throughout, \(\P^{0}_{\cdot}\) denotes the unknown data-generating law, whereas \(\P_{\cdot}\) (without superscript) denotes the working model. Differences between \(\P_{\cdot}\) and \(\P^{0}_{\cdot}\) can arise at three distinct levels:

\begin{figure}[ht]
    \centering
    \input{Figures/casual_diagram}
  \caption{Graphical representation of the regression structure without ME and under the two canonical ME mechanisms considered in this paper. In classical ME, the observed covariate is a noisy measurement of the latent covariate; in Berkson ME, the latent covariate is a perturbation of the observed value.}
  \label{fig:intro-graphs}
\end{figure}

\begin{table}[ht]
\centering
\small
\begin{tabular}{@{}p{2.7cm} p{7.5cm} p{3.2cm}@{}}
\toprule
\textbf{ME type} & \textbf{Example} & \textbf{Literature}\\
\midrule
\multirow[c]{8}{*}{\centering Classical}%
  & Economic study of Engel curves 
  & \multirow{4}{3.2cm}{\raggedright\cite{hausman1995nonlinear}}\\
  & \quad\(X=\) true household expenditure &\\
  & \quad\(W=\) self-reported expenditure from surveys &\\
  & \quad\(Y=\) budget share of some commodity (e.g. food) &\\\cmidrule{2-3}
  & Effect of potassium intake on health outcomes 
  & \multirow{4}{3.2cm}{\raggedright\cite{curley2022nonlinear}}\\
  & \quad\(X=\) true long-term intake of potassium &\\
  & \quad\(W=\) potassium intake converted from self-reported diet &\\
  & \quad\(Y=\) systolic blood pressure &\\
\midrule[\heavyrulewidth]
\multirow[c]{8}{*}{\centering Berkson}%
  & Relationship between body fat level and risk of diabetes 
  & \multirow{4}{3.2cm}{\raggedright\cite{haber2021bias}}\\
  & \quad\(X=\) true body fat percentage &\\
  & \quad\(W=\) BMI-predicted body fat percentage&\\
  & \quad\(Y=\) blood sugar level (HbA1c) &\\\cmidrule{2-3}
  & Effect of air pollution on respiratory health 
  & \multirow{4}{3.2cm}{\raggedright\cite[Section 6]{schennach2013regressions}}\\
  & \quad\(X=\) true exposure to pollution for each individual &\\
  & \quad\(W=\) state-level average of pollution measure &\\
  & \quad\(Y=\) respiratory health indicator &\\
\bottomrule
\end{tabular}
\caption{Applications that involve classical or Berkson ME. 
For each example we specify the latent covariate $X$, its noisy observation $W$ and the response $Y$.}
\label{tab:mot_examples}
\end{table}

\emph{ME mechanism.}
     The working ME density may deviate from the truth. A common example is scale miscalibration \(f_{N,\tau}(u)=\tau^{-1} f_N^0(u/\tau)\). See \citet{yi2021estimation} for examples.

\emph{Regression function.}
      The parametric family chosen by the decision maker
      \(\{g(\cdot,\theta):\theta\in\Theta\}\), where
\(g:\mathcal X\times\Theta\to\R\) is a nonlinear parametric regression function indexed by \(\theta\in\Theta\subseteq\R^{p}\),
      need not contain the regression function: 
      \(
        g^0(\cdot)\notin
        \bigl\{g(\cdot,\theta):\theta\in\Theta\bigr\}
      \).
      
\emph{Outcome noise.}
      Heavy tails, contamination, or heteroskedasticity may render the working distribution \(F_E\) different from the true \(F_E^0\). A well-known class of outcome noise misspecification is the Huber contamination model \citep{huber1992robust}, where \(F_E^0=(1-\eta)F_E
      +\eta Q_E\) with contamination ratio \(\eta\in[0,1)\).
      
Table~\ref{tab:me_review} gathers notable references on regression with ME, organized by the error mechanism (classical or Berkson), the regression type (linear or nonlinear), and the types of misspecification they address.

Our goal is to find $\theta_0\in\Theta$ such that $g(\cdot, \theta_0)$ recovers $g^0(\cdot)$ as accurately as possible, despite the existence of ME and misspecification coming from all aspects of the model. In Section~\ref{sec:Methodology_v2}, we formalize this by defining the \emph{optimal} estimator in the MMD sense,
$
  \theta_0 =
  \arg\min_{\theta\in\Theta}
  \MMD\bigl(\P_{XY}^0,\P^{\theta}_{XY}\bigr)
$,
where $\P^{\theta}_{XY}$ denotes the joint law of $(X,Y)$ induced by the working regression function $g(\cdot,\theta)$ and the outcome noise model $F_E$. The MMD is the distance between kernel mean embeddings of probability measures in a Reproducing Kernel Hilbert Space (RKHS). With characteristic kernels, the MMD is a robust metric that has been a popular choice as an optimization target in recent robustness literature. It limits the influence of outliers under bounded kernels, and admits unbiased U-statistic estimators with straightforward stochastic gradient calculations \citep{gretton2012kernel,briol2019statistical,alquier2024universal,cherief2025parametric}.

\begin{table}[ht]
\centering\small
\begin{tabular}{lccccc}
\toprule
\textbf{Method} & \textbf{Error type} & \textbf{Regression type} & \textbf{RF} & \textbf{RN} & \textbf{MEM}\\
\midrule
\cite{deming1943statistical}         & C               & Linear                          & $\times$ & $\times$ & $\times$\\
\cite{berkson1950there}              & B               & Linear                          & $\times$  & $\times$ & $\times$\\
\cite{zamar1989robust}               & C               & Linear                          & $\times$ & $\times$ & $\checkmark$\\
\cite{nakamura1990corrected}         & C               & Nonlinear (GLM)                & $\times$  & $\times$ & $\times$\\
\cite{cook1994simulation}            & C               & Nonlinear (parametric)               & $\times$  & $\times$ & $\times$\\
\cite{berry2002bayesian}             & C               & Nonlinear (splines)             & $\checkmark$  & $\times$ & $\times$\\
\cite{schennach2013regressions}      & B               & Nonlinear (Instrumental Variable)               & $\times$  & $\times$ & $\checkmark$\\
\cite{zhou2023gaussian}        & C               & Nonlinear (Gaussian Process)                  & $\checkmark$ & $\times$ & $\checkmark$\\
\cite{dellaporta2023robust}          & C \emph{or} B   & Nonlinear        & $\times$  & $\times$ & $\checkmark$\\
\textbf{Present paper (2026)}        & C \emph{or} B   & Nonlinear  & $\checkmark$ & $\checkmark$ & $\checkmark$\\
\bottomrule
\end{tabular}
\caption{Representative methods for regression with ME. C = classical ME; B = Berkson ME. RF = target regression-function misspecification; RN = target outcome-noise misspecification; MEM = target misspecification of the ME distribution.}
\label{tab:me_review}
\end{table}

\subsection{Main contributions}
We develop a unified framework that learns the latent covariate--response joint distribution under Berkson or classical ME, while remaining robust to joint misspecification of (i) the regression model, (ii) the outcome-noise law, and (iii) the ME distribution. Our framework is flexible: prior strength can be tuned based on confidence levels, and decision makers may either pseudo-sample latent \(X\) or work directly with ME-prone observations \(W\), depending on the scale of ME and the reliability of the pseudo-sampling procedure. We provide a thorough theoretical assessment via finite-sample generalization bounds that offer interpretable guarantees through a decomposition of excess risk. We also establish consistency of the NPL estimator across variants of the framework. Practically, we implement a posterior bootstrap that combines Hamiltonian Monte Carlo (HMC)-based pseudo-sampling of latent $X$ with gradient-based MMD minimisation. In simulations and two real-data studies, the method yields lower estimation error and greater stability under misspecification as ME increases, compared to recent robust Bayesian and frequentist methods. Code to reproduce results in this paper is available at \texttt{https://github.com/MengqiChenMC/tot_robust_code}.

\section{Methodology}
\label{sec:Methodology_v2}
\subsection{Overview of methodology}
\label{sec:methodology_overview}
We retain parametric structure for the inferential target of the regression function \(g(\cdot,\theta)\), and we model all remaining components nonparametrically. In the spirit of Bayesian NPL, we place a DP prior on the unknown joint law $\P_{XY}$, covering both Berkson and classical ME regimes without fixing a parametric likelihood. The regression parameter \(\theta\) enters through \(g(\cdot,\theta)\) and is learned via minimizing the MMD between the nonparametric DP posterior of the joint distribution of \((X,Y)\) and the $\theta$-implied joint distribution of \((X,Y)\). The latent covariate values required to update the DP are handled by posterior pseudo-sampling. Fig.~\ref{fig:workflow-comparison} illustrates our DP construction and compares it with that of \citet{dellaporta2023robust}.

\begin{figure}[t]
\centering
\input{Figures/diagram_comparison} 
\caption{Two DP constructions under Berkson ME. Dotted arrows denote prior-driven components; dashed arrows denote posterior-updated components. Panel (a) shows our framework, where the DP is built on the joint law and updated using pseudo-samples $\tilde X$ informed by $W$ and $Y$. Panel (b) shows the construction of \citet{dellaporta2023robust}, where latent covariates are sampled from $W$ alone and then paired with $Y$, so the DP reference distribution breaks the covariate--response dependence.}
\label{fig:workflow-comparison}
\end{figure}

\subsection{MMD target}
\label{sec:MMD-target}
We begin by recalling the definition of the MMD, our chosen loss, from \citet{gretton2012kernel}. 
\begin{definition}[MMD with characteristic kernel]
    Given an RKHS $(\Hilb,k)$ with characteristic kernel $k$, the MMD between two probability measures $P$ and $Q$ on~$\mathcal X$ is \(\MMD(P,Q)=\norm{\mu_P-\mu_Q}_{\Hilb}\), where \(\mu_P(\cdot):=\int k(\cdot,x)\dd P(x)\).
\end{definition}
Here, $k$ being characteristic means that $\MMD(P,Q)=0$ if and only if $P=Q$ \citep{gretton2012kernel}. This is satisfied by common kernel choices (e.g. Gaussian kernels, Matérn kernels, Laplace kernels). Therefore, minimizing MMD aligns two distributions without requiring a correctly specified likelihood. MMD-based losses are commonly used in the robustness literature: by comparing distributions in feature space, they can reduce sensitivity to atypical observations under heavy tails, outliers, or model misspecification \citep{briol2019statistical,alquier2024universal}. These properties are useful in our setting, where both the regression model and the ME mechanism can be misspecified.

To formalize our loss target, assume hypothetically that we know the true conditional laws \(\P^0_{X\mid w}\) (capturing ME) and \(\P^0_{Y\mid x}\) (capturing the regression model). Then, we could form
\[
  \P_{XY}^0=\int_{\mathcal{W}} \P^0_{XY\mid w}F^0_W(\dd w) = \int_{\mathcal{W}} \P^0_{X\mid w}\times \P^0_{Y\mid X}F^0_W(\dd w),
\]
where \(F^0_W\) is the marginal law of \(W\). Suppose we posit a parametric family \(\{g(\cdot,\theta):\theta\in\Theta\}\) for the regression function. Write \(\P_{g(x, \theta)}(\,\cdot\,)\) for the \(g(\cdot,\theta)\)-induced conditional law of \(Y\) given \(X=x\), i.e. \(\P_{g(x, \theta)}(\,\cdot\,)=\operatorname{Law}\{Y\mid X=x;\theta\}\). The resulting model-implied joint distribution of \((X,Y)\) is
\[
   \P^{\theta}_{XY} = \P_X^0 \times \P_{g(X, \theta)}
  = \int_{\mathcal{W}} \P^{0}_{X\mid w} \times \P_{g(X, \theta)}F^0_W(\dd w).
\]
We then define the optimal $\theta_0$ in the MMD sense: $\theta_0=\argmin_{\theta \in \Theta}\MMD\Bigl(\P_{XY}^0,\P^{\theta}_{XY}\Bigr)$. In reality, neither $\P^0_{XY\mid w}$ nor $\P^0_{X\mid w}$ is known. We therefore place DP priors on these laws, leading to the Bayesian NPL procedure described next. 

\subsection{DP-based Bayesian NPL framework: general formulation}
\label{sec:DP-Constr}
To model the latent covariates $X$ and the response $Y$ without restrictive parametric assumptions, we place a DP prior, $\DP(c,F)$, on the joint law $\P_{X, Y \mid w}$ \citep{ferguson1973bayesian}. 
\begin{definition}
For every measurable partition $\left(A_1, \ldots, A_k\right)$ of the sample space, a random measure $P \sim \DP(c, F)$ satisfies
$
\left(P\left(A_1\right), \ldots, P\left(A_k\right)\right) \sim \operatorname{Dir}\left(c F\left(A_1\right), \ldots, c F\left(A_k\right)\right).
$
\end{definition}
The concentration parameter $c>0$ determines how tightly draws of $P$ concentrate around the base measure $F$. A key property for our framework is \emph{conjugacy}: after observing data $z_{1: m}$, the posterior remains a DP, 
$
P \mid z_{1: m} \sim \DP\left(c+m, \frac{c}{c+m} F+\frac{1}{c+m} \sum_{j=1}^m \delta_{z_j}\right).
$
Hence each posterior draw simply reweights the empirical atoms and the prior. 

Let \(\{w_i\}_{i=1}^n\) be the observed noisy covariates, with unknown true \(\{x_i\}_{i=1}^n\). For each $w_i$, we define a base measure \(\Q_{XY,i}\) that reflects any initial beliefs about \((X,Y)\) (discussed in Section~\ref{sec:prior-construct}). Given data \(\{(\tilde{x}_{i,j}, y_i)\}_{j=1}^m\), where \(\{\tilde{x}_{i,j}\}_{j=1}^m\) are posterior pseudo-samples of the unobserved covariate \(x_i\) given \(w_i\) and \(y_i\) (discussed in Section~\ref{sec:pseudo-sampling}), we propose the DP framework by conjugacy:
\label{def:joint-DP}
    \begin{align}
  \text{Prior: } &
    \P_{XY\mid w_i}
  \sim
  \DP\bigl(c,\Q_{XY,i}\bigr), \label{eq:joint-prior} \\
  \text{Posterior: } &\P_{XY \mid w_i} \bigm|
  \bigl\{(\tilde{x}_{i,j},y_i)\bigr\}_{j=1}^m
  \sim
  \DP\Bigl(c + m,\frac{c}{c + m}\Q_{XY,i}+\frac{1}{c + m}\sum_{j=1}^m \delta_{(\tilde{x}_{i,j},y_i)}\Bigr). \label{eq:joint-posterior}
\end{align}
Each posterior realization from the posterior DP (\ref{eq:joint-posterior}) is a distribution over \((X,Y)\).

\subsection{Constructing the prior centring measure \texorpdfstring{$\Q_{XY,i}$}{Qi}}
\label{sec:prior-construct}
The joint prior centring measures $\Q_{XY,i}$ encode prior knowledge about $(X,Y)$. In practice, this could be a prior built upon historical data or domain knowledge. In the absence of such information, we can often define $\Q_{XY,i}$ through three components - a prior on $\theta$, a prior on $X\mid W$, and a prior for $Y\mid X$:
\begin{enumerate}
    \item Prior on $\theta$: let $\theta$ have prior density $f(\theta)$, for example a normal distribution centred at an initial estimate or a uniform distribution on a compact parameter space $\Theta$.
    \item Prior on $X\mid W$: We denote the marginal prior for $\P_{X\mid w_i}$ as $\Q_{X,i}$, which needs to be defined differently in the classical and Berkson ME cases. Denote $F_N$ (with density $f_N$) as our assumed ME distribution, which is not necessarily the same as the true ME $F_N^0$.
    \begin{enumerate}
        \item \emph{Berkson:} We have $X = W + N, \; N\ind W, \; N\sim F_N$, so $\Q_{X,i}$ simply has density
        \begin{equation*}
            q_{X,i}(x) = f_N(x-w_i).
        \end{equation*}
        \item \emph{Classical:} We have $W = X + N, \; N\ind X, \; N\sim F_N$. We further assume the marginal distribution of $X\sim\P_X$ (with density $p_X$), which also need not equal the true $\P_X^0$. Then
        \begin{equation*}
            q_{X,i}(x) = \frac{f_N(w_i-x) p_X(x)}{\int_{\mathcal{X}}f_N(w_i-t) p_X(t) \dd t}.
        \end{equation*}
    \end{enumerate}
    \item Prior for $\P_{Y\mid X}^{\mathrm{prior}}$: the integral
   \(\P_{Y\mid X}^{\mathrm{prior}}(\dd y)=\int_{\Theta} \P_{g(X,\theta)}(\dd y)f(\theta)\,\dd\theta\)
   often lacks a closed-form for nonlinear $g(x, \theta)$. We approximate it by Monte Carlo.
\end{enumerate}
Hence, $\Q_{XY,i}$ can be expressed as
\[
 \Q_{XY,i} \;=\; 
 \underbrace{\Q_{X,i}}_{\text{prior for }X\text{ around }w_i\text{, depending on classical or Berkson ME}}
 \;\times\;
 \underbrace{\int \P_{g(X,\theta)} f(\theta)\,\dd\theta}_{\text{for }Y\mid X}.
\]
This ensures $\Q_{XY,i}$ is sufficiently rich to capture a broad range of $(X,Y)$ configurations. 

This prior construction also defines the corresponding marginal DP for $\P_{X\mid w_i}$:
\label{def:marginal-DP}
\begin{align}
  \text{Prior: } &\quad
  \P_{X\mid w_i} \sim \DP\Bigl(c,\,\Q_{X,i}\Bigr), \label{eq:marginal-prior} \\
  \text{Posterior: } &\quad
  \P_{X\mid w_i}\,\bigm|\,\bigl\{\tilde{x}_{i,j}\bigr\}_{j=1}^m
  \sim \DP\Biggl(
    c+m, \,
    \frac{c}{c+m}\,\Q_{X,i}
    +\frac{1}{c+m}\,\sum_{j=1}^m \delta_{\tilde{x}_{i,j}}
  \Biggr). \label{eq:marginal-posterior}
\end{align}    

The base measures in the marginal DPs \eqref{eq:marginal-prior} and \eqref{eq:marginal-posterior} are the $X$-marginals of the base measures in the joint DPs \eqref{eq:joint-prior} and \eqref{eq:joint-posterior}.

\subsection{Sampling latent covariates for NPL updates}
\label{sec:pseudo-sampling}
Before defining the pseudo-sampling scheme, it is useful to see why latent covariates cannot in general be generated from \(W\) alone and then paired with the observed responses, which is the approach taken by \citet{dellaporta2023robust}. Consider fixed-design Berkson ME with \(W_i\equiv 0\) and no outcome noise,
\[
X_i = W_i+\nu_i,\qquad \nu_i \simiid \mathcal N(0,\sigma_\nu^2),
\qquad
Y_i = X_i+\theta_0.
\]
Now suppose that latent covariates are sampled from \(\P^0_{X\mid W_i}\) alone, namely
\[
\tilde X_i \simiid \P^0_{X\mid W_i}=\mathcal N(0,\sigma_\nu^2),
\]
independently of \(Y_i\). If one then estimates \(\theta_0\) from the synthetic pairs \((\tilde X_i,Y_i)\) by least squares, one obtains
\[
\hat\theta :=\arg\min_{\theta\in\mathbb R}\sum_{i=1}^n (Y_i-\tilde X_i-\theta)^2
=\frac1n\sum_{i=1}^n (Y_i-\tilde X_i)
=\theta_0+\frac1n\sum_{i=1}^n (\nu_i-\tilde \nu_i),
\]
where \(\tilde \nu_i \simiid\mathcal N(0,\sigma_\nu^2)\) and \(\tilde \nu_i \ind \nu_i\). Hence \(\hat\theta-\theta_0\sim \mathcal N(0,2\sigma_\nu^2/n)\), and for any fixed \(M>0\),
\[
\Pr\bigl(|\hat\theta-\theta_0|>M\bigr)
=
2\Bigl\{1-\Phi\!\Bigl(\frac{M\sqrt n}{\sqrt2\,\sigma_\nu}\Bigr)\Bigr\}
\longrightarrow1
\qquad\text{as }\sigma_\nu/\sqrt n\to\infty,
\]
in particular as \(\sigma_\nu\to\infty\) for fixed \(n\). Thus, even with perfect knowledge of the measurement-error law, sampling latent covariates from \(W\) alone can produce synthetic pairs \((\tilde X_i,Y_i)\) that no longer represent the regression relation of interest. This is precisely the pathology our pseudo-sampling step is designed to avoid.

Motivated by this pathology, we require samples compatible with the joint law \(\P^0_{XY\mid w_i}\). Since \(X\) is unobserved, we employ a \emph{pseudo-sampling} procedure to acquire plausible realizations of the latent covariate \(X_i\) given observed data \((w_i,y_i)\).

Let \(\mathcal{L}(\theta,x_{1:n};w_{1:n},y_{1:n})\) denote the negative log-likelihood for the joint model of \((x_{1:n},y_{1:n},\theta)\), and let \(f(\theta)\) be a prior density for \(\theta\).
The resulting posterior over the parameters and latent variables is proportional to
\begin{equation}
\label{stationary-dist}
      P\bigl(x_{1:n},\,\theta \,\bigm|\,w_{1:n},y_{1:n}\bigr)
  \;\propto\;
  \exp\Bigl\{-\,\mathcal{L}\bigl(\theta,\,x_{1:n};\,w_{1:n},y_{1:n}\bigr)\Bigr\}
  f(\theta).
\end{equation}
In general, the conditional distribution of \( X_i \) given \( (w_i, y_i) \) is not available in closed form due to the nonlinear structure of the outcome model \( g(X,\theta) \). We therefore run Markov chain Monte Carlo (MCMC) schemes that target the joint posterior in \eqref{stationary-dist}. The DP update requires \emph{independent} draws from the posterior-predictive kernel
\begin{equation}
\label{eq:predictive-kernel}
    \Psi_n(\dd x\mid w_i,y_i)
:=
\int_{\Theta}\Pi\bigl(\dd x\mid\theta,w_i,y_i\bigr)\,
                \Pi_n\bigl(\dd\theta\mid w_{1:n},y_{1:n}\bigr),
\end{equation}
so we generate \(\tilde x_{i,j}\stackrel{\text{iid}}{\sim}\Psi_n(\cdot\mid w_i,y_i)\) via \emph{posterior predictive sampling}: draw \(\theta_{ij}\stackrel{\text{iid}}{\sim}\Pi_n(\cdot\mid w_{1:n},y_{1:n})\) and then \(\tilde x_{i,j}\sim\Pi(\cdot \mid\theta_{ij},w_i,y_i)\). Classic examples include: posterior predictive checks, which sample \(\theta\) and simulate replicated data for model checking \citep{gelman1996posterior}; and multiple imputation, which repeatedly draws missing or latent values from the posterior predictive conditional on \(\theta\) and combines analyses across imputations \citep{rubin1987multiple}.

As \(n \to \infty\), the posterior \(\Pi_n(\dd\theta \mid w_{1:n}, y_{1:n})\) concentrates around the parameter value minimizing the Kullback-Leibler (KL) divergence to the data-generating model, effectively transferring the best information the observed data carries on the model parameter into the pseudo-sampling procedure. Consequently, the independence requirement \( \theta_{ij}\stackrel{\text{iid}}{\sim}\Pi_n(\cdot \mid w_{1:n},y_{1:n}) \) is asymptotically immaterial: under contraction, the bias induced by the dependence among $\{\tilde x_{i,j}\}_{i=1,j=1}^{n,\;\;\,m}$ through the \(\theta\)-mixture in \eqref{eq:predictive-kernel} vanishes. For small \(n\), however, near-independent draws of \(\theta\) can improve exploration of a potentially dispersed posterior. A detailed theoretical assessment of this procedure is provided in Section~\ref{subsec:genbounds}, and we discuss dropping the conditional independence requirement in Appendix~\ref{appendix-independence-discussion}.

To implement the joint posterior sampling in \eqref{stationary-dist} we use HMC: its gradient-informed proposals typically yield high effective sample sizes with low autocorrelation, and independent chains parallelize naturally with modest additional cost. In theory, one could obtain independent posterior-predictive draws by running \(n\times m\) independent chains, but this is computationally onerous. In practice, we run a small number ($<10$) of well-mixed chains targeting \(\Pi_n(\cdot\mid w_{1:n},y_{1:n})\) and retain sparsely spaced post-burn-in states so that autocorrelation of the retained \(\theta\) is negligible. Convergence diagnostics and effective sample sizes are reported in Appendix~\ref{subsec:HMC-mixing-diagnostics}. A sensitivity analysis in Appendix~\ref{sec:HMC-sensitivity} shows that the resulting distributions of \(\{\tilde X_{ij}\}\) are empirically indistinguishable from those obtained by the theoretical construction of \(n\times m\) independent chains.

The pseudo-sampling scheme can be less desirable or infeasible in some settings. For example, when the scale of ME is small, acquiring information about the latent covariate $X_i$ from $(w_i,y_i)$ may not justify the extra computation. When the parameter space is high-dimensional, or the model is severely misspecified, reaching stationarity can be difficult, and the information about \(X_i\) contained in \((w_i,y_i)\) may be too weak to recover reliably. In such cases, we can set $m\equiv 1$ and replace the pseudo-samples $\tilde x_{ij}$ with $w_i$, so that the DP posteriors are updated by $(w_i,y_i)$. This update does not coincide with the true joint distribution $\P_{XY}^0$ and serves as a compromise. We call this the \emph{no-pseudo-sampling variant} of our framework. Section~\ref{subsec:genbounds} provides detailed theoretical assessments of this variant relative to the pseudo-sampling scheme.

Now we have the DP posterior realizations of our target distributions $\P^0_{XY\mid w_i}$ and $\P^0_{X\mid w_i}$, which we denote by $\P^{\DP}_{XY\mid w_i}$ (from \eqref{eq:joint-posterior}) and $\P^{\DP}_{X\mid w_i}$ (from \eqref{eq:marginal-posterior}), respectively. Collecting these DP posteriors over $i$ and representing $F_W$ as the empirical distribution of $W$, $F_W = n^{-1}\sum_{i=1}^n \delta_{w_i}$, we formally define the DP counterpart of our MMD target as
\begin{equation}
\label{eq:MMD-DP}
\hat{\theta}_n=\argmin_{\theta\in\Theta} \MMD_k\left(\frac{1}{n}\sum_{i=1}^{n}\P_{XY\mid w_i}^{\DP}, \Bigl(\frac{1}{n}\sum_{i=1}^{n}\P_{X\mid w_i}^{\DP}\Bigr) \P_{g(X, \theta)}\right).
\end{equation}
\subsection{Posterior bootstrap implementation}
\label{sec:posterior-boostrap}
For each bootstrap replicate \(b=1,\dots,B\) we independently draw random measures $\{\P_{XY\mid w_i}^{\DP,(b)}\}_{i=1}^n$ and $\{\P_{X\mid w_i}^{\DP,(b)}\}_{i=1}^n$ from \eqref{eq:joint-posterior} and \eqref{eq:marginal-posterior}, respectively. We then solve
\[
  \hat\theta_{n,b}=\argmin_{\theta\in\Theta} \MMD_k\left(\frac1n\sum_{i=1}^{n}\P_{XY\mid w_i}^{\DP, (b)},
\Bigl(\frac1n\sum_{i=1}^{n}\P_{X\mid w_i}^{\DP, (b)}\Bigr) \P_{g(X, \theta)}\right).
\]
The collection \(\{\hat\theta_{n,b}\}_{b=1}^{B}\) forms an empirical approximation to the posterior implied by our MMD loss, thereby placing the proposed method within the Bayesian NPL framework. Algorithm~\ref{alg:dp-boot} in Appendix~\ref{subsec:algo} demonstrates our bootstrapping procedure.

\section{Theoretical assessments}\label{sec:theoretical-assessments}
\subsection{Notations and assumptions}
This section studies the estimator $\hat\theta_n$ defined above in \eqref{eq:MMD-DP}. We provide two types of results. First, in Section~\ref{subsec:genbounds}, we derive generalization bounds for the excess risk under model misspecification. The bounds separate three contributions: statistical fluctuation, prior-data discrepancy, and pseudo-sample discrepancy, which are weighted by the DP centring parameter $c$ and the number of pseudo-samples per observation $m$. Second, we prove consistency in Section~\ref{subsec:consistency}: under regularity and identifiability conditions, $\hat\theta_n$ converges in probability to the minimizer of the limiting MMD, which coincides with the data-generating parameter under correct model specification. Proofs are supplied in Appendix~\ref{sec:additional_proofs}.

We now introduce the construction used in both results and define the notation. Let $\D=\{(W_i,Y_i)\}_{i=1}^n$ be the observed sample, drawn i.i.d.\ from $\P_{WY}^0$. Fix a pseudo-sample size $m\ge 1$. Given $\D$, define the posterior-predictive kernel $\Psi_n(\dd x\mid w,y):=\int_{\Theta}\Pi_\theta(\dd x\mid w,y)\Pi_n(\dd\theta\mid\D)$, where $\Pi_\theta(\dd x\mid w,y):=\Pi(\dd x\mid \theta, w, y)$ and $\Pi_n(\dd\theta\mid\D)$ is the (possibly misspecified) posterior for $\theta$. For each $i$, independently draw pseudo-samples \(\tilde X_{i1},\ldots,\tilde X_{im}\mid\D \simiid \Psi_n(\cdot\mid W_i,Y_i)\)
and collect them in $S:=\{\tilde X_{ij}\}_{i=1,j=1}^{n,\;\;\;m}$. They feed the DP update by supplying atoms for the posteriors of $(X,Y)\mid W_i$ and $X\mid W_i$. 


Let $\{\Q_{XY,i},\Q_{X,i}\}_{i=1}^n$ be prior centring measures. Given $(\D,S)$, draw independent realizations from the joint and marginal DP posteriors \eqref{eq:joint-posterior} and \eqref{eq:marginal-posterior}:
\begin{equation}
\label{DP-posterior-def}
    \begin{aligned}
        \P_{XY\mid W_i}^{\DP}&\sim \DP\left(c+m,\frac{c}{c+m}\Q_{XY,i}+\frac{1}{c+m}\sum_{k=1}^m \delta_{(\tilde X_{ik},Y_i)}\right), \\
\P_{X\mid W_i}^{\DP}&\sim \DP\left(c+m,\frac{c}{c+m}\Q_{X,i}+\frac{1}{c+m}\sum_{k=1}^m \delta_{\tilde X_{ik}}\right).
    \end{aligned}
\end{equation}
Finally, we define the DP-based MMD minimization target and the resulting estimator $\hat\theta_n$
\[
M_n(\theta):=\MMD_k\left(\frac1n\sum_{i=1}^{n}\P_{XY\mid W_i}^{\DP},
\Bigl(\frac1n\sum_{i=1}^{n}\P_{X\mid W_i}^{\DP}\Bigr) \P_{g(X, \theta)}\right),
\qquad
\hat\theta_n:=\arg\min_{\theta\in\Theta} M_n(\theta).
\]
To facilitate theoretical assessments, we define the following notation for the average prior measures and empirical measures for the pseudo-samples:
\[
\P^{\mathrm{prior}}_{XY}:=\frac1n\sum_{i=1}^n \Q_{XY,i},
\;
\P^{\mathrm{prior}}_{X}:=\frac1n\sum_{i=1}^n \Q_{X,i}; \;
\P^{\mathrm{pseudo}}_{XY} := \frac{1}{nm}\sum_{i=1}^n\sum_{j=1}^m \delta_{(\tilde X_{ij},Y_i)},
\;
\P^{\mathrm{pseudo}}_{X} := \frac{1}{nm}\sum_{i=1}^n\sum_{j=1}^m \delta_{\tilde X_{ij}}.
\]
We also define the average base measures in the DP posteriors \eqref{DP-posterior-def}:
\[
\P^{\mathrm{base}}_{XY}:=\frac{c}{c+m}\P^{\mathrm{prior}}_{XY} + \frac{m}{c+m} \P^{\mathrm{pseudo}}_{XY},
\quad
\P^{\mathrm{base}}_{X}:=\frac{c}{c+m}\P^{\mathrm{prior}}_{X} + \frac{m}{c+m} \P^{\mathrm{pseudo}}_{X}.
\]
All randomness is defined on the product law
\(
\Pr = \Pr_{\D,S}\,\Pr_{\DP\mid \D,S}
\),
where $\Pr_{\D,S}$ is the joint law of $(\D,S)$ under the DGP and the pseudo-sampling scheme, and $\Pr_{\DP\mid \D,S}$ is the conditional law of the DP realizations given $(\D,S)$. We write
\[
\E_{\D,S}[\cdot]:=\E_{\Pr_{\D,S}}[\cdot],\qquad
\E_{\DP}[\cdot\mid \D,S]:=\E_{\Pr_{\DP\mid \D,S}}[\cdot].
\]
Unless noted otherwise, expectations are taken with respect to $\Pr$. 

If the decision maker chooses not to perform pseudo-sampling as discussed at the end of Section~\ref{sec:pseudo-sampling}, they can set $m\equiv 1$ and replace the pseudo-samples by the observed covariates, which gives
$\P^{\mathrm{base}}_{XY,i}=(c+1)^{-1}\{c\,\Q_{XY,i}+\delta_{(W_i,Y_i)}\}$ and
$\P^{\mathrm{base}}_{X,i}=(c+1)^{-1}\{c\,\Q_{X,i}+\delta_{W_i}\}$.
The empirical laws corresponding to $\P^{\mathrm{pseudo}}_{XY}$ and $\P_X^{\mathrm{pseudo}}$ reduce to
$\hat{\P}^n_{WY}=n^{-1}\sum_{i=1}^n \delta_{(W_i,Y_i)}$ and
$\hat{\P}^n_{W}=n^{-1}\sum_{i=1}^n \delta_{W_i}$.
All other notation remains unchanged. This formulation is investigated in Section~\ref{sec:nopseudo}.

We impose two standing conditions that apply throughout the theory section.

\begin{enumerate}[label=G\arabic*]
\item\label{ass:G1}
The MMD kernel on $\mathcal{X}\times\mathcal{Y}$ is
$k\big((x,y),(x',y')\big)=k_X(x,x')\,k_Y(y,y')$, where $k_X,k_Y$ are measurable, positive-definite, and characteristic, with
$\sup_{x,x'}k_X(x,x')=\kappa_X<\infty$ and $\sup_{y,y'}k_Y(y,y')=\kappa_Y<\infty$.
Consequently $k$ is measurable, positive-definite, and characteristic on $\mathcal X\times\mathcal Y$ with $\kappa:=\sup k(\cdot, \cdot) = \kappa_X\kappa_Y<\infty$.

\item\label{ass:G2}
Let $\mathcal H_X,\mathcal H_Y$ be the RKHSs of $k_X,k_Y$. For each $\theta\in\Theta$, the conditional mean embedding
$\mu_\theta(x):=\int_{\mathcal Y} k_Y(y,\cdot)\,\P_{g(x, \theta)}(\dd y)\in\mathcal H_Y$
extends to a bounded linear operator $\mathcal C_\theta:\mathcal H_X\to\mathcal H_Y$ with
$\sup_{\theta\in\Theta}\|\mathcal C_\theta\|_{\mathrm{op}}=\Lambda<\infty$, where $\|\cdot\|_{\mathrm{op}}$ is the operator norm.
\end{enumerate}

\begin{remark}
\begin{enumerate}[label=(\alph*)]
\item  Common bounded, characteristic kernels such as the Gaussian, Laplace, Matérn and rational-quadratic kernels all satisfy Assumption~\ref{ass:G1}.

\item Assumption~\ref{ass:G2} is commonly assumed in robust regression methods involving kernel mean embeddings \citep{alquier2024universal, dellaporta2023robust}. We need it to apply Lemma 2 of \citet{alquier2024universal}: under Assumption~\ref{ass:G2}, we have
\[
\MMD_k\left(\P_X \P_{g(X,\theta)}, \P_X^{\prime} \P_{g(X,\theta)}\right)
\le \|\mathcal{C}_{\theta}\|_{\mathrm{op}}\MMD_{k_X^2}\left(\P_X,\P_X^{\prime}\right).
\]
Here $\MMD_{k_X^2}(\P_X,\P_X^{\prime})$ denotes the MMD computed with kernel $k_X^2:=k_X\otimes k_X$ on $\Hilb_{k_X^2}:=\Hilb_{k_X}\otimes\Hilb_{k_X}$: it compares the mean embeddings of $\P_X$ and $\P_X^{\prime}$ in the tensor product space $\Hilb_{k_X^2}$; equivalently, it is the MMD (with kernel $k_X^2$) between the pushforward laws of $(X,X)$ with $X\sim \P_X$ and of $(X^{\prime},X^{\prime})$ with  $X^\prime\sim \P_X^\prime$ under the diagonal map $x\mapsto(x,x)$. This lemma links the MMD between joint distributions and marginals, which is necessary when $\P_X$ is unknown or misspecified. We refer the reader to \citep{alquier2024universal} for a detailed discussion on the existence and boundedness of $\mathcal{C}_{\theta}$.
\end{enumerate}
\end{remark}

\subsection{Generalization bound: both settings}

We present the general version of the generalization error bound that applies to both Berkson and classical ME settings.
\label{subsec:genbounds}
\begin{theorem}[Generalization error bound]\label{thm:gen-bound-main} Under Assumptions~\ref{ass:G1}-\ref{ass:G2}:
        \begin{align*}
        &\E_{\D,S}\left[\E_{\DP} \left[\MMD_k\left(\P_{XY}^0, \P_X^0 \P_{g(X,\hat{\theta}_n)}\right) \mid \D,S\right] \right] - \inf_{\theta\in\Theta} \MMD_k\left(\P_{XY}^0, \P_X^0\P_{g(X, \theta)}\right)
        \\ &\le \underbrace{\frac{2(\sqrt{\kappa}+\kappa_X\Lambda)}{\sqrt{n(c+m+1)}}}_{\text{statistical fluctuation}} + \underbrace{\frac{2c}{c+m}\E_{\D,S}\left[\Lambda\MMD_{k_X^2}\left(\P_X^0,\P_X^{\mathrm{prior}}\right)+\MMD_{k}\left(\P_{XY}^0,\P^{\mathrm{prior}}_{XY} \right)\right]}_{\text{prior model discrepancy, vanishes as } c/m\to 0 }
        \\&\qquad \qquad \qquad \qquad   + \underbrace{\frac{2m}{c+m}\E_{\D,S}\left[\Lambda\MMD_{k_X^2}\left(\P_X^0,\P_X^{\mathrm{pseudo}}\right)+\MMD_k\left(\P_{XY}^0,\P^{\mathrm{pseudo}}_{XY}\right)\right]}_{\text{pseudo-sample distribution discrepancy}}.
    \end{align*}
\end{theorem}

The first term, $2(\sqrt{\kappa}+\kappa_X\Lambda)\{n(c+m+1)\}^{-1/2}$, arises from DP variations. 

The second term measures how far the chosen centring measures $(\Q_{XY,i},\Q_{X,i})$ are from the true DGP.  Its weight is $c/(c+m)$.  Thus we can borrow strength from a well-calibrated prior by taking a moderate or large value of $c$. Conversely, if historical information is unreliable, we can downweight it by letting $c/m\to0$, after which this term vanishes.

The final line quantifies the error introduced by the latent covariate pseudo-sampling scheme.  Its weight is $m/(c+m)$.  We will show in Theorem~\ref{thm:pseudo-classical} that, for any sequence $M_n \to \infty$,
\[
     \E_{\D,S}\left[\MMD_{k}\left(\P^{\mathrm{pseudo}}_{XY},\P_{XY}^0\right)\right]=O\left(\frac{M_{n}}{\sqrt n}+\sqrt{1-e^{-\KL_{\ast}}}+ r_{n}\right),
\]
with an analogous bound for the $X$-marginal.  The term
$\sqrt{1-e^{-\KL_{\ast}}}$, which will be formally defined in Theorem~\ref{thm:pseudo-classical}, measures the \emph{total misspecification}: it is zero when the working model is correct and otherwise gives the best error attainable under the chosen family. 
\begin{remark}
    Theorem~\ref{thm:pseudo-classical} yields a remainder term $r_n\to 0$.  Since the misspecified Bernstein-von Mises theorem \citep{kleijn2012bernstein} required for this bound is asymptotic, no general rate for $r_n$ is available. We therefore retain $r_n$ explicitly in the bound.
\end{remark}
If one adopts the no-pseudo-sampling variant (end of Section~\ref{sec:pseudo-sampling}) and $\P^{\mathrm{pseudo}}_{XY}$ reduces to $\hat{\P}^n_{WY}:=\sum_{i=1}^n n^{-1} \delta_{(W_i, Y_i)}$, we will show a replacement bound in Lemma~\ref{thm:nopseudo} that depends solely on the true ME distribution:
\begin{equation*}
    \E_{\D} \left[\MMD_k(\hat{\P}^n_{WY}, \P_{XY}^0)\right] = O\left(\frac{1}{\sqrt{n}}+\sqrt{\MMD_{k_X}(F_N^0,\delta_0)}\right).
\end{equation*}
These three pieces provide a clear strategy.  A decision-maker can start with a relatively informative prior (moderate $c$) to exploit domain knowledge. If such knowledge is unreliable, reducing $c$ or increasing the pseudo-sample size $m$ shifts weight to the empirical component.

\subsection{Pseudo-sampling bounds: classical ME}
\label{sec:pseudo-classical}
We now control the pseudo-sample discrepancy term in Theorem~\ref{thm:gen-bound-main} under \emph{classical} ME. Let the true DGP under classical ME be 
\begin{equation}
\label{classical:model-def}
    W=X+N,\; Y=g^0(X)+E,\; X\sim \P_{X}^{0}, \; N\sim F_{N}^{0}, \; E\sim F_{E}^{0},\; N,E\ind X,\; N\ind E. 
\end{equation}
The joint density of the observed \((W,Y)\) is $p^{0}_{WY}(w,y)=\int_{\mathcal X}p^{0}_{X}(x)f^{0}_{N}(w-x)f_E^0\bigl(y-g^0(x)\bigr){\dd}x$.

We work under a misspecified model
\[W=X+N,\; Y=g(X, \theta)+E,\; X\sim \P_{X}, \; N\sim F_{N}, \; E\sim F_{E},\; N,E\ind X,\; N\ind E, \] 
and \(g^0(\cdot)\notin\{g(\cdot,\theta):\theta\in\Theta\}\).
Here \(\Theta\subset\R^{p}\) is compact.  For each \(\theta\) the induced density of \((W,Y)\) is \(p^{\theta}_{WY}(w,y)=\int_{\mathcal X}p_{X}(x)f_{N}(w-x)f_{E}\bigl(y-g(x, \theta)\bigr){\dd}x\). Define the pseudo-true parameter by $\theta^{\ast}:=\arg\min_{\theta\in\Theta}\KL\bigl(p^{0}_{WY}\Vert p^{\theta}_{WY}\bigr)$. We make the following assumptions:
\begin{enumerate}[label=A\arabic*]
\item \label{A1} The family $\{p_{WY}^{\theta}:\theta\in\Theta\}$ satisfies the conditions of \citet[Theorem~3.1]{kleijn2012bernstein} around $\theta^\ast$, ensuring posterior convergence.
\item \label{A2}There exists a neighbourhood $\Theta_\rho$ of $\theta^\ast$ such that, for every $(w,y)$ and all $\theta_1,\theta_2\in\Theta_\rho$, we have
$\MMD_{k_X}\left(\Pi_{\theta_1}(\cdot\mid w,y),\Pi_{\theta_2}(\cdot\mid w,y)\right)
\le L(w,y)\,\|\theta_1-\theta_2\|$ with $\E_{(W,Y)\sim \P_{WY}^0}L(W,Y)<\infty$.
\end{enumerate}

\begin{remark}
Assumption~\ref{A1} requires Bernstein-von Mises-type assumptions for posterior contraction, which collect the local asymptotic normality, smoothness and integrability conditions of \citet[Theorem~3.1]{kleijn2012bernstein}. Assumption~\ref{A2} is a local stability condition that transfers this contraction to the posterior predictive. Appendix~\ref{sec:appendix-A1} relists sufficient conditions for Assumptions~\ref{A1}--\ref{A2} and provides practical scenarios where they hold.
\end{remark}

\begin{theorem}\label{thm:pseudo-classical}
For each \(x \in \mathcal X\) write \(p^{0}_{Y\mid x}(y):=f_{E}^{0}\bigl(y-g^0(x)\bigr)\) and \(p^{\theta^{\ast}}_{Y\mid x}(y):=f_{E}\bigl(y-g(x,\theta^{\ast})\bigr)\). Let
\[
\begin{aligned}
\KL_{X}&:=\KL\bigl(p_{X}^{0}\Vert p_{X}\bigr),&
\KL_{N}&:=\KL\bigl(f_{N}^{0}\Vert f_{N}\bigr),&
\KL_{E}&:=\E_{X\sim p_{X}^{0}}\KL\bigl(p^{0}_{Y\mid X}\Vert p^{\theta^{\ast}}_{Y\mid X}\bigr).
\end{aligned}
\]
Denote \(\KL_{\ast}:=\KL_{X}+\KL_{N}+\KL_{E}\) and recall that \(\kappa=\kappa_{X}\kappa_{Y}\).
Under Assumptions~\ref{ass:G1} and~\ref{A1}--\ref{A2}, for all \(n,m\ge 1\) and any sequence $M_n\to\infty$,
\[
   \E_{\D,S}\Bigl[\MMD_{k}\bigl(\P^{\mathrm{pseudo}}_{XY},\P_{XY}^0\bigr)\Bigr]
   \le\frac{4\sqrt{\kappa}}{\sqrt n}
      +\frac{M_{n}}{\sqrt n}
      +2\sqrt{\kappa}\sqrt{1-\exp(-\KL_{\ast})}
      +\sqrt{\kappa} r_{n},
\]
where $r_n$ depends on $M_n$ and satisfies \(0\le r_{n}\le1\) and \(r_{n}\to0\).
\end{theorem}
The pseudo-true parameter $\theta^{\ast}:=\arg\min_{\theta\in\Theta}\KL\!\bigl(p^{0}_{WY}\Vert p^{\theta}_{WY}\bigr)$ is the value around which the misspecified posterior $\Pi_n(\dd\theta\mid\D)$ concentrates by the misspecified Bernstein-von Mises theorem; it is the best estimator attainable from the information in the error-prone pairs $(W_i,Y_i)$ within the working family $\{p^{\theta}_{WY}\}$ in a Bayesian posterior. Theorem~\ref{thm:pseudo-classical} shows how this information is conveyed to the latent covariate distribution through the pseudo-sampling kernel $\Psi_n(\dd x\mid w,y)=\int \Pi_\theta(\dd x\mid w,y)\,\Pi_n(\dd\theta\mid\D)$: as $n$ increases, $\Psi_n(\cdot\mid w,y)$ tracks $\Pi_{\theta^\ast}(\cdot\mid w,y)$ and the resulting pseudo-sample joint distribution $\P^{\mathrm{pseudo}}_{XY}$ approaches the true joint distribution $\P_{XY}^0$ up to the term $\sqrt{1-\exp(-\KL_{\ast})}$ summarizing total misspecification. 

Our framework employs the ME-contaminated pairs $(W_i,Y_i)$ only to form a Bayesian predictive law for $X_i\mid(W_i,Y_i)$. It does not commit to the parametric model when estimating $\theta$. The pseudo-samples $\{\tilde X_{ij}\}$ propagate the information about $\theta$ learned from the data into the nonparametric stage by updating the DP priors. The final estimator is the $\theta$ that best aligns the DP-updated joint distribution with the model-implied joint under the MMD.

This construction is flexible: any (generalized) posterior for $\theta$ that contracts to a limiting pseudo-truth can be used inside $\Psi_n(\cdot\mid w,y)$ (e.g., MMD-Bayes \citep{cherief2020mmd} or $\alpha$-posteriors \citep{medina2022robustness}) and produce pseudo-samples $\{\tilde X_{ij}\}$ corresponding to the generalized posterior. Under Assumption~\ref{A2}, the pseudo-samples inherit the required stability, and extending the bounds to a generalized posterior amounts to verifying the analogue of Assumption~\ref{A1} (posterior contraction under misspecification). 

The quantity $\sqrt{1-\exp(-\KL_{\ast})}$ in the bound provides a summary of \emph{total} misspecification of the working model relative to the DGP, with $\KL_{\ast}=\KL_X+\KL_N+\KL_E=0$ if and only if the latent covariate law, the ME distribution, and the outcome model are correctly specified. A sharper alternative is to replace $\sqrt{1-e^{-\KL_{\ast}}}$ by
$\MMD_{k}\bigl(\Pi_{\theta^\ast}^{XY},\P_{XY}^0\bigr)$, where
\[
\Pi_{\theta^\ast}^{XY}(\dd x,\dd y)
=\int_{\mathcal W\times\mathcal Y}
       \Pi_{\theta^\ast}(\dd x\mid w,y)\,\delta_y(\dd y)\,
       p_{WY}^0(\dd w,\dd y).
\]
Although this alternative does not explicitly separate the contributions of the different sources of misspecification, it avoids relying on the KL divergence and does not collapse to a trivial bound even when some component KLs are infinite.

\begin{proposition}[Classical ME: Marginal--$X$]
\label{cor:pseudo-marginalX-classical}
Define the kernel
$k_{X}^{2}:=k_{X}\otimes k_{X}$ with RKHS $\Hilb_{k_{X}^{2}}:=\Hilb_{k_{X}}\otimes\Hilb_{k_{X}}$. Under Assumptions~\ref{ass:G1} and~\ref{A1}--\ref{A2}, for all $n,m\ge 1$ and every $M_{n}\to\infty$,
\begin{equation*}
      \E_{\D,S}
     \bigl[\MMD_{k_{X}^{2}}(\P_X^{\mathrm{pseudo}},\P_X^0)\bigr]
  \le
      \frac{4\kappa_{X}}{\sqrt n}
      + \frac{M_{n}}{\sqrt n}
      + 2\kappa_{X}\sqrt{1-\exp(-\KL_{\ast})}
      + \kappa_{X} r_{n},
\end{equation*}
where $r_n$ depends on $M_n$ and satisfies \(0\le r_{n}\le1\) and \(r_{n}\to0\).
\end{proposition}

\subsection{Pseudo-sampling bounds: Berkson ME}
\label{sec:pseudo-berkson}
Next, we demonstrate the \emph{Berkson} ME counterpart for Theorem~\ref{thm:pseudo-classical}. Let the true DGP be  
\begin{equation}
\label{berkson:model-def}
       X = W + N,\quad Y = g^0(X)+E, \quad N\sim F_{N}^{0}, \quad E\sim F_{E}^{0},\quad  N, E\ind W, \;N\ind E,
\end{equation}
where the covariates $W_{1},\dots ,W_{n}$ are i.i.d. draws from a design law assigned by experts $\P_{W}^{0}$ or treated as fixed design values with empirical law \(\P_{W}^{0} = \frac{1}{n} \sum_{i=1}^n \delta_{w_i}\), as is common in Berkson ME \citep{delaigle2006nonparametric}.

For $\theta\in\Theta\subset\R^{p}$ consider the (misspecified) model
\[
   X = W + N,\; Y = g(X, \theta)+E, \; N\sim F_{N}, \; E\sim F_{E},\;  N\ind W, \;E\ind(N,W).
\]
Then we can define the true and model-implied joint densities of $(W,Y)$ as
\[ p^{0}_{WY}(w,y)=\int_{\mathcal X}p^{0}_{W}(w)f_N^0(x-w)f_E^0\bigl(y-g^0(x)\bigr)\dd x, \; p^{\theta}_{WY}(w,y)=\int_{\mathcal X} p^{0}_{W}(w)f_N(x-w)f_E\bigl(y-g(x,\theta)\bigr)\dd x\]
Define the pseudo-true parameter
\(
   \theta^{\ast}:=\arg\min_{\theta\in\Theta}\KL(p^{0}_{WY}\Vert p^{\theta}_{WY}).
\)
We set
\[
   \KL_{N}:=\KL(f_{N}^{0}\Vert f_{N}),\;
   \KL_{E}:=\int_{\R}\int_{\mathcal X}
               p^{0}_{W}(w)
               f_{N}^{0}(x-w)
               \KL\bigl(p^{0}_{Y\mid x}\Vert p^{\theta^{\ast}}_{Y\mid x}\bigr){\dd}x{\dd}w, \; \KL_{\ast}:=\KL_{N}+\KL_{E}.
\]
\begin{proposition}[Berkson ME]\label{cor:pseudo-berkson}
Under Assumptions~\ref{ass:G1} and~\ref{A1}--\ref{A2},
for all $n,m\ge1$ and any sequence $M_{n}\to\infty$, 
\[
\E_{\D,S}\Bigl[\MMD_{k}\bigl(\P^{\mathrm{pseudo}}_{XY},\P_{XY}^0\bigr)\Bigr]
   \le\frac{4\sqrt{\kappa}}{\sqrt n}
   +\frac{M_{n}}{\sqrt n}
   +2\sqrt{\kappa}\sqrt{1-\exp(-\KL_{\ast})}
   +\sqrt{\kappa} r_{n},
\]
\[
\E_{\D,S}\Bigl[\MMD_{k_X^2}\bigl(\P^{\mathrm{pseudo}}_{X},\P_{X}^{0}\bigr)\Bigr]
   \le
   \frac{4\kappa_X}{\sqrt n}
   +\frac{M_{n}}{\sqrt n}
   +2\kappa_X\sqrt{1-\exp(-\KL_{\ast})}
   +\kappa_X r_{n},
\]
where $r_n$ depends on $M_n$ and satisfies \(0\le r_{n}\le1\) and \(r_{n}\to0\).
\end{proposition}
These bounds highlight a difference from the Robust--MEM method of \citet{dellaporta2023robust}. Their construction places a single DP prior on the conditional latent-covariate law \(\P_{X\mid W=w_i}\) and updates the DP with the error-contaminated observation \(w_i\). This results in the nonvanishing terms $\sqrt{\Var_{\mathcal{H}_{k_X}}\left(F_N^0\right)}$ and $\MMD_{k_X}\left(F_N^0, \delta_0\right)$ in their generalization bounds, which remain $O(1)$ and can approach the trivial upper bound for the MMD as the ME scale grows. By contrast, our bounds decompose the excess MMD risk into misspecification components.
$$
\begin{aligned}
\text {Robust--MEM}&: \, \mathcal{E}_n \lesssim n^{-1 / 2}+\frac{c}{c+1}\left(\Delta_{\text {prior }}+\sqrt{\Var_{\mathcal{H}_{k_X}}\left(F_N^0\right)}\right)+\frac{1}{c+1} \sqrt{\MMD_{k_X}\left(F_N^0, \delta_0\right)}, \\
\text {Ours}&: \,  \mathcal{E}_n \lesssim(n(c+m))^{-1 / 2}+\frac{c}{c+m} \Delta_{\text {prior }}+\frac{m}{c+m}\left\{n^{-1 / 2}+\sqrt{1-e^{-\KL_{\ast}}}+r_n\right\} .
\end{aligned}
$$

\subsection{Bounds without pseudo-sampling}
\label{sec:nopseudo}
When we replace the pseudo-samples $\{\tilde X_{ij}\}_{i=1,j=1}^{n,\;\;\; m}$ by the observed covariates $\{W_i\}_{i=1}^n$ with $m\equiv 1$, the following bounds replace the pseudo-sample discrepancy term in Theorem~\ref{thm:gen-bound-main}.
\begin{lemma}
\label{thm:nopseudo}
Let $(X,W,Y)$ be jointly distributed random variables on $\R^d \times\R^d \times \R$ satisfying the classical ME model \eqref{classical:model-def} or Berkson ME model \eqref{berkson:model-def}, with $N \sim F_N^0$. Define the empirical measures
\(
\hat{\P}^n_{W} := \frac{1}{n} \sum_{i=1}^n \delta_{W_i}. \;
\hat{\P}^n_{WY} := \frac{1}{n} \sum_{i=1}^n \delta_{(W_i, Y_i)}
\).

We assume that $k_X$ is translation-invariant, i.e. that there exists a positive-definite function $\psi$ on $\mathcal{X}$ such that \(k_X(x,x') = \psi(x-x')\), $\forall x,x'\in\mathcal{X}$. Then $\psi(0) = \kappa_X$ by positive-definiteness of $ k_X$. Under Assumption~\ref{ass:G1}, we have
\begin{align}
&\E_{\D} \left[\MMD_k(\P_{XY}^0, \hat{\P}^n_{WY})\right] \leq \sqrt{2}\kappa_Y^{1/2}\kappa_X^{1/4} \sqrt{\MMD_{k_X}(F_N^0,\delta_0)} + \frac{\sqrt{\kappa}}{\sqrt{n}},\label{eq:nopseudo-joint}
    \\&\E_{\D} \left[\MMD_{k_X^2}(\P_X^0, \hat{\P}^n_{W})\right] \leq \sqrt{2\kappa_X\MMD_{k_X^2}(F_N^0,\delta_0)} + \frac{\kappa_X}{\sqrt{n}}. \label{eq:nopseudo-marginal}
\end{align}
\end{lemma}

\subsection{Consistency}
We next establish consistency of $\hat\theta_n$ for $\theta^\dagger$, the unique minimizer of the limiting loss $M(\theta)$. Before specializing to our ME settings, we first establish a general consistency guarantee for MMD-based NPL procedures, a setting that is popular \citep[e.g.][]{dellaporta2022robust,fazeli2023semi, dellaporta2023robust} but for which asymptotic results for the NPL estimator are limited. The forms of the limiting measures $\P_{XY}^{\infty}$ and $\P_{X}^{\infty}$, which depend only on the DP base measures, are given explicitly in Proposition~\ref{prop:cond}.
\label{subsec:consistency}

\begin{theorem}
\label{thm:consistency-main}
    Assume that Assumptions~\ref{ass:G1}--\ref{ass:G2} hold. Furthermore, assume that there exist fixed probability measures $\P_{XY}^{\infty}$ and $\P_{X}^{\infty}$, such that 
    \begin{equation}
        \label{B4}
\MMD_{k}\bigl(\P^{\mathrm{base}}_{XY},\P_{XY}^{\infty}\bigr)\xrightarrow{\Pr}0, \qquad \MMD_{k_X^2}\bigl(\P^{\mathrm{base}}_{X},\P_{X}^{\infty}\bigr) \xrightarrow{\Pr}0 \qquad \text{as } n\to\infty.
    \end{equation} 
    Define $M(\theta):= \MMD_k(\P_{XY}^{\infty}, \P_{X}^{\infty}\P_{g(X, \theta)})$. If $\Theta$ is compact, $M(\theta)$ is continuous, and $M(\theta)$ attains its unique minimum at an interior point $\theta^\dagger$, then $ \hat\theta_n \overset{\Pr}{\rightarrow} \theta^\dagger.$
\end{theorem}

Condition~\eqref{B4} requires that the DP base measures concentrate to fixed limits. The next proposition establishes explicit forms of $\P_{XY}^{\infty}$ and $\P_X^{\infty}$ for pseudo-sampling or non-pseudo-sampling schemes under different $(c,m)$-parameter settings. We first state an assumption on the convergence of the constructed prior centring measures:
\begin{enumerate}[label=(D)]
    \item \label{prop:D} There exist fixed probability measures $\Q_{XY}^{\infty}$ and $\Q_{X}^{\infty}$ such that $\MMD_k(n^{-1}\sum_{i=1}^n \Q_{XY,i}, \Q_{XY}^{\infty})\overset{\Pr}{\rightarrow}0$ and $\MMD_{k_X^2}(n^{-1}\sum_{i=1}^n \Q_{X,i}, \Q_{X}^{\infty})\overset{\Pr}{\rightarrow}0$.
\end{enumerate}
This requirement is automatically met when each centring measure $\Q_{XY,i},\Q_{X,i}$ is a fixed, data-independent choice (for example, a historical distribution or an expert prior). When the centring measures depend on the data, e.g., around some naive estimator, the assumption still holds as long as the preliminary estimator contracts to a point estimate as $n\to\infty$.

\begin{proposition}[Sufficient conditions for~\eqref{B4}]
\label{prop:cond}
\mbox{}
    \begin{enumerate}[label=(\theproposition.\alph*)]
        \item \label{prop:1.1} If the DP base measures are constructed via the pseudo-sampling scheme and Assumptions~\ref{A1}--\ref{A2} are satisfied, recall that $ \theta^{\ast}=\arg\min_{\theta\in\Theta}\KL(p^{0}_{WY}\Vert p^{\theta}_{WY})$ and define
        \begin{equation*}
        \begin{aligned}
       \Pi_{\theta^\ast}^{XY} (\dd x, \dd y) := \int_{\mathcal W\times\mathcal{Y}}\Pi_{\theta^\ast}({\dd}x\mid w,y)\delta_y(\dd y)
     p^{0}_{WY}({\dd}w,{\dd}y), 
     \\
     \Pi_{\theta^\ast}^X(\dd x) := \int_{\mathcal W\times\mathcal{Y}}\Pi_{\theta^\ast}({\dd}x\mid w,y)
     p^{0}_{WY}({\dd}w,{\dd}y).  
        \end{aligned}
        \end{equation*}
        If $c,m$ are finite constants and \ref{prop:D} holds, then~\eqref{B4} holds with 
        \begin{equation}
        \label{eq:pseudo-with-prior}
            \P_{XY}^{\infty} = \frac{c}{c+m}\Q_{XY}^\infty+\frac{m}{c+m}\Pi_{\theta^\ast}^{XY}, \qquad  \P_{X}^{\infty} = \frac{c}{c+m}\Q_{X}^\infty+\frac{m}{c+m}\Pi_{\theta^\ast}^{X}.
        \end{equation}
        Otherwise, if $c/m\to 0$ as $n\to\infty$, then~\eqref{B4} holds with $ \P_{XY}^{\infty} = \Pi_{\theta^\ast}^{XY}, \;  \P_{X}^{\infty} =\Pi_{\theta^\ast}^{X}$. 
        \item \label{prop:1.2} If the DP base measures are constructed without pseudo-sampling, $c$ is a finite constant and \ref{prop:D} holds, then~\eqref{B4} holds with  
    \begin{equation}
        \label{eq:nopseudo-with-prior}
        \P_{XY}^{\infty} = \frac{c}{c+1}\Q_{XY}^\infty + \frac{1}{c+1}\P_{WY}^0,\qquad 
     \P_{X}^{\infty} = \frac{c}{c+1}\Q_{X}^\infty + \frac{1}{c+1}\P_W^0. 
        \end{equation}
    Otherwise, if $c\to 0$ as $n\to \infty$, then $ \P_{XY}^{\infty} = \P_{WY}^0, \;  \P_{X}^{\infty} = \P_W^0$.
    \end{enumerate}
\end{proposition}

The proposition shows that Assumption~\ref{prop:D} is only required when the prior weight ratio $c/m$ does not vanish. If the ratio $c/m\to 0$ with $n$, the prior contribution is asymptotically negligible and Assumption~\ref{prop:D} can be dropped.

\begin{remark}
\label{rem:joint-eta}
We have focused on estimating $\theta$, which parametrizes $g(\cdot)$. The construction extends to a joint parameter $\eta=(\theta,s)\in\Theta\times\mathcal S$, where $s$ parametrizes $f_E(\cdot\,;s)$. 

All statements in Sections~\ref{subsec:genbounds}-\ref{subsec:consistency} then hold with $(\hat\theta_n,\theta^\ast,\theta^\dagger)$ replaced by $(\hat\eta_n,\eta^\ast,\eta^\dagger)$, provided \ref{A1}-\ref{A2} are imposed on the joint family $\{p^{\eta}_{WY}:\eta\in\Theta\times\mathcal S\}$. 

\end{remark}

\section{Synthetic experiments}
\label{sec:Synthetic-Experiments}

We consider the sigmoid regression model
\(
g(x, \theta)=\theta_1/\left[1+\exp \big\{-\theta_2 (x-\theta_3)\big\}\right].
\)
It is a popular choice in synthetic experiments with nonlinear regression and is widely used in practical applications where one observes threshold behaviour \citep{yin2003flexible,klimstra2008sigmoid}. We study both Berkson and classical ME:
\[
\text{Berkson: } X=W+N,\qquad
\text{Classical: } W=X+N, \qquad F_N^0= \mathcal{N}(0,\sigma_{N}^2).
\]
To introduce misspecification, we contaminate the outcome noise using a Huber-type mixture:
\begin{equation}
\label{eq:Huber-contamination}
    F_E^0= (1-\varepsilon) \mathcal N(0,\sigma_E^2)+\varepsilon \mathcal N(0,\eta_E^2\sigma_E^2).
\end{equation}
Fig.~\ref{fig:synthetic-sigmoid-scatters} illustrates samples under both ME mechanisms with $\theta=(5, 1, 0.02)$, $\sigma_E=0.5$, $(\varepsilon, \eta_E)=(0.1, 9)$, and $\sigma_N=2$. We compare four estimators: \emph{NPL--HMC} (ours); \emph{Robust--MEM} \citep{dellaporta2023robust}, a robust approach for ME models that places a single DP prior on the conditional law of $\P_{X\mid w_i}$, updates this prior using the observed $\{w_i\}_{i=1}^n$ values, and performs parameter inference via a posterior bootstrap; \emph{NLS} (nonlinear least squares fitted to $(W,Y)$); and \emph{HMC} (posterior mean from HMC chains). Performance is summarized by the root mean squared error (RMSE) over 100 independent replications. Fig.~\ref{fig:rmse-synthetics} reports RMSE as the ME scale $\sigma_N$ increases under joint model and ME-distribution misspecification.
\begin{align*}
    &\mathrm{DGP}: \qquad N\sim\mathcal{N}(0,\sigma_N^2), \qquad E\sim (1-\varepsilon)\, \mathcal N(0,\sigma_E^2)+\varepsilon\, \mathcal N(0,\eta_E^2\sigma_E^2) \\
    &\mathrm{Model}: \qquad \qquad N\sim\mathcal{N}(0,\tau_N^2\sigma_N^2), \qquad E\sim\mathcal{N}(0,\tau_E^2\sigma_E^2)
\end{align*}
We fix $n = 300$, $(\varepsilon, \eta_E)=(0.1, 9)$, and $(\tau_N,\tau_E) = (0.7,2)$. In both Berkson and classical settings, NPL--HMC attains the lowest median RMSE and interquartile ranges for moderate-to-large ME ($\ge$ 2), while remaining comparable to Robust--MEM for small ME. At the largest ME scales considered, Robust--MEM can even underperform the HMC baseline, which demonstrates the ME-ignoring posterior update pathology anticipated by our theoretical comparison. The advantage of NPL--HMC widens as $\sigma_N$ increases: all competing methods degrade markedly, whereas the RMSE for NPL--HMC remains comparatively stable. Implementation and additional set-up details are in Appendix~\ref{app:exp-details}.


\begin{figure}[t]
  \centering
  \begin{subfigure}[b]{0.475\textwidth}
    \centering
    \includegraphics[width=\linewidth]{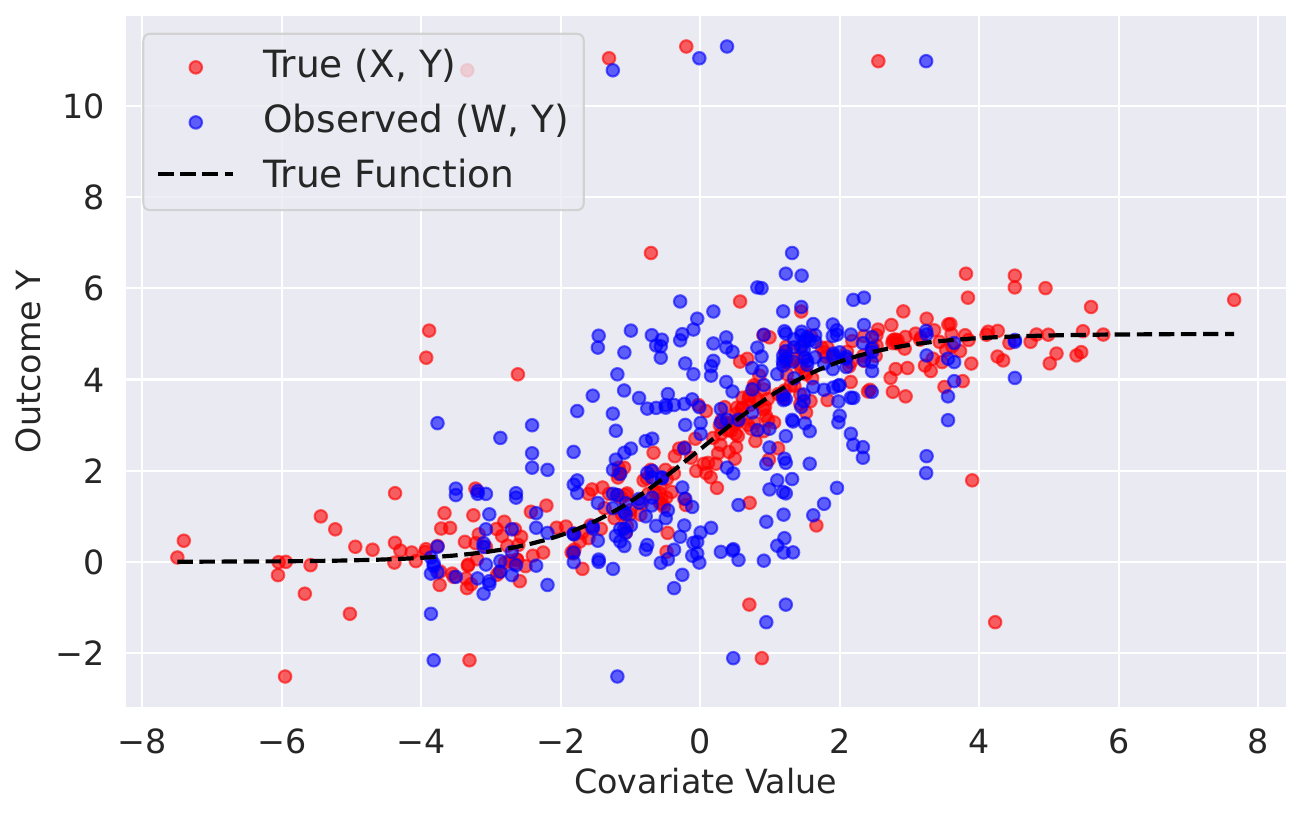}
    \caption{Berkson ME: $(X,Y)$ (red) and $(W,Y)$ (blue).}
  \end{subfigure}\hfill
  \begin{subfigure}[b]{0.49\textwidth}
    \centering
    \includegraphics[width=\linewidth]{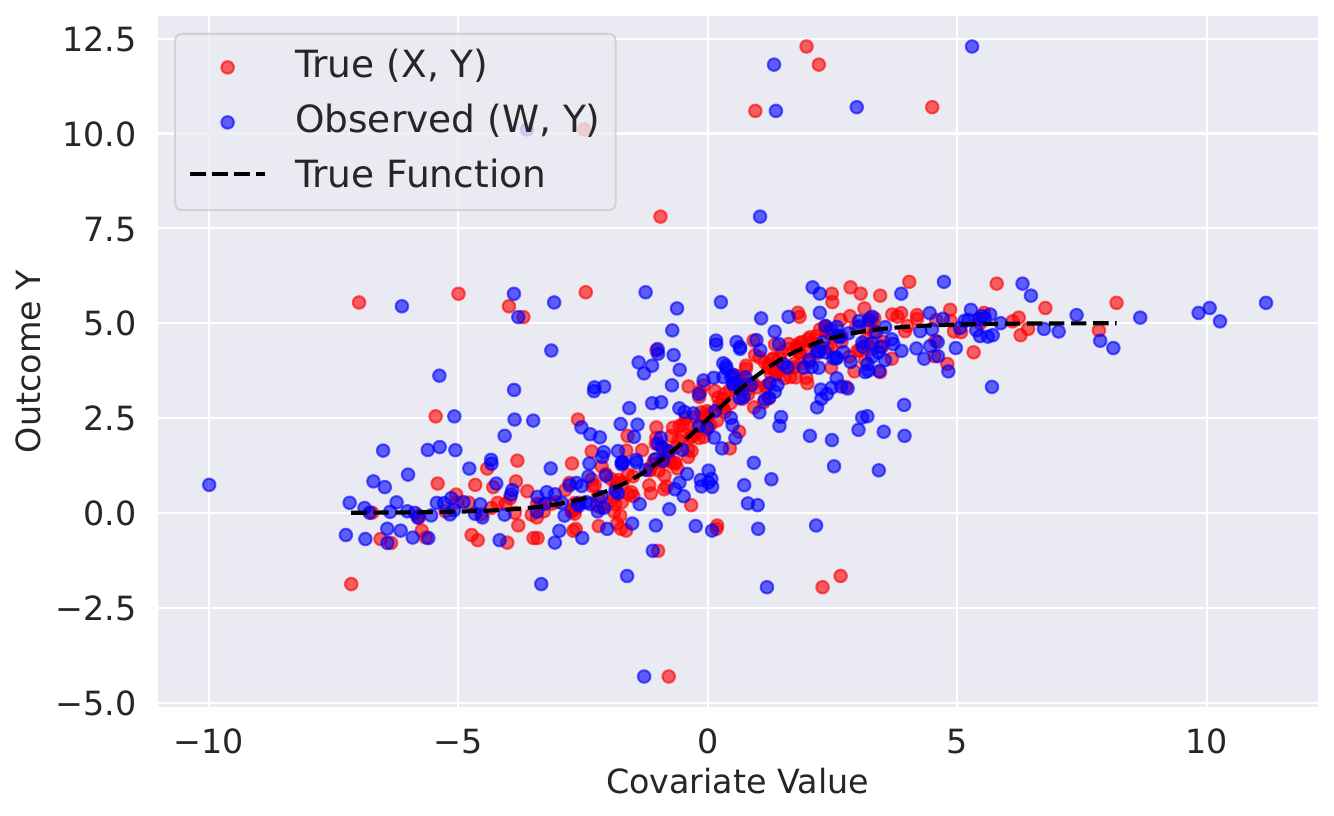}
    \caption{Classical ME: $(X,Y)$ (red) and $(W,Y)$ (blue).}
  \end{subfigure}
  \caption{Illustrative samples for the sigmoid model under ME and 10\% Huber contamination.}
  \label{fig:synthetic-sigmoid-scatters}
\end{figure}

\begin{figure}[t]
  \centering
  \begin{subfigure}[b]{0.48\textwidth}
    \centering
    \includegraphics[width=\linewidth]{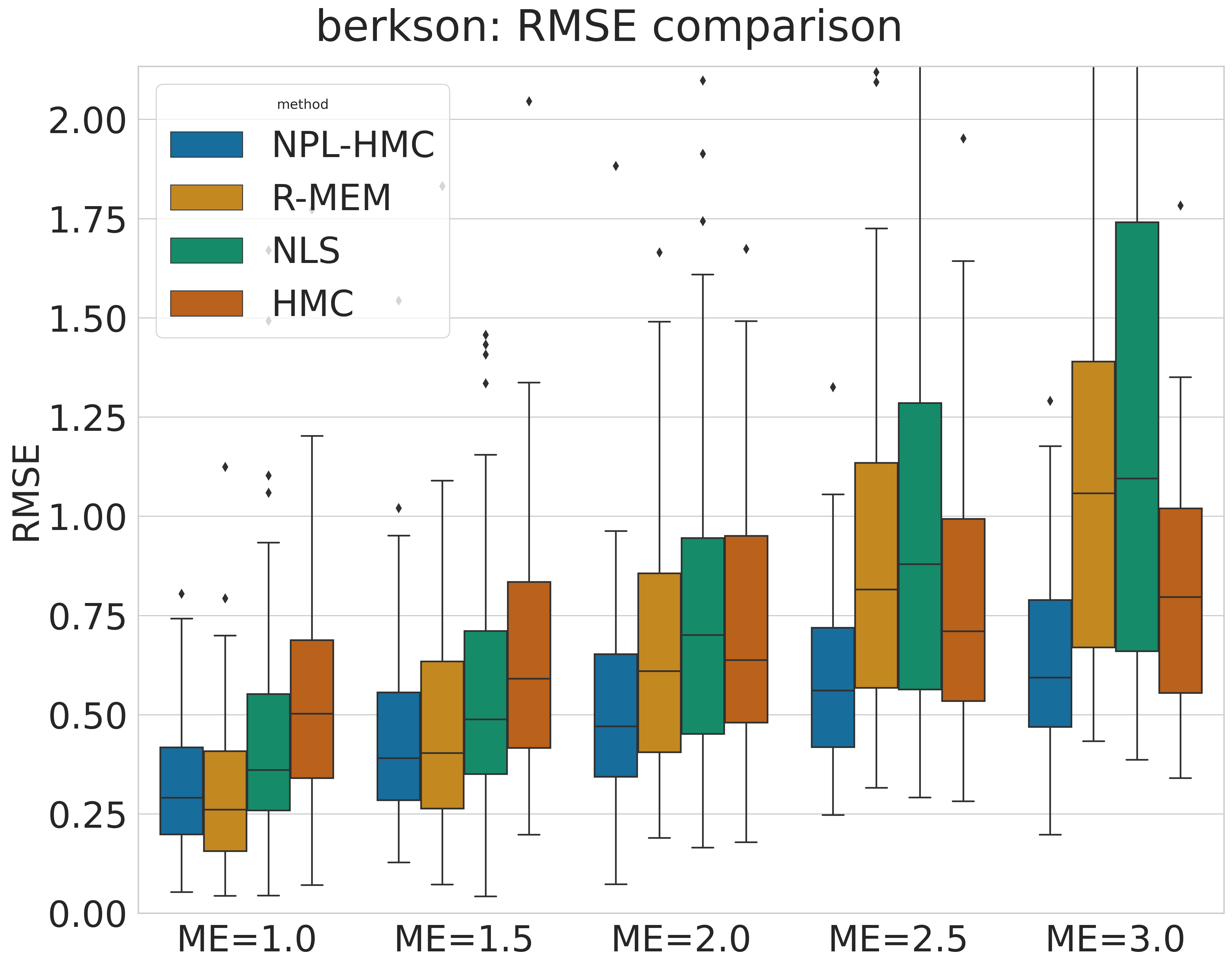}
    \caption{Berkson ME.}
    \label{fig:rmse-berkson}
  \end{subfigure}\hfill
  \begin{subfigure}[b]{0.48\textwidth}
    \centering
    \includegraphics[width=\linewidth]{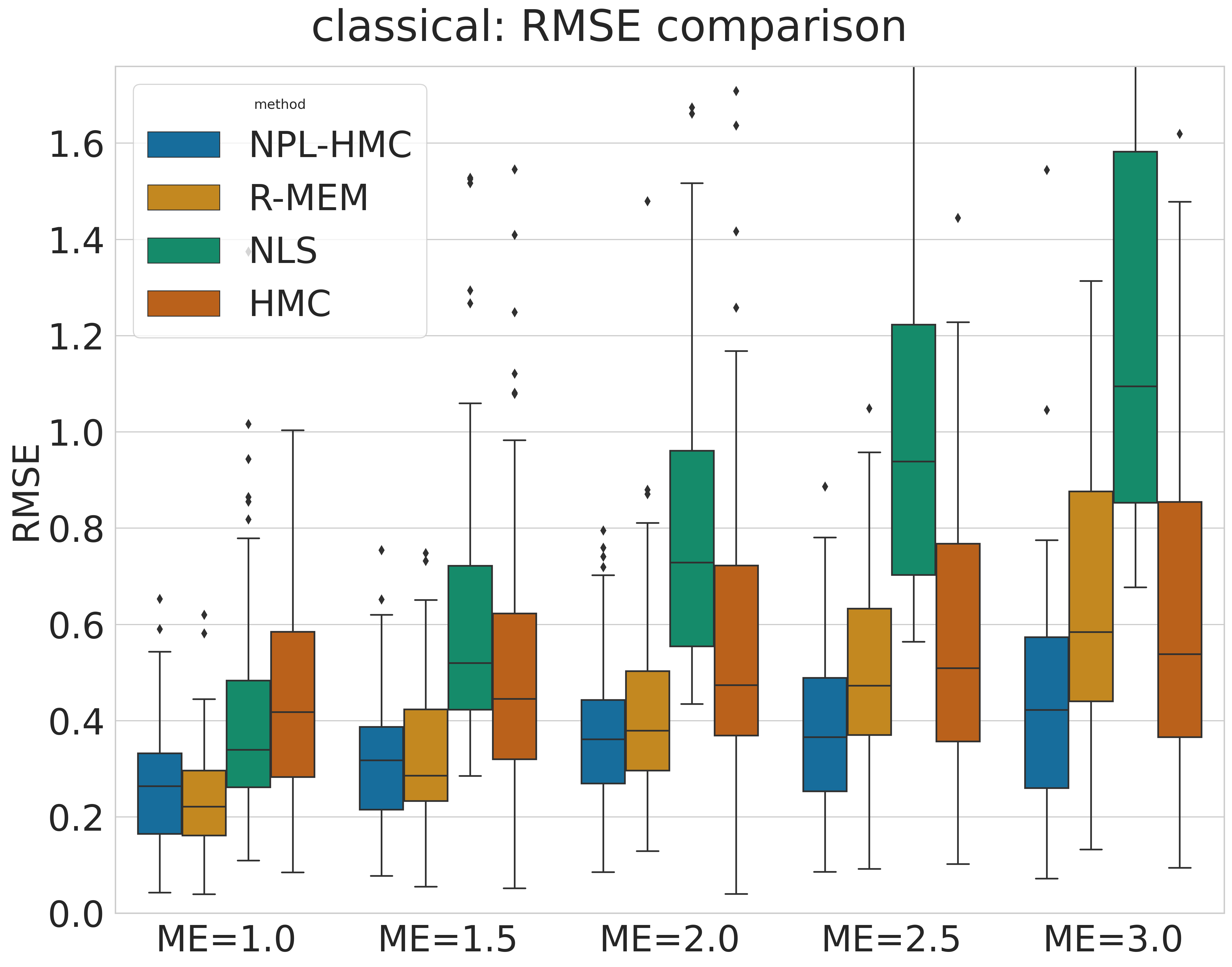}
    \caption{Classical ME.}
    \label{fig:rmse-classical}
  \end{subfigure}
  \caption{RMSE comparison for the sigmoid model under misspecification. ME denotes $\sigma_N$. Blue: NPL--HMC; yellow: Robust--MEM (shortened as R--MEM); green: NLS; orange: HMC.}
  \label{fig:rmse-synthetics}
\end{figure}

\section{Real-world experiments}
\label{sec:real-Experiments}
\subsection{Berkson ME: LIDAR range data}

We analyse a Berkson-type ME setting using the LIDAR data \citep{Sigrist1994AirMB} studied by \citet{ruppert2003variance}. The response is the log ratio \(Y_i\) and the covariate is range \(X_i\). The conditional variance \(\Var(Y\mid X)\) varies substantially with range and is not well represented as a function of \(\E(Y\mid X)\) \citep{ruppert2003variance}. To emulate a coarsened covariate measurement, we construct an observed regressor \(W_i\) by partitioning the empirical support of \(X\) into \(K_{\text{bins}}=20\) equal--width bins and replacing each \(X_i\) by the within--bin mean \(W_i\). This yields a Berkson decomposition \(X_i=W_i+\nu_i\), where the ME \(\nu_i=X_i-W_i\) represents within--bin variation. We fit the working model
\[
Y_i=g_\theta(X_i)+\varepsilon_i,\qquad 
g_\theta(x)=\theta_3+\frac{\theta_0}{1+\exp\{-\theta_1(x-\theta_2)\}},\qquad 
\varepsilon_i\sim\mathcal N(0,\sigma_\varepsilon^2).
\]
To assess robustness, we consider contamination ratios \(r_Y\in\{0.05, 0.10, 0.15, 0.20, 0.25\}\) and construct contaminated datasets by shifting a proportion \(r_Y\) of responses by \(6\sqrt{\hat v(X_i)}\), where \(\hat v(\cdot)\) is an estimated variance function for \(\Var(Y\mid X)\). We compare our NPL--HMC estimator with two baselines. Nonlinear least squares (NLS) fits \(g_\theta\) on \((W,Y)\); under Berkson ME it targets \(\E(Y\mid W)\) and does not recover \(\E(Y\mid X)\) in general. SIMEX \citep{cook1994simulation} is implemented for Berkson ME by adding synthetic normal noise to \(W\) at known multiples of the induced ME variance \(\Var(X-W)\) and extrapolating to the error--free limit.

Fig.~\ref{fig:lidar-fits} shows fitted curves. For smaller values of \(r_Y\), the fitted curves are similar; at \(r_Y=0.25\), NLS and SIMEX are more affected by the upward shifts, whereas NPL--HMC remains closer to the dashed oracle curve. Table~\ref{tab:lidar-rmse} reports (i) a dimensionless coefficient RMSE for \(\hat\theta\), compared with an oracle \(\theta^\star\) fitted using the true \(X\) while accounting for heteroscedasticity via an iterative variance-function procedure (smoothing the log squared residuals against \(X\) to estimate the variance function), as in \citet{ruppert2003variance}; and (ii) \(Y\)-RMSE computed using the true \(X\) and evaluated on uncontaminated \(Y\). The dimensionless \(\theta\)-RMSE uses componentwise scaling \(s_j=\max\{|\theta_j^\star|,0.01\,\mathrm{median}_k|\theta_k^\star|\}\) before forming the usual RMSE. 

\begin{figure}[t]
  \centering
  \includegraphics[width=\textwidth]{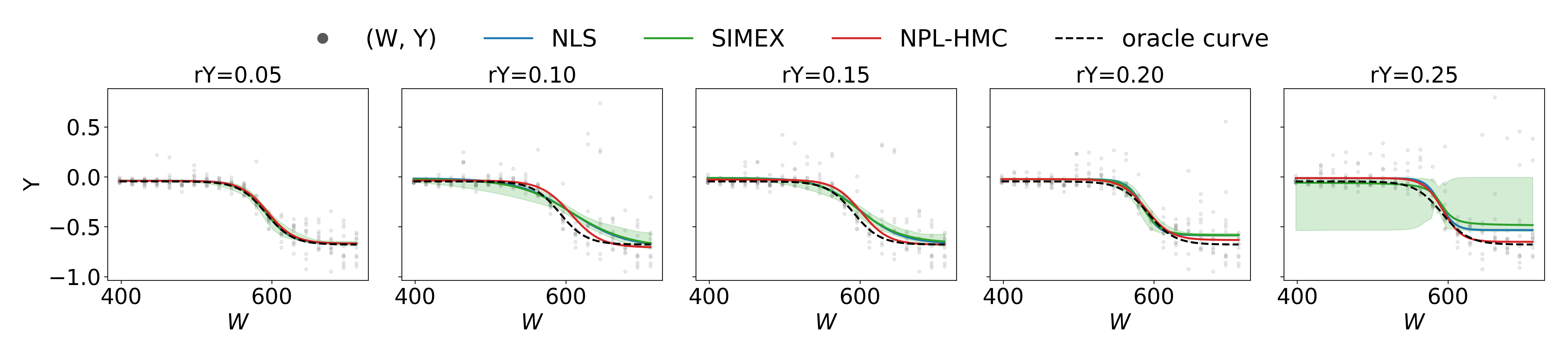}
  \caption{Estimated curves for the LIDAR data under a range of contamination ratios \(r_Y\), with \(K_{\text{bins}}=20\) in the Berkson construction. NLS (blue), SIMEX (green), NPL--HMC (red); 95\% bands are shaded; the dashed line is the oracle fit based on latent \(X\).}
  \label{fig:lidar-fits}
\end{figure}

\begin{table}[t]
\centering
\small
  \begin{tabular}{@{}lccc@{\hspace{1.2em}}ccc@{}}
    \toprule
    & \multicolumn{3}{c}{dimensionless $\theta$-RMSE} & \multicolumn{3}{c}{$Y$-RMSE} \\
    \cmidrule(lr){2-4} \cmidrule(lr){5-7}
    $r_Y$ & NLS & SIMEX & \textbf{NPL--HMC} & NLS & SIMEX & \textbf{NPL--HMC} \\
    \midrule
    0.05 & 0.0342 & 0.0852 (0.0577) & \textbf{0.0140} (0.0061) & \textbf{0.0825} & 0.0839 (0.0016) & 0.0834 (0.0002) \\
    0.10 & 0.2586 & 0.2747 (0.0689) & \textbf{0.0670} (0.0096) & 0.0992 & 0.1042 (0.0145) & \textbf{0.0918} (0.0006) \\
    0.15 & 0.1816 & 0.1851 (0.0708) & \textbf{0.0227} (0.0061) & 0.0963 & 0.0994 (0.0050) & \textbf{0.0920} (0.0005) \\
    0.20 & 0.3178 & 0.3309 (0.1872) & \textbf{0.0594} (0.0107) & 0.0975 & 0.0995 (0.0033) & \textbf{0.0899} (0.0003) \\
    0.25 & 0.4395 & 1.2994 (3.8130) & \textbf{0.2096} (0.0293) & 0.1143 & 0.1546 (0.0999) & \textbf{0.0878} (0.0006) \\
    \bottomrule
  \end{tabular}
  \caption{Dimensionless coefficient RMSE ($\theta$-RMSE) and prediction RMSE ($Y$-RMSE) by contamination ratio $r_Y$ for the LIDAR experiment with $K_{\text{bins}}=20$. For each $r_Y$, the dataset is fixed across methods; SIMEX and NPL--HMC entries report mean and standard deviation (SD), and NLS entries are point estimates.}
  \label{tab:lidar-rmse}
\end{table}

\subsection{Classical ME: Engel curves}

We analyse a log-quadratic Engel curve for food expenditure using the Belgian household data assembled by Engel in the 1850s (235 households; distributed as \texttt{statsmodels::engel} in the Python package by \citet{seabold2010statsmodels}). The outcome is food expenditure \(Y_i\) (francs) and the true covariate is household income \(X_i\) \citep{engel1857produktions}. The working model is
\(
Y_i=\theta_0+\theta_1\log X_i+\theta_2(\log X_i)^2+\varepsilon_i,
\; \varepsilon_i \ind \log X_i.
\)
Income is observed with error via self-reports \(W_i\). Motivated by evidence that misreporting is roughly proportional to income, we adopt a classical error model on the log scale,
\[
\log W_i=\log X_i+N_i,\qquad N_i\simiid\mathcal{N}(0,\sigma_N^2),
\]
with \(N_i\) independent of \((X_i,\varepsilon_i)\) \citep{kedir2003quadratic}. We set \(\log X_i\sim\mathcal{N}(\mu,\sigma_X^2)\) and index the ME magnitude by the \emph{error--variance ratio} on the log scale, \(\rho=\sigma_N^2/\sigma_X^2\).

In our application, the sample is small and the log-quadratic model is an approximation, so the error ratio cannot be identified precisely. Because Engel-curve elasticities inform economic and policy analysis, it is important that estimates are not overly sensitive to plausible choices of $\rho$. Estimated survey error varies across studies: \citet{bound2001measurement} and \citet{aasness1993engel} place $\rho$ in the range $[0.1,0.5]$, while \citet{hausman1995nonlinear} estimates $\rho$ values up to $0.72$ in a generalized setting with total expenditure as $X_i$ and budget shares as $Y_i$. We therefore examine $\rho \in [0,0.8]$, which covers the empirical ranges observed under different settings.

We first study how fitted curves change with \(\rho\). We compare our NPL--HMC estimator with SIMEX. Fig.~\ref{fig:reliability-selected} shows fitted curves with pointwise \(95\%\) bands for four different $\rho$ values. As \(\rho\) increases, the SIMEX curves move substantially, whereas their band widths remain roughly constant. By contrast, the NPL--HMC curves vary little with \(\rho\), and their bands widen as \(\rho\) grows. Thus NPL--HMC yields more stable point estimates and a cautious widening of uncertainty as the assumed error variance increases.

To summarize curve variation across \(\rho\), we use
\[
\widehat{S}=\Bigl\{\tfrac{1}{|\mathcal R|}\sum_{\rho\in\mathcal R}\sum_{x\in\mathcal G}
w(x)\bigl[g(x,\theta_\rho)-\bar g(x)\bigr]^2\Bigr\}^{1/2},
\]
the weighted deviation of fitted curves \(g(x,\theta_\rho)=\theta_{0,\rho}+\theta_{1,\rho}\log x+\theta_{2,\rho}(\log x)^2\) over a set $\mathcal{R}$ of $\rho$ values on a grid \(\mathcal G\), with weights \(w\) given by the empirical histogram of \(W\). We compute \(\widehat S\) with $\mathcal{R} = \{0, 0.1, \dots, 0.8\}$, which is smaller for NPL--HMC (0.017) than for SIMEX (0.028).

\begin{figure}[t]
  \centering
  \includegraphics[width=\textwidth]{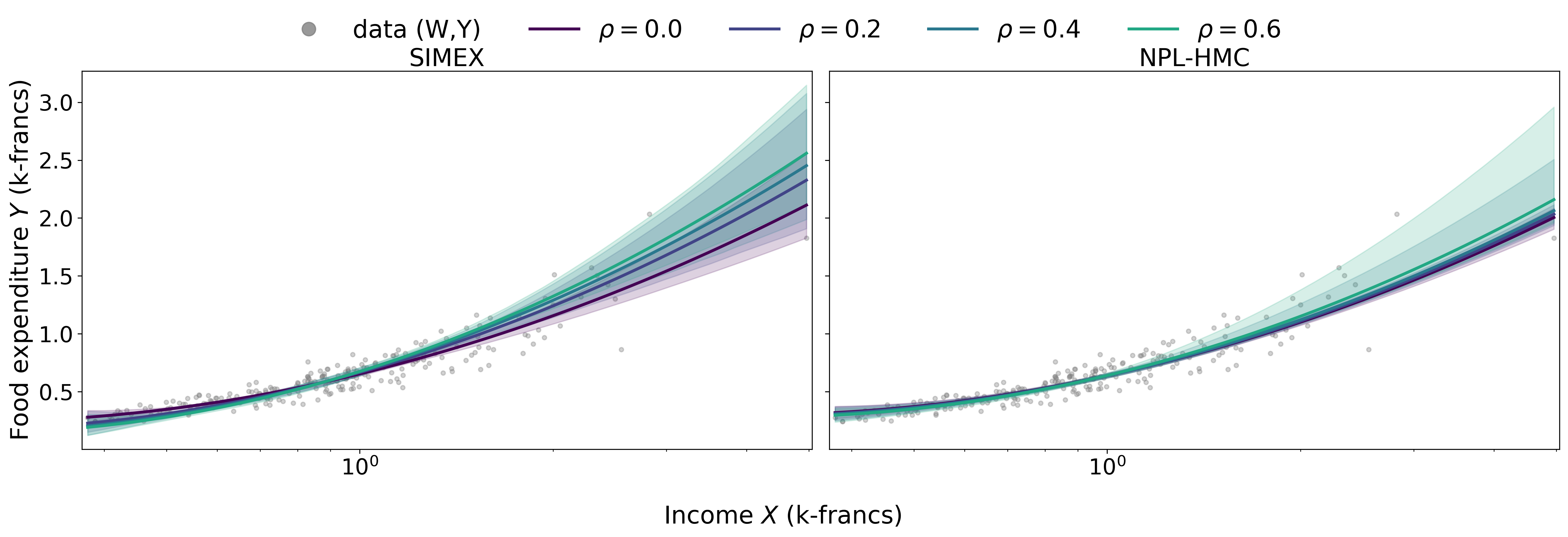}
  \caption{Fitted Engel curves for increasing values of $\rho$ with pointwise \(95\%\) bands for SIMEX (left) and NPL--HMC (right).}
  \label{fig:reliability-selected}
\end{figure}

We next assess sampling variability by repeated subsampling. For each \(\rho\in\mathcal{R}\), we draw \(M =100\) independent subsamples of size \(0.8n\) and re-estimate the curve. Table~\ref{tab:within-reliability-variance} reports, for each \(\rho\), the SD across subsamples of \((\widehat\theta_0,\widehat\theta_1,\widehat\theta_2)\). As a reference, we also report NLS, which ignores ME and is $\rho$-invariant: the SDs (multiplied by 100) are 0.53 for $\hat\theta_{0}$, 1.55 for $\hat\theta_{1}$, and 3.52 for $\hat\theta_{2}$. NPL--HMC has smaller SDs than SIMEX for almost all parameters and all \(\rho\), and its SDs are comparable to the NLS baselines. To summarize sensitivity of parameter estimation across different \(\rho\) values, Table~\ref{tab:across-rho-variance-meanwithinrep} reports the mean (over subsamples) of the variance of $\hat\theta_i$ across \(\rho\), defined as $\Var_\rho (\hat\theta_i) =\frac{1}{M} \sum_{m=1}^M \Var_\rho\left(\hat \theta_{i,\rho}^{(m)}\right)$, where $\hat \theta_{i,\rho}^{(m)}$ is the estimator for $\theta_i$ under ME ratio $\rho$ and for subsample index $m$. The across-\(\rho\) variance is smaller for NPL--HMC than for SIMEX for each parameter, with the largest difference in \(\hat\theta_1\), the \emph{income-elasticity index} used in Engel-curve applications.

\begin{table}[t]
\centering
\small

\begin{tabular}{lcccccc}
\toprule
& \multicolumn{3}{c}{\textbf{NPL--HMC}} & \multicolumn{3}{c}{\textbf{SIMEX}} \\
\cmidrule(lr){2-4}\cmidrule(lr){5-7}
$\rho$ & $\mathrm{SD}(\hat\theta_{0})$ & $\mathrm{SD}(\hat\theta_{1})$ & $\mathrm{SD}(\hat\theta_{2})$
       & $\mathrm{SD}(\hat\theta_{0})$ & $\mathrm{SD}(\hat\theta_{1})$ & $\mathrm{SD}(\hat\theta_{2})$ \\
\midrule
0.00  & 0.38  & 1.49  & 3.06  & 0.53  & 1.56  & 3.53 \\
0.10  & 0.47  & 1.57  & 3.10  & 0.87  & 2.10  & 5.17 \\
0.20  & 0.45  & 1.56  & 3.19  & 1.05  & 2.24  & 5.73 \\
0.30  & 0.54  & 1.55  & 3.19  & 1.00  & 2.60  & 6.39 \\
0.40  & 0.55  & 1.50  & 3.13  & 1.16  & 2.60  & 6.20 \\
0.50  & 0.53  & 1.91  & 3.29  & 1.18  & 2.75  & 6.63 \\
0.60  & 0.59  & 2.40  & 3.29  & 1.22  & 2.76  & 6.45 \\
0.70  & 0.57  & 2.82  & 3.67  & 1.29  & 2.76  & 6.22 \\
0.80  & 0.66  & 3.92  & 4.94  & 1.25  & 2.79  & 6.33 \\
\midrule
\textbf{Mean} & \textbf{0.53} & \textbf{2.08} & \textbf{3.43}
& \textbf{1.06} & \textbf{2.46} & \textbf{5.85} \\
\bottomrule
\end{tabular}
\caption{Within--$\rho$ SDs of subsample estimates ($\times 100$).}
\label{tab:within-reliability-variance}
\end{table}

\begin{table}[t]
\centering
\small
\begin{tabular}{lccc}
\toprule
Method & $\Var_\rho(\hat\theta_0)$ & $\Var_\rho(\hat\theta_1)$ & $\Var_\rho(\hat\theta_2)$ \\
\midrule
NPL--HMC & 0.20 & 15.46 & 9.11 \\
SIMEX    & 0.97  & 47.10 & 15.57 \\
\bottomrule
\end{tabular}
\caption{Across-\(\rho\) variance of parameter estimates (\(\times 10^{\,4}\)), reported as the mean (over subsamples) of the per-subsample variance across different \(\rho\).}
\label{tab:across-rho-variance-meanwithinrep}
\end{table}

\section{Discussion}
\label{sec:Conclusions_v2} 

Our framework opens promising avenues for
future research. For example, the pseudo-sampling step is generalizable: any (generalized) posterior for $\theta$ that contracts to a pseudo-true value can be used, and our risk decomposition still applies once the same contraction and stability conditions are verified. A limitation is that, in high-dimensional settings, the HMC step can mix slowly and be computationally expensive. Alternatives such as preconditioned/tempered or stochastic-gradient MCMC can be desirable, and variational approximations (e.g., mean-field or $\alpha$-variational inference \citep{blei2017variational, yang2020alpha}) may be used to produce pseudo-samples at the price of approximation bias. More generally, our framework is compatible with modern machine-learning components (e.g., Bayesian neural networks) as modelling and inference modules. Exploring these integrations is a natural direction for future work. Furthermore, a potential computational extension is to replace HMC pseudo-sampling with an amortized conditional sampler, so that latent draws can be generated at negligible marginal cost. Recent work on amortized generalized Bayes for simulator-based models using neural score-matching surrogates suggests one route to reducing or even removing the need for MCMC in such updates \citep{bharti2026amortised}.

Several important theoretical questions remain open. First, the distributional properties of $\hat\theta_n$ beyond consistency are unknown. In particular, it is open whether it satisfies asymptotic normality or a Bernstein-von Mises-type limit. Settling these would enable interval estimation and hypothesis testing. Second, fully nonparametric modelling of $g$ via Gaussian process priors under ME is attractive, but current approaches typically assume a known ME law and lack guarantees under misspecification. In parallel, it would be useful to develop adversarial robustness guarantees and (near-)minimax rates under $\varepsilon$-contamination and ME, with rates that degrade explicitly with the level of total model misspecification.



\section*{Acknowledgements}
MC is supported by the Warwick Statistics Centre for Doctoral Training and acknowledges funding from the University of Warwick. CD was supported by EPSRC grant [EP/Y022300/1]. TBB was supported by European Research Council Starting Grant 101163546. TD acknowledges support from a UKRI Turing AI acceleration Fellowship [EP/V02678X/1].


\bibliographystyle{apalike}
\bibliography{bibliography}

\pagebreak
\appendix
\section*{Supplementary Material}
Section~\ref{subsec:algo} presents the DP posterior bootstrapping algorithm. Section~\ref{sec:additional_proofs} contains proofs for results in the main text. Section~\ref{sec:appendix-A1} lists sufficient conditions for Assumptions~\ref{A1}--\ref{A2} and gives example scenarios where they hold. Section~\ref{app:exp-details} provides implementation and additional set-up details for synthetic and real-world experiments. Section~\ref{sec:HMC} includes diagnostics and sensitivity analysis for the HMC-based pseudo-sampling procedure. Furthermore, a detailed discussion of dropping the conditional independence requirement in the pseudo-sampling procedure (Section~\ref{sec:pseudo-sampling} in the main text) is included at the end.

\section{Algorithm}
\label{subsec:algo}
Below is our DP posterior bootstrapping algorithm as described in Section~\ref{sec:posterior-boostrap}. We use the truncated stick-breaking procedure \citep{sethuraman1994constructive} to approximate samples from the DP posterior:
\[
\P_{XY\mid w_i}^{\DP} \sim \DP\Bigl(c + m,\frac{c}{c + m}\Q_{XY,i}+\frac{1}{c + m}\sum_{j=1}^m \delta_{(\tilde{x}_{i,j},y_i)}\Bigr).
\]
Take
\[
  \{(x_{i,k}^{(\mathrm{prior})},y_{i,k}^{(\mathrm{prior})})\}_{k=1}^{T_{\DP}}\simiid\Q_{XY,i},\qquad \xi_{1:(T_{\DP}+m)}^{i} \sim \mathrm{Dir}\Bigl({\underbrace{\frac{c}{T_{\DP}},\dots,\frac{c}{T_{\DP}}}_{T_{\DP}\text{ terms}}}, {\underbrace{1,\dots,1}_{m \text{ terms}}} \Bigr).
  \]
Then
\[
  \P_{XY\mid w_i}^{\DP} \approx \sum_{k=1}^{T_{\DP}} \xi_k^i \delta_{\left(x_{i,k}^{(\mathrm{prior})},y_{i,k}^{(\mathrm{prior})} \right)} + \sum_{j=1}^m \xi_{T_{\DP}+j}^i \delta_{\left(\tilde x_{ij}, y_i\right)}.
  \]
Here $T_{\DP}$ is a finite truncation limit: in all experiments, we fix $T_{\DP} =100$.
{\singlespacing
\begin{algorithm}[ht]
\DontPrintSemicolon
\caption{DP posterior bootstrapping with truncation and MMD minimization}
\label{alg:dp-boot}
\KwIn{Pseudo-samples $\{\tilde{x}_{i,j}\}_{i=1,\dots,n;\,j=1,\dots,m}$; observed pairs $\{(w_i,y_i)\}_{i=1}^n$; DP concentration $c\ge0$; base measures $\{\Q_{XY,i}\}$; truncation $T_{\DP}$; bootstrap loops $B_{\mathrm{boot}}$.}
\KwOut{Posterior bootstrap draws $\{\hat\theta_{n,b}\}_{b=1}^{B_{\mathrm{boot}}}$.}
\For{$b\leftarrow 1$ \KwTo $B_{\mathrm{boot}}$}{
  \For{$i\leftarrow 1$ \KwTo $n$}{
    Draw $T_{\DP}$ atoms from the prior:
    $\{(x_{i,k}^{(\mathrm{prior})},y_{i,k}^{(\mathrm{prior})})\}_{k=1}^{T_{\DP}}\simiid\Q_{XY,i}$\;
    Sample weights
    $\xi_{1:(T_{\DP}+m)}^{(i,b)}\sim\mathrm{Dir}\bigl(\frac{c}{T_{\DP}},\dots,\frac{c}{T_{\DP}},1,\dots,1\bigr)$\;
    Construct
    $\P_{XY\mid w_i}^{(b)}:=\sum_{k=1}^{T_{\DP}}\xi_k^{(i,b)}\delta_{(x_{i,k}^{(\mathrm{prior})},y_{i,k}^{(\mathrm{prior})})}+\sum_{j=1}^m\xi_{T_{\DP}+j}^{(i,b)}\delta_{(\tilde{x}_{i,j},y_i)}$\;
    Construct
    $\P_{X\mid w_i}^{(b)}:=\sum_{k=1}^{T_{\DP}}\xi_k^{(i,b)}\delta_{x_{i,k}^{(\mathrm{prior})}}+\sum_{j=1}^m\xi_{T_{\DP}+j}^{(i,b)}\delta_{\tilde{x}_{i,j}}$\;
  }
  Set $\mathcal{P}^{\DP,(b)}_{XY}=\frac{1}{n}\sum_{i=1}^n\P_{XY\mid w_i}^{(b)}$\;
  Set $\widetilde{\mathcal{P}}^{(\theta,b)}_{XY}=\Bigl(\frac{1}{n}\sum_{i=1}^n\P_{X\mid w_i}^{(b)}\Bigr)\times\P_{g(\cdot,\theta)}$\;
  Compute $\hat\theta_{n,b}=\argmin_{\theta\in\Theta}\MMD_k\bigl(\mathcal{P}^{\DP,(b)}_{XY},\widetilde{\mathcal{P}}^{(\theta,b)}_{XY}\bigr)$\;
}
\end{algorithm}
}

\section{Proofs}\label{sec:additional_proofs}
\subsection{Proof of Theorem~\ref{thm:gen-bound-main}}
Before proving Theorem~\ref{thm:gen-bound-main}, we first state and prove the following lemma.
\begin{lemma}[Joint DP MMD Bound with concentration parameter \(c+m\)]
\label{lem:joint-dp-MMD}
Let \(\mathcal{Z} = \mathcal{X}\times\mathcal{Y}\) be a joint input-output space, and let \(k: \mathcal{Z}\times\mathcal{Z} \to \R\) be a positive-definite kernel with associated Hilbert space \(\Hk\). Suppose \(\P_{XY}^0\) is a target joint distribution on \(\mathcal{Z}\). For \(i=1,\dots,n\), let \(\P^i = \P_{XY\given w_i}^{\DP}\) be drawn from a Dirichlet Process
\[
  \DP(c+m, \P^{\mathrm{base}}_{XY,i}),
\]
where each base measure is
\[
  \P^{\mathrm{base}}_{XY,i} =\frac{c}{c+m}\Q_{XY,i} +\frac{1}{c+m}\sum_{k=1}^m \delta_{\bigl(\tilde{X}_{ik},Y_i\bigr)},
\]
and define the aggregated measure
\[
  \mathcal{P}_{XY}^{\DP} =\frac{1}{n}\sum_{i=1}^n \P^i, \qquad
  \P^{\mathrm{base}}_{XY} =\frac{1}{n}\sum_{i=1}^n \P^{\mathrm{base}}_{XY,i},
  \qquad
  \P^{\mathrm{prior}}_{XY}  =\frac{1}{n}\sum_{i=1}^n \Q_{XY,i}.
\]
Also define the pseudo-sample empirical distribution
\[
  \P^{\mathrm{pseudo}}_{XY}
  =\frac{1}{n}\sum_{i=1}^n \Bigl[\frac{1}{m}\sum_{k=1}^m \delta_{\bigl(\tilde{X}_{ik},Y_i\bigr)}\Bigr]
  =\frac{1}{nm}\sum_{i=1}^n \sum_{j=1}^m \delta_{\bigl(\tilde{X}_{ij},Y_i\bigr)}.
\]
Let $\mu(\D,S)$ be the joint distribution of $\P = (\P^1,\dots,\P^n)$ given the observed data $\D := \{(W_i, Y_i)\}_{i=1}^n$ and pseudo-samples \(S:=\{\tilde X_{ij}:1\le i\le n,1\le j\le m\}\). Then 
\begin{align*}
\E_{\D,S}\E_{\DP}\left[\MMD_k(\P_{XY}^0,\mathcal{P}_{XY}^{\DP})\mid \D,S\right]
   \le
   \frac{\sqrt{\kappa}}{\sqrt{n(c+m+1)}}
   +
   \frac{c}{c+m}\E_{\D,S}\MMD_k(\P_{XY}^0,\P^{\mathrm{prior}}_{XY} )
   \\+
   \frac{m}{c+m}\E_{\D,S}\left[\MMD_k(\P_{XY}^0,\P^{\mathrm{pseudo}}_{XY})\right].
\end{align*}

\end{lemma}
\begin{proof}[of Lemma~\ref{lem:joint-dp-MMD}]
\label{proof:joint-dp-MMD}
Throughout this proof, we condition on $(\D,S)$ and treat them as fixed, until the final step where we take the expectation over all $(\D,S)$. Each \(\P^i\) can be written via Sethuraman’s stick-breaking representation with concentration parameter \(\alpha = c+m\) \citep{sethuraman1994constructive}:
\[
  \P^i
  =
  \sum_{j=1}^{\infty}
    \xi_j^i\delta_{z_j^i},
\]
where \(\{\xi_j^i\}\sim\mathrm{GEM}(c+m)\), \(\{z_j^i\}\) are i.i.d.\ draws from the base measure \(\P^{\mathrm{base}}_{XY,i}\).
Hence
\[
  \mathcal{P}_{XY}^{\DP}
  =
  \frac{1}{n}\sum_{i=1}^n\sum_{j=1}^{\infty}\xi_j^i\delta_{z_j^i}.
\]
We use the triangle inequality:
\[
  \MMD_k\left(\P_{XY}^0,\mathcal{P}_{XY}^{\DP}\right)
  \le
  \MMD_k\left(\P_{XY}^0,\P^{\mathrm{base}}_{XY}\right)
  +
  \MMD_k\left(\P^{\mathrm{base}}_{XY},\mathcal{P}_{XY}^{\DP}\right),
\]
where
\[
  \P^{\mathrm{base}}_{XY}
  =
  \frac{1}{n}\sum_{i=1}^n \P^{\mathrm{base}}_{XY,i}
  =
  \frac{c}{c+m}\P^{\mathrm{prior}}_{XY} 
  +
  \frac{m}{c+m}\P^{\mathrm{pseudo}}_{XY}.
\]

We will show the following two bounds:

\begin{enumerate}
\item
\(\displaystyle
   \E_{\DP}\bigl[\MMD_k(\mathcal{P}_{XY}^{\DP},\P^{\mathrm{base}}_{XY})\mid \D,S\bigr] \le \frac{\sqrt{\kappa}}{\sqrt{n(c+m+1)}}
.\)

\item
\(\displaystyle
   \MMD_k\left(\P_{XY}^0,\P^{\mathrm{base}}_{XY}\right)
   \le
   \frac{c}{c+m}\MMD_k\left(\P_{XY}^0,\P^{\mathrm{prior}}_{XY} \right)
   +
   \frac{m}{c+m}\MMD_k\left(\P_{XY}^0,\P^{\mathrm{pseudo}}_{XY}\right).
\)
\end{enumerate}

We first bound \(\E_{\DP}\left[\MMD_k(\mathcal{P}_{XY}^{\DP},\P^{\mathrm{base}}_{XY})\mid \D,S\right]\).

We proceed in two steps: (A) compute the squared MMD, (B) apply Jensen’s inequality.

Rewrite
\[
  \MMD_k^2(\mathcal{P}_{XY}^{\DP},\P^{\mathrm{base}}_{XY})
  =
  \bigl\|\varphi(\mathcal{P}_{XY}^{\DP})-\varphi(\P^{\mathrm{base}}_{XY})\bigr\|_{\Hk}^2
  =
  \bigl\langle
    \varphi(\mathcal{P}_{XY}^{\DP})-\varphi(\P^{\mathrm{base}}_{XY}),    
    \varphi(\mathcal{P}_{XY}^{\DP})-\varphi(\P^{\mathrm{base}}_{XY})
  \bigr\rangle_{\Hk}.
\]
Express the embeddings:
\[
  \varphi(\mathcal{P}_{XY}^{\DP})
  =
  \int k(z,\cdot)d\mathcal{P}_{XY}^{\DP}(z)
  =
  \frac{1}{n}\sum_{i=1}^n \int k(z,\cdot)d\P^i(z),
\]
\[
  \varphi(\P^{\mathrm{base}}_{XY})
  =
  \int k(z,\cdot)d\P^{\mathrm{base}}_{XY}(z)
  =
  \frac{1}{n}\sum_{i=1}^n \int k(z,\cdot)d\P^{\mathrm{base}}_{XY,i}(z).
\]
Denote
\(\varphi(\P^i)=\int k(z,\cdot)d\P^i(z),
  \varphi(\P^{\mathrm{base}}_{XY,i})=\int k(z,\cdot)d\P^{\mathrm{base}}_{XY,i}(z).\)
Then
\[
  \varphi(\mathcal{P}_{XY}^{\DP})
  =
  \frac{1}{n}\sum_{i=1}^n \varphi(\P^i),
  \quad
  \varphi(\P^{\mathrm{base}}_{XY})
  =
  \frac{1}{n}\sum_{i=1}^n \varphi(\P^{\mathrm{base}}_{XY,i}).
\]
Hence
\[
  \varphi(\mathcal{P}_{XY}^{\DP}) - \varphi(\P^{\mathrm{base}}_{XY})
  =
  \frac{1}{n}\sum_{i=1}^n \left(\varphi(\P^i)-\varphi(\P^{\mathrm{base}}_{XY,i})\right).
\]
Therefore,
\begin{align*}
    \bigl\|\varphi(\mathcal{P}_{XY}^{\DP}) - \varphi(\P^{\mathrm{base}}_{XY})\bigr\|_{\Hk}^2
 &=
  \left\langle
    \frac{1}{n}\sum_{i=1}^n (\varphi(\P^i)-\varphi(\P^{\mathrm{base}}_{XY,i})),
    \frac{1}{n}\sum_{\ell=1}^n (\varphi(\P^\ell)-\varphi(\P^{\mathrm{base}}_{XY,\ell}))
  \right\rangle_{\Hk}
  \\&=\frac{1}{n^2} \sum_{i=1}^n \sum_{\ell=1}^n
    \Bigl\langle
      \varphi(\P^i)-\varphi(\P^{\mathrm{base}}_{XY,i}),
      \varphi(\P^\ell)-\varphi(\P^{\mathrm{base}}_{XY,\ell})
    \Bigr\rangle_{\Hk}.
\end{align*}  
We want its expectation w.r.t. \(\P=(\P^1,\dots,\P^n)\) under $\E_{\DP}[\cdot\mid\D,S]$. For brevity, we omit the conditioning on $\D,S$ for the rest of the proof and write $\E_{\DP}[\cdot] := \E_{\DP}[\cdot\mid\D,S]$. We have
\[
  \E_{\DP}\bigl[
    \MMD_k^2(\mathcal{P}_{XY}^{\DP},\P^{\mathrm{base}}_{XY})
  \bigr]
  =
  \frac{1}{n^2}\sum_{i,\ell=1}^n
  \E_{\DP}\Bigl[
    \bigl\langle
      \varphi(\P^i)-\varphi(\P^{\mathrm{base}}_{XY,i}),
      \varphi(\P^\ell)-\varphi(\P^{\mathrm{base}}_{XY,\ell})
    \bigr\rangle_{\Hk}
  \Bigr]
\]
Next, since \(\P^i\) and \(\P^\ell\) are drawn independently from the Dirichlet processes \(\DP(c+m,\P^{\mathrm{base}}_{XY,i})\) and \(\DP(c+m,\P^{\mathrm{base}}_{XY,\ell})\), so:
\begin{enumerate}
    \item If \(i\neq\ell\),  \[
          \E_{\DP}\bigl[
            \langle\varphi(\P^i)-\varphi(\P^{\mathrm{base}}_{XY,i}),\varphi(\P^\ell)-\varphi(\P^{\mathrm{base}}_{XY,\ell})\rangle
          \bigr]
          =
          \E_{\DP}\bigl[
            \varphi(\P^i)-\varphi(\P^{\mathrm{base}}_{XY,i})
          \bigr]
          \cdot
          \E_{\DP}\bigl[
            \varphi(\P^\ell)-\varphi(\P^{\mathrm{base}}_{XY,\ell})
          \bigr],
        \]
        since \(\P^i\) is independent from \(\P^\ell\) given $(\D,S)$. But \(\E[\varphi(\P^i)]=\varphi(\P^{\mathrm{base}}_{XY,i})\) by definition of the DP’s mean. Indeed, for any function \(f\),
        \(\E_{\P^i\sim\mu_i}\bigl[\int f(z)\mathrm{d}\P^i(z)\bigr] = \int f(z)\mathrm{d}\P^{\mathrm{base}}_{XY,i}(z)\). Hence
        \[
          \E_{\DP}[\varphi(\P^i)]=\varphi(\P^{\mathrm{base}}_{XY,i}).
        \]
        So
        \[
          \E_{\DP}\bigl[
            \varphi(\P^i)-\varphi(\P^{\mathrm{base}}_{XY,i})
          \bigr]
          =
          \varphi(\P^{\mathrm{base}}_{XY,i})-\varphi(\P^{\mathrm{base}}_{XY,i})
          =
          0.
        \]
        Thus any cross-term with \(i\neq\ell\) has expectation zero:
        \[
          \E_{\DP}\bigl[
            \langle\varphi(\P^i)-\varphi(\P^{\mathrm{base}}_{XY,i}),\varphi(\P^\ell)-\varphi(\P^{\mathrm{base}}_{XY,\ell})\rangle
          \bigr]
          = 0.
        \]
        \item If \(i=\ell\) we consider
  \(\E\bigl[
    \|\varphi(\P^i)-\varphi(\P^{\mathrm{base}}_{XY,i})\|^2_{\Hk}
  \bigr].\)
\end{enumerate}
By Sethuraman’s construction, we write $\P^i = \sum_{j=1}^{\infty} \xi_j^i \delta_{z_j^i}$, where $\{\xi_j^i\}_{j=1}^{\infty} \sim \mathrm{GEM}(\alpha)$, and $ z_j^i \simiid \P^{\mathrm{base}}_{XY,i}$.
Define the kernel mean embeddings
\[
  \varphi(\P^i) =
  \int k(z,\cdot)d\P^i(z)
  =
  \sum_{j=1}^{\infty} \xi_j^i k\left(z_j^i, \cdot\right),
\]
and
\[
  \varphi(\P^{\mathrm{base}}_{XY,i}) =
  \int k(z,\cdot)d\P^{\mathrm{base}}_{XY,i}(z).
\]
We aim to bound
\[
  \E_{\DP}\Bigl[
    \bigl\|\varphi(\P^i) - \varphi(\P^{\mathrm{base}}_{XY,i})\bigr\|_{\mathcal{H}_k}^2
  \Bigr].
\]
First write
\[
  \bigl\|\varphi(\P^i) - \varphi(\P^{\mathrm{base}}_{XY,i})\bigr\|_{\mathcal{H}_k}^2
  =
  \bigl\|\varphi(\P^i)\bigr\|^2_{\mathcal{H}_k}
  -
  2\bigl\langle
    \varphi(\P^i), \varphi(\P^{\mathrm{base}}_{XY,i})
  \bigr\rangle_{\mathcal{H}_k}
  +
  \bigl\|\varphi(\P^{\mathrm{base}}_{XY,i})\bigr\|^2_{\mathcal{H}_k}.
\]
Taking expectations over the random weights \(\{\xi_j^i\}\) and atoms \(\{z_j^i\}\) gives three terms:

\[
  \E_{\DP}\Bigl[\bigl\|\varphi(\P^i)\bigr\|^2_{\mathcal{H}_k}\Bigr],
  \quad
  \E_{\DP}\Bigl[\bigl\langle
    \varphi(\P^i), \varphi(\P^{\mathrm{base}}_{XY,i})
  \bigr\rangle_{\mathcal{H}_k}\Bigr],
  \quad
  \E_{\DP}\Bigl[\bigl\|\varphi(\P^{\mathrm{base}}_{XY,i})\bigr\|^2_{\mathcal{H}_k}\Bigr].
\]

We denote these three pieces as \(\mathrm{(I)}\), \(\mathrm{(II)}\), and \(\mathrm{(III)}\), respectively.

\begin{enumerate}
    \item \(\mathrm{(I)}\): \(\E_{\DP}[\|\varphi(\P^i)\|^2]\)
            Recall \(\varphi(\P^i) = \sum_{j=1}^\infty \xi_j^ik(z_j^i,\cdot)\). Then:
        \[
          \|\varphi(\P^i)\|^2_{\Hk}
          =
          \Bigl\langle
            \sum_{j=1}^\infty \xi_j^ik(z_j^i,\cdot),
            \sum_{t=1}^\infty \xi_t^ik(z_t^i,\cdot)
          \Bigr\rangle_{\Hk}
          =
          \sum_{j=1}^\infty \sum_{t=1}^\infty
            \xi_j^i\xi_t^i
            \langle k(z_j^i,\cdot),k(z_t^i,\cdot)\rangle_{\Hk}.
        \]
        By reproducing-kernel property, \(\langle k(a,\cdot), k(b,\cdot)\rangle = k(a,b)\). Hence
        \[
          \|\varphi(\P^i)\|^2_{\Hk}
          =
          \sum_{j,t}
            \xi_j^i \xi_t^i
            k\left(z_j^i,z_t^i\right).
        \]
        Taking expectation:
        \[
          \mathrm{(I)}
          =
          \E_{\DP}\Bigl[\|\varphi(\P^i)\|^2\Bigr]
          =
          \sum_{j,t}
            \E\bigl[\xi_j^i\xi_t^i\bigr]
            \E\bigl[k(z_j^i,z_t^i)\bigr],
        \]
        where we used independence of \(\xi_j^i\) from \(z_j^i\), plus the i.i.d.\ property across all \(j\). Define
        \[
          \Gamma(\P^{\mathrm{base}}_{XY,i}) := \E_{z,z'\sim \P^{\mathrm{base}}_{XY,i}}\bigl[k(z,z')\bigr], \quad \Psi(\P^{\mathrm{base}}_{XY,i}) := \E_{z\sim \P^{\mathrm{base}}_{XY,i}}[k(z,z)]
        \]
        Then 
        \begin{align*}
            2\Psi(\P^{\mathrm{base}}_{XY,i})-2\Gamma(\P^{\mathrm{base}}_{XY,i}) &=  \E_{z\sim \P^{\mathrm{base}}_{XY,i}}[k(z,z)] + \E_{z'\sim \P^{\mathrm{base}}_{XY,i}}[k(z,z)] - 2\E_{z,z'\sim \P^{\mathrm{base}}_{XY,i}}\bigl[k(z,z')\bigr]
            \\&= \E_{z,z'\sim \P^{\mathrm{base}}_{XY,i}}\bigl[k(z,z) + k(z',z')-2k(z,z')\bigr]
            \\&=\norm{\phi(z)-\phi(z')}_{\Hilb_k}^2
            \\&\ge 0,
        \end{align*}
        where $\phi$ is the feature map. This means $\Psi(\P^{\mathrm{base}}_{XY,i})\ge\Gamma(\P^{\mathrm{base}}_{XY,i})$. We separate diagonal (\(j=t\)) from off-diagonal:
        \[
          \mathrm{(I)}
          =
          \sum_{j=1}^\infty \E[(\xi_j^i)^2]
                              \E[k(z_j^i,z_j^i)]
          +
          \sum_{j\neq t}\E[\xi_j^i\xi_t^i]
                           \E[k(z_j^i,z_t^i)].
        \]
        Hence
        \[
          \mathrm{(I)}
          =
          \sum_{j=1}^\infty \E[(\xi_j^i)^2]\Psi(\P^{\mathrm{base}}_{XY,i})
          +
          \sum_{j\neq t}\E[\xi_j^i\xi_t^i]\Gamma(\P^{\mathrm{base}}_{XY,i}).
        \]
    \item Term \(\mathrm{(III)}\): \(\E_{\DP}[\|\varphi(\P^{\mathrm{base}}_{XY,i})\|^2]\)
    Here,
        \[
          \varphi(\P^{\mathrm{base}}_{XY,i})
          =
          \int k(z,\cdot)d\P^{\mathrm{base}}_{XY,i}(z),
          \quad
          \|\varphi(\P^{\mathrm{base}}_{XY,i})\|^2
          =
          \iint k(z,z')d\P^{\mathrm{base}}_{XY,i}(z)d\P^{\mathrm{base}}_{XY,i}(z').
        \]
        Since \(\varphi(\P^{\mathrm{base}}_{XY,i})\) is fixed given $(\D,S)$ \(\P^i\), we have 
        
        \[
          \mathrm{(III)} = \E_{\DP}\bigl[\|\varphi(\P^{\mathrm{base}}_{XY,i})\|^2\bigr]
          =
          \iint k(z,z')d\P^{\mathrm{base}}_{XY,i}(z)d\P^{\mathrm{base}}_{XY,i}(z').
        \]
        Define
        \[
          \Gamma(\P^{\mathrm{base}}_{XY,i})
          =
          \E_{z,z'\sim \P^{\mathrm{base}}_{XY,i}}[k(z,z')],
        \]
        thus
        \[
          \mathrm{(III)} = \Gamma(\P^{\mathrm{base}}_{XY,i}).
        \]
    \item Term \(\mathrm{(II)}\): \(\E_{\DP}[-2\langle \varphi(\P^i), \varphi(\P^{\mathrm{base}}_{XY,i})\rangle]\)
    We have
\[
  \langle \varphi(\P^i), \varphi(\P^{\mathrm{base}}_{XY,i})\rangle
  =
  \Bigl\langle
    \sum_{j=1}^\infty \xi_j^ik(z_j^i,\cdot),
    \int k(z,\cdot)d\P^{\mathrm{base}}_{XY,i}(z)
  \Bigr\rangle_{\Hk}
  =
  \sum_{j=1}^\infty \xi_j^i
   \int \bigl\langle k(z_j^i,\cdot), k(z,\cdot)\bigr\rangle_{\Hk}
         d\P^{\mathrm{base}}_{XY,i}(z).
\]
Again, \(\langle k(a,\cdot),k(b,\cdot)\rangle = k(a,b)\). So
\[
  \langle \varphi(\P^i), \varphi(\P^{\mathrm{base}}_{XY,i})\rangle
  =
  \sum_{j=1}^\infty
    \xi_j^i
    \int k(z_j^i,z)d\P^{\mathrm{base}}_{XY,i}(z).
\]
Taking expectation,
\[
  \mathrm{(II)}
  =
  -2\E_{\DP}\bigl[\langle \varphi(\P^i), \varphi(\P^{\mathrm{base}}_{XY,i})\rangle\bigr]
  =
  -2
  \sum_{j=1}^\infty
    \E[\xi_j^i]
    \E\Bigl[\int k(z_j^i,z)d\P^{\mathrm{base}}_{XY,i}(z)\Bigr].
\]
But 
\[
  \E\left[\int k(z_j^i,z)d\P^{\mathrm{base}}_{XY,i}(z)\right]
  =
  \E \left[\E_{z\sim \P^{\mathrm{base}}_{XY,i}}[k(z_j^i,z)]\right] = \E\left[\Gamma(\P^{\mathrm{base}}_{XY,i})\right] = \Gamma(\P^{\mathrm{base}}_{XY,i}),
\]
Hence
\[
  \mathrm{(II)}
  =
  -2 \Gamma(\P^{\mathrm{base}}_{XY,i})\sum_{j=1}^\infty \E[\xi_j^i].
\]
Now for a GEM(\(\alpha\)) distribution, we know that 
\(\sum_{j=1}^\infty \E[\xi_j^i] = 1\) by \citet[Lemma 2(b)]{sethuraman2024some}. So:

\[
  \mathrm{(II)}
  =
  -2 \Gamma(\P^{\mathrm{base}}_{XY,i})\cdot 1
  =
  -2\Gamma(\P^{\mathrm{base}}_{XY,i}).
\]
\end{enumerate}

Combining \(\mathrm{(I)}\), \(\mathrm{(II)}\), \(\mathrm{(III)}\), we get
\begin{align}
  \E_{\DP}\bigl[\|\varphi(\P^i)-\varphi(\P^{\mathrm{base}}_{XY,i})\|^2\bigr]
  &=\sum_{j=1}^\infty
    \E[(\xi_j^i)^2]\Psi(\P^{\mathrm{base}}_{XY,i})
    + \sum_{j\neq t}
      \E[\xi_j^i\xi_t^i]\Gamma(\P^{\mathrm{base}}_{XY,i})
  -
  2\Gamma(\P^{\mathrm{base}}_{XY,i}) 
  +
  \Gamma(\P^{\mathrm{base}}_{XY,i}).   \\
  &=
  \sum_{j=1}^\infty \E[(\xi_j^i)^2]\Psi(\P^{\mathrm{base}}_{XY,i})
  + \Gamma(\P^{\mathrm{base}}_{XY,i})\sum_{j\neq t} \E[\xi_j^i\xi_t^i]
  - \Gamma(\P^{\mathrm{base}}_{XY,i}).
\end{align}

Now we recall this identity from a GEM(\(c+m\)) distribution \citep[Lemma 3]{sethuraman2024some}:

\[
  \sum_{j=1}^\infty \E[(\xi_j^i)^2]
  +
  \sum_{j\neq t} \E[\xi_j^i\xi_t^i]
  =
  \frac{1}{c+m +1} + \frac{c+m}{c+m +1}
  =
  1.
\]

So we can write
\[
  \Gamma(\P^{\mathrm{base}}_{XY,i})\sum_{j\neq t} \E[\xi_j^i\xi_t^i]
  =
  \Gamma(\P^{\mathrm{base}}_{XY,i})\Bigl[
    1 - \sum_{j=1}^\infty \E[(\xi_j^i)^2]
  \Bigr].
\]
Thus:
\begin{align*}
  \E_{\DP}\bigl[\|\varphi(\P^i)-\varphi(\P^{\mathrm{base}}_{XY,i})\|^2\bigr]
  &=
  \sum_{j=1}^\infty \E[(\xi_j^i)^2]\Psi(\P^{\mathrm{base}}_{XY,i})
  +
  \Gamma(\P^{\mathrm{base}}_{XY,i})\Bigl[
    1 - \sum_{j=1}^\infty \E[(\xi_j^i)^2]
  \Bigr]
  -
  \Gamma(\P^{\mathrm{base}}_{XY,i}) \\
  &=
  \sum_{j=1}^\infty
    \E[(\xi_j^i)^2]\Psi(\P^{\mathrm{base}}_{XY,i})
  -
  \Gamma(\P^{\mathrm{base}}_{XY,i})\sum_{j=1}^\infty \E[(\xi_j^i)^2]\\
  &=
  \sum_{j=1}^\infty \E[(\xi_j^i)^2]
   \Bigl[\Psi(\P^{\mathrm{base}}_{XY,i}) - \Gamma(\P^{\mathrm{base}}_{XY,i})\Bigr] \\
   &\leq \kappa\sum_{j=1}^\infty \E[(\xi_j^i)^2]
\end{align*}
By \citet[Lemma 2(a), 2(b)]{sethuraman2024some}, we have the identity
\[
\sum_{j=1}^\infty \E[(\xi_j^i)^2] = \E[\xi_1^i] = \frac{1}{c+m+1}
\]
The last inequality follows from $\Psi(\P^{\mathrm{base}}_{XY,i}) - \Gamma(\P^{\mathrm{base}}_{XY,i}) \le\Psi(\P^{\mathrm{base}}_{XY,i}) = \E_{z\sim \P^{\mathrm{base}}_{XY,i}}[k(z,z)]\leq \kappa$.
Therefore,
\begin{align*}    \E_{\DP}\bigl[\MMD_k^2(\mathcal{P}_{XY}^{\DP},\P^{\mathrm{base}}_{XY})\bigr]
  &=\frac{1}{n^2}\sum_{i,\ell=1}^n\E_{\DP}\left[\bigl\langle\varphi(\P^i)-\varphi(\P^{\mathrm{base}}_{XY,i}),\varphi(\P^\ell)-\varphi(\P^{\mathrm{base}}_{XY,\ell})\bigr\rangle_{\Hk}\right]
  \\&=\frac{1}{n^2}\sum_{i=1}^n \E_{\DP}\left[\|\varphi(\P^i)-\varphi(\P^{\mathrm{base}}_{XY,i})\|^2\right]
  \\&\leq \frac{\kappa}{n^2}\sum_{i=1}^n \sum_{j=1}^\infty \E[(\xi_j^i)^2]
  \\&= \frac{\kappa}{n^2} \cdot n \cdot \frac{1}{c+m+1}
\end{align*}

By Jensen’s inequality:
\begin{equation}
\label{eq:MMD-term1}
      \E_{\DP}[\MMD_k(\mathcal{P}_{XY}^{\DP},\P^{\mathrm{base}}_{XY})]
  \le
  \sqrt{\E_{\DP}[\MMD_k^2(\mathcal{P}_{XY}^{\DP},\P^{\mathrm{base}}_{XY})]}
  \le
  \frac{\sqrt{\kappa}}{\sqrt{n(c+m+1)}}.
\end{equation}

We then bound \(\MMD_k(\P_{XY}^0,\P^{\mathrm{base}}_{XY})\).
We next look at
\[
  \P^{\mathrm{base}}_{XY,i}
  =
  \frac{c}{c+m}\Q_{XY,i}
  + \frac{1}{c+m}\sum_{k=1}^m \delta_{(\tilde{X}_{i,j},Y_i)},
\]
and
\[
  \P^{\mathrm{base}}_{XY}
  =
  \frac{c}{c+m}\P^{\mathrm{prior}}_{XY} 
  + \frac{m}{c+m}\P^{\mathrm{pseudo}}_{XY}.
\]
By definition:
\[
  \MMD_k(\P_{XY}^0, \P^{\mathrm{base}}_{XY})
  =
  \|\varphi\left(\P_{XY}^0\right) - \varphi(\P^{\mathrm{base}}_{XY})\|_{\Hk},
\]
where
\[
  \varphi(\P^{\mathrm{base}}_{XY})
  =
  \int k(z,\cdot)d\P^{\mathrm{base}}_{XY}(z)
  =
  \frac{c}{c+m}\varphi(\P^{\mathrm{prior}}_{XY} )
  +
  \frac{m}{c+m}\varphi\left(\P^{\mathrm{pseudo}}_{XY}\right).
\]
Hence
\begin{align*}
      \varphi\left(\P_{XY}^0\right)-\varphi(\P^{\mathrm{base}}_{XY})
  &=
  \varphi\left(\P_{XY}^0\right)
  - \Bigl[\frac{c}{c+m}\varphi(\P^{\mathrm{prior}}_{XY} ) + \frac{m}{c+m}\varphi\left(\P^{\mathrm{pseudo}}_{XY}\right)\Bigr]
  \\&=
  \frac{c}{c+m}\left(\varphi\left(\P_{XY}^0\right)-\varphi(\P^{\mathrm{prior}}_{XY} )\right)
  + \frac{m}{c+m}\left(\varphi\left(\P_{XY}^0\right)-\varphi\left(\P^{\mathrm{pseudo}}_{XY}\right)\right).
\end{align*}

Applying the triangle inequality:
\[
  \|\varphi\left(\P_{XY}^0\right)-\varphi(\P^{\mathrm{base}}_{XY})\|_{\Hk}
  \le
  \frac{c}{c+m}\|\varphi\left(\P_{XY}^0\right)-\varphi(\P^{\mathrm{prior}}_{XY} )\|_{\Hk}
  + \frac{m}{c+m}\|\varphi\left(\P_{XY}^0\right)-\varphi\left(\P^{\mathrm{pseudo}}_{XY}\right)\|_{\Hk}.
\]
Thus
\[
  \MMD_k(\P_{XY}^0,\P^{\mathrm{base}}_{XY})
  =\|\varphi\left(\P_{XY}^0\right)-\varphi(\P^{\mathrm{base}}_{XY})\|_{\Hk}
  \le
  \frac{c}{c+m}\MMD_k(\P_{XY}^0,\P^{\mathrm{prior}}_{XY} )
  +
  \frac{m}{c+m}\MMD_k(\P_{XY}^0,\P^{\mathrm{pseudo}}_{XY}).
\]
Combining this with \eqref{eq:MMD-term1}, we have, given $(\D,S)$
\begin{equation*}
\E_{\DP}\left[\MMD_k(\P_{XY}^0,\mathcal{P}_{XY}^{\DP})\right]
   \le
   \frac{\sqrt{\kappa}}{\sqrt{n(c+m+1)}}
   +
   \frac{c}{c+m}\MMD_k(\P_{XY}^0,\P^{\mathrm{prior}}_{XY} )+
   \frac{m}{c+m}\MMD_k(\P_{XY}^0,\P^{\mathrm{pseudo}}_{XY}).
\end{equation*}
Taking expectation over $(\D,S)$ completes the proof.
\end{proof}

By exactly the same argument as that in Lemma~\ref{lem:joint-dp-MMD} but applied to the marginal DP on $\P_{X\mid W}$ defined in Section~\ref{def:marginal-DP}, we have the following corollary.
\begin{corollary}
\label{cor:marginal-dp-MMD}
    \begin{align*}
                \E_{\D,S}\E_{\DP}\left[\MMD_{k_X^2}(\P_X^0,\mathcal{P}_{X}^{\DP})\mid \D,S\right]
   \le
   \frac{\kappa_X}{\sqrt{n(c+m+1)}}
   +
   \frac{c}{c+m}\E_{\D,S}\left[\MMD_{k_X^2}(\P_X^0,\P_X^{\mathrm{prior}})\right]
   \\+
   \frac{m}{c+m}\E_{\D,S}\left[\MMD_{k_X^2}(\P_X^0,\P^{\mathrm{pseudo}}_{X})\right].
    \end{align*}
\end{corollary}

We are now able to prove the theorem
\label{proof:gen-bound-main}
\begin{proof}[of Theorem~\ref{thm:gen-bound-main}]
We introduce the notation $\mathcal{P}^{\DP}_{XY}:= \frac{1}{n}\sum_{i=1}^n \P^{\DP}_{XY\mid w_i}$ and $\mathcal{P}^{\DP}_{X}:= \frac{1}{n}\sum_{i=1}^n \P^{\DP}_{X\mid w_i}$ for brevity. For any $\theta\in\Theta$, using the triangle inequality and the definition of $\hat{\theta}_n$, we have
    \begin{align*}
        &\MMD_k\left(\P_{XY}^0, \P_X^0 \P_{g(X,\hat{\theta}_n)}\right) 
        \\ &\le \MMD_k\left(\P_{XY}^0, \mathcal{P}_{XY}^{\DP} \right) + \MMD_k\left( \mathcal{P}_{XY}^{\DP}, \P_X^0 \P_{g(X,\hat{\theta}_n)} \right)
       \\ &\le \MMD_k\left(\P_{XY}^0, \mathcal{P}_{XY}^{\DP} \right) + \MMD_k\left( \mathcal{P}_{XY}^{\DP}, \mathcal{P}_{X}^{\DP} \P_{g(X,\hat{\theta}_n)} \right) + \MMD_k\left( \mathcal{P}_{X}^{\DP} \P_{g(X,\hat{\theta}_n)}, \P_X^0 \P_{g(X,\hat{\theta}_n)} \right) 
       \\ &\le \MMD_k\left(\P_{XY}^0, \mathcal{P}_{XY}^{\DP} \right) + \MMD_k\left( \mathcal{P}_{XY}^{\DP}, \mathcal{P}_{X}^{\DP} \P_{g(X, \theta)} \right) + \MMD_k\left( \mathcal{P}_{X}^{\DP} \P_{g(X,\hat{\theta}_n)}, \P_X^0 \P_{g(X,\hat{\theta}_n)} \right) 
       \\ &\le 2\MMD_k\left(\P_{XY}^0, \mathcal{P}_{XY}^{\DP} \right) +  \MMD_k\left( \P_{XY}^0, \mathcal{P}_{X}^{\DP} \P_{g(X, \theta)} \right) + \MMD_k\left( \mathcal{P}_{X}^{\DP} \P_{g(X,\hat{\theta}_n)}, \P_X^0 \P_{g(X,\hat{\theta}_n)} \right) 
       \\ &\le 2\MMD_k\left(\P_{XY}^0, \mathcal{P}_{XY}^{\DP} \right) +  \MMD_k\left( \P_{XY}^0, \P_X^0 \P_{g(X, \theta)} \right) \\&\quad + \MMD_k\left( \mathcal{P}_{X}^{\DP} \P_{g(X, \theta)}, \P_X^0 \P_{g(X, \theta)} \right) + \MMD_k\left( \mathcal{P}_{X}^{\DP} \P_{g(X,\hat{\theta}_n)}, \P_X^0 \P_{g(X,\hat{\theta}_n)} \right) 
    \end{align*}
    Taking expectations and taking the infimum over $\theta\in\Theta$, we have
    \begin{align*}
        &\E_{\D,S}\left[\E_{\DP} \left[\MMD_k\left(\P_{XY}^0, \P_X^0 \P_{g(X,\hat{\theta}_n)}\right) \mid \D,S\right] \right] - \inf_{\theta\in\Theta} \MMD_k\left(\P_{XY}^0, \P_X^0\P_{g(X, \theta)}\right)
        \\&\le 2\E_{\D,S}\left[\E_{\DP} \left[\MMD_k\left(\P_{XY}^0, \mathcal{P}_{XY}^{\DP} \right) \mid \D,S\right] \right] + 2\Lambda\E_{\D,S}\left[\E_{\DP} \left[\MMD_{k_X^2}\left(\P_X^0, \mathcal{P}_{X}^{\DP} \right)\mid \D,S \right] \right]
\\ &\le \frac{2(\sqrt{\kappa}+\kappa_X\Lambda)}{\sqrt{n(c+m+1)}} + \frac{2c}{c+m}\E_{\D,S}\left[\Lambda\MMD_{k_X^2}\left(\P_X^0,\P_X^{\mathrm{prior}}\right)+\MMD_{k}\left(\P_{XY}^0,\P^{\mathrm{prior}}_{XY} \right)\right] 
\\&\quad \quad \quad \quad   + \frac{2m}{c+m}\E_{\D,S}\left[\Lambda\MMD_{k_X^2}\left(\P_X^0,\P^{\mathrm{pseudo}}_{X}\right)+ \MMD_k\left(\P_{XY}^0,\P^{\mathrm{pseudo}}_{XY}\right)\right]
    \end{align*}
    The first inequality follows from \citet[Lemma 2]{alquier2024universal}, and the final line follows from Lemma~\ref{lem:joint-dp-MMD} and Corollary~\ref{cor:marginal-dp-MMD}. 
\end{proof}

\subsection{Proof of Theorem~\ref{thm:pseudo-classical}}
Before proving the theorem, we first state and prove three lemmas relating to the MMD:
\begin{lemma}[TVD to MMD]\label{lem:MMD-TVD}
Let $(\mathcal X,\mathcal A)$ be a measurable space and $P,Q$ two probability
measures on it.  Let $k:\mathcal X\times\mathcal X\to\R$ be a bounded,
positive-definite kernel with 
\[
K:=\sup_{x\in\mathcal X}k(x,x)<\infty.
\]
Denote by $\mathcal H_k$ the reproducing-kernel Hilbert space (RKHS)
associated with~$k$, let
\[
\MMD_k(P,Q):=\sup_{\|f\|_{\mathcal H_k}\le1}
        \Bigl|\int_{\mathcal X} f\dd (P-Q)\Bigr|,
\qquad
\|P-Q\|_{\TV}
       :=\frac12\int_{\mathcal X}\bigl|\dd P-\dd Q\bigr|
\]
be, respectively, the MMD and the (half-$L^{1}$-normalized)
total-variation distance.  Then
\[
      \MMD_k(P,Q)\le
      2\sqrt K\|P-Q\|_{\TV}.
\]
\end{lemma}

\begin{proof}
For $x\in\mathcal X$ write $\phi(x)\in\mathcal H_k$ for the canonical feature
map.  If
$f\in\mathcal H_k$, then by Cauchy-Schwarz
\[
|f(x)| =\bigl|\langle f,\phi(x)\rangle_{\mathcal H_k}\bigr|
           \le\|f\|_{\mathcal H_k}\|\phi(x)\|_{\mathcal H_k}
           =\|f\|_{\mathcal H_k}\sqrt{k(x,x)}
           \le\|f\|_{\mathcal H_k}\sqrt K.
\]
Hence
\[
\sup_{x\in\mathcal X}|f(x)|
      \le\sqrt K\|f\|_{\mathcal H_k}.
\]

We have the IPM representation of the TVD 
\[
\|P-Q\|_{\TV}
     =\frac{1}{2}\sup_{\|g\|_{\infty}\le1}
               \Bigl|\int_{\mathcal X} g\dd (P-Q)\Bigr|.
\]
For any $f\in\mathcal H_k$ with $\|f\|_{\mathcal H_k}\le1$,
we have $\|f/\sqrt K\|_{\infty}\le1$. Therefore,
\[
\Bigl|\int_{\mathcal X} f\dd (P-Q)\Bigr|
   =\sqrt K
          \Bigl|\int_{\mathcal X} \frac{f}{\sqrt{K}}\dd (P-Q)\Bigr|
   \le2\sqrt K\|P-Q\|_{\TV}.
\]

Since this holds for all $\|f\|_{\mathcal H_k}\le1$, taking the supremum over the unit ball gives
\[
\MMD_k(P,Q)
   =
   \sup_{\|f\|_{\mathcal H_k}\le1}
        \Bigl|\int_{\mathcal X} f\dd (P-Q)\Bigr|
   \le
   2\sqrt K\|P-Q\|_{\TV}.
\]
\end{proof}

\begin{lemma}[Joint--conditional MMD reduction]\label{lem:MMD-conditional}
Let
\[
   k_X:\mathcal X\times\mathcal X\longrightarrow[0,\kappa_X],
   \qquad
   k_Y:\mathcal Y\times\mathcal Y\longrightarrow[0,\kappa_Y],
   \qquad
   0<\kappa_X,\kappa_Y<\infty,
\]
be bounded, measurable, non-negative, positive-definite kernels.
Define the \emph{product kernel}
\[
   k\bigl((x,y),(x',y')\bigr)
      :=k_X(x,x')\,k_Y(y,y'),
   \qquad (x,y),(x',y')\in\mathcal X\times\mathcal Y.
\]
Fix a probability measure \(P_Y^{0}\) on \(\mathcal Y\).
For each \(y\in\mathcal Y\) let
\(P_{X\mid y}\) and \(Q_{X\mid y}\) be probability measures on
\(\mathcal X\), measurable in \(y\).
Form the \emph{joint} laws
\[
   \tilde P(\dd x,\dd y):=P_Y^{0}(\dd y)\,P_{X\mid y}(\dd x),\qquad
   \tilde Q(\dd x,\dd y):=P_Y^{0}(\dd y)\,Q_{X\mid y}(\dd x).
\]

Then
\begin{equation}\label{eq:MainIneq}
   \MMD_{k}\bigl(\tilde P,\tilde Q\bigr)
   \;\le\;
   \sqrt{\kappa_Y}\;
   \E_{Y\sim P_Y^{0}}
        \Bigl[\MMD_{k_X}\bigl(P_{X\mid Y},Q_{X\mid Y}\bigr)\Bigr].
\end{equation}
\end{lemma}

\begin{proof}
Throughout the proof we use this identity of the MMD:
\begin{equation}\label{eq:MMDsq}
   \MMD_{k}^{2}(\mu,\nu)
   =
   \E_{X,X'\sim\mu}k(X,X')
   +\E_{Y,Y'\sim\nu}k(Y,Y')
   -2\,\E_{X\sim\mu,\,Y\sim\nu}k(X,Y).
\end{equation}

Apply~\eqref{eq:MMDsq} with \(\mu=\tilde P\) and
\(\nu=\tilde Q\):
\begin{align}
   \MMD_{k}^{2}(\tilde P,\tilde Q)
   &=
   \underbrace{\E_{(X,Y),(X',Y')\sim\tilde P}K}_{(A)}
   +\underbrace{\E_{(X,Y),(X',Y')\sim\tilde Q}K}_{(B)}
   -2\,\underbrace{\E_{(X,Y)\sim\tilde P,\,(X',Y')\sim\tilde Q}K}_{(C)},
   \label{eq:ThreeTerms}
\end{align}
where we recall that $ k\bigl((x,y),(x',y')\bigr)=k_X(x,x')\,k_Y(y,y')$ is the product kernel. Because \(K\le\kappa_X\kappa_Y<\infty\), all integrands are bounded, and Fubini’s lemma allows us to change integration order freely.
\begin{align}
   (A)
   &=
   \iint_{\mathcal Y^2} P_Y^{0}(\dd y)\,P_Y^{0}(\dd y')
   \iint_{\mathcal X^2} P_{X\mid y}(\dd x)\,P_{X\mid y'}(\dd x')\;
         k_X(x,x')\,k_Y(y,y').
   \label{eq:TermA}
\end{align}
\begin{align}
   (B)
   &=
   \iint_{\mathcal Y^2} P_Y^{0}(\dd y)\,P_Y^{0}(\dd y')
   \iint_{\mathcal X^2} Q_{X\mid y}(\dd x)\,Q_{X\mid y'}(\dd x')\;
         k_X(x,x')\,k_Y(y,y').
   \label{eq:TermB}
\end{align}
\begin{align}
   (C)
   &=
   \iint_{\mathcal Y^2} P_Y^{0}(\dd y)\,P_Y^{0}(\dd y')
   \iint_{\mathcal X^2} P_{X\mid y}(\dd x)\,Q_{X\mid y'}(\dd x')\;
         k_X(x,x')\,k_Y(y,y').
   \label{eq:TermC}
\end{align}

For each \((y,y')\in\mathcal Y^2\) define
\begin{align}
   \Delta(y,y')
   &:=
   \iint_{\mathcal X^2} k_X(x,x')\,
      \Bigl\{
           P_{X\mid y}(\dd x)\,P_{X\mid y'}(\dd x')
          -P_{X\mid y}(\dd x)\,Q_{X\mid y'}(\dd x') \nonumber\\
          &\hspace{7em}
          -Q_{X\mid y}(\dd x)\,P_{X\mid y'}(\dd x')
          +Q_{X\mid y}(\dd x)\,Q_{X\mid y'}(\dd x')
      \Bigr\}.
   \label{eq:DeltaDef} 
\end{align}

Substituting~\eqref{eq:TermA}-\eqref{eq:TermC} into~\eqref{eq:ThreeTerms} yields (because the integrand and measure are symmetric in $(y,y')$):
\begin{equation}\label{eq:MMDJointDelta}
   \MMD_{k}^{2}(\tilde P,\tilde Q)
   =
   \iint_{\mathcal Y^2} k_Y(y,y')\,\Delta(y,y')\,
         P_Y^{0}(\dd y)\,P_Y^{0}(\dd y').
\end{equation}
We have the identity:
\begin{equation}\label{eq:k-to-mean-embedding} 
    \begin{aligned}
        \iint_{\mathcal X^2}
      k_X(x,x')\,
      P_{X\mid y}(\dd x)\,P_{X\mid y'}(\dd x') &= \iint_{\mathcal X^2}
      \bigl\langle\phi_X(x),\phi_X(x')\bigr\rangle
      P_{X\mid y}(\dd x)\,P_{X\mid y'}(\dd x') 
      \\&= \Bigl\langle
      \int_{\mathcal X}\phi_X(x)\,P_{X\mid y}(\dd x),\;
      \int_{\mathcal X}\phi_X(x')\,P_{X\mid y'}(\dd x')
 \Bigr\rangle \\
 &=\bigl\langle m_P(y),\,m_P(y')\bigr\rangle,
    \end{aligned}
\end{equation}
where $m_P(y)$ is the kernel mean embedding of $P_{X\mid y}$: $m_P(y):=m_{P_{X\mid y}}$, and $\phi$ is the feature map. The second equality uses Fubini’s lemma (the integrand is bounded). Similarly, we define $m_Q(y):=m_{Q_{X\mid y}}$. Using similar calculations to that of (\ref{eq:k-to-mean-embedding}) in all components of $\Delta(y,y')$ in (\ref{eq:DeltaDef}), we have
\[
   \Delta(y,y')
   =
   \bigl\langle
      m_P(y)-m_Q(y),\;
      m_P(y')-m_Q(y')
   \bigr\rangle_{\mathcal H_{k_X}},
\]

By Cauchy-Schwarz:
\begin{equation}\label{eq:DeltaBound}
   |\Delta(y,y')|
   \;\le\;
   \|m_P(y)-m_Q(y)\|\;
   \|m_P(y')-m_Q(y')\|.
\end{equation}
By definition, 
\begin{equation}\label{eq:dDef}
   \MMD_{k_X}\bigl(P_{X\mid y},Q_{X\mid y}\bigr)
   =\|m_P(y)-m_Q(y)\| \ge 0.
\end{equation}
Substituting~\eqref{eq:dDef} into~\eqref{eq:DeltaBound} yields
\begin{equation}\label{eq:DeltaFinalBound}
   |\Delta(y,y')|
   \;\le\;   \MMD_{k_X}\bigl(P_{X\mid y},Q_{X\mid y}\bigr)\,\MMD_{k_X}\bigl(P_{X\mid y'},Q_{X\mid y'}\bigr) \qquad \forall (y,y')\in\mathcal Y^2.
\end{equation}

Insert~\eqref{eq:DeltaFinalBound} into~\eqref{eq:MMDJointDelta} and use the
fact that \(k_Y\ge0\):
\begin{align*}
   \MMD_{k}^{2}(\tilde P,\tilde Q) &=\iint_{\mathcal Y^2} k_Y(y,y')\,\Delta(y,y')\,
         P_Y^{0}(\dd y)\,P_Y^{0}(\dd y')
    \\&\le \iint_{\mathcal Y^2} k_Y(y,y')\,\abs{\Delta(y,y')}\,
         P_Y^{0}(\dd y)\,P_Y^{0}(\dd y')
   \\&\le
   \iint_{\mathcal Y^2} k_Y(y,y')\,\MMD_{k_X}\bigl(P_{X\mid y},Q_{X\mid y}\bigr)\,\MMD_{k_X}\bigl(P_{X\mid y'},Q_{X\mid y'}\bigr)\,
         P_Y^{0}(\dd y)\,P_Y^{0}(\dd y').
   \\&\le
   \kappa_Y
   \iint_{\mathcal Y^2} \MMD_{k_X}\bigl(P_{X\mid y},Q_{X\mid y}\bigr)\,\MMD_{k_X}\bigl(P_{X\mid y'},Q_{X\mid y'}\bigr)\,
         P_Y^{0}(\dd y)\,P_Y^{0}(\dd y') \\
   &=
   \kappa_Y
   \Bigl(\int_{\mathcal Y} \MMD_{k_X}\bigl(P_{X\mid y},Q_{X\mid y}\bigr)\,P_Y^{0}(\dd y)\Bigr)^{2}
 \\
   &= \kappa_Y \left(\E_{Y\sim P_Y^{0}}
        \Bigl[\MMD_{k_X}\bigl(P_{X\mid Y},Q_{X\mid Y}\bigr)\Bigr]\right)^2
\end{align*}
Taking square roots on both sides completes the proof.
\end{proof}

\begin{lemma}[Jensen's inequality for the MMD]
\label{lem:MMD-Jensen}
\begin{equation}
    \MMD(P(X),Q(X))\le \int_{\mathcal{Y}}\MMD(P(X\mid y), Q(X\mid y))p(y) \;\dd y
\end{equation}
We call it Jensen's inequality because the LHS equals $\MMD(\E_{Y} P(X\mid Y),\E_{Y}Q(X\mid Y))$, and the RHS is $\E_Y \MMD(P(X\mid Y), Q(X\mid Y))$.
\end{lemma} 
\begin{proof}
    For any $f \in \mathcal{H}$ with $\norm{f}_{\mathcal{H}}\le 1$: 
\begin{align*}
    \abs{\int_{\mathcal{X}}f(x)P(x)\;\dd x-\int_{\mathcal{X}}f(x)Q(x)\;\dd x} &= \abs{\int_{\mathcal{X}} \int_{\mathcal{Y}}f(x)P(x\mid y) p(y) \;\dd y\;\dd x-\int_{\mathcal{X}} \int_{\mathcal{Y}}f(x)Q(x\mid y)p(y) \;\dd y\;\dd x} \\
    &=\abs{\int_{\mathcal{Y}} \left[ \int_{\mathcal{X}}f(x)P(x\mid y) \;\dd x-\int_{\mathcal{X}}f(x)Q(x\mid y)\;\dd x\right]p(y)\;\dd y} \\
    &\le \int_{\mathcal{Y}} \abs{ \int_{\mathcal{X}}f(x)P(x\mid y) \;\dd x-\int_{\mathcal{X}}f(x)Q(x\mid y)\;\dd x}p(y)\;\dd y. \\
\end{align*}

Because $|f(x)|\le \|f\|_{\mathcal H}\sqrt{k(x,x)}\le K$
with $K:=\sup_{x\in\mathcal X}\sqrt{k(x,x)}<\infty$, the integrand is
absolutely integrable.  Hence, Fubini’s theorem allows the change in the
order of integration in the second equality. The last inequality follows from triangle inequality. Now by definition of the MMD:
\begin{align*}
\MMD(P(X),Q(X)) &= \sup_{f:\norm{f}_{\mathcal{H}}\le 1}  \abs{\int_{\mathcal{X}}f(x)P(x)\;\dd x-\int_{\mathcal{X}}f(x)Q(x)\;\dd x} \\
&\le \sup_{f:\norm{f}_{\mathcal{H}}\le 1} \int_{\mathcal{Y}} \abs{ \int_{\mathcal{X}}f(x)P(x\mid y) \;\dd x-\int_{\mathcal{X}}f(x)Q(x\mid y)\;\dd x}p(y)\;\dd y 
\end{align*}

By triangle inequality. For every $f:\norm{f}_{\mathcal{H}}\le 1$ and $y\in\mathcal{Y}$, we have
\begin{align*}
    \abs{ \int_{\mathcal{X}}f(x)P(x\mid y) \;\dd x-\int_{\mathcal{X}}f(x)Q(x\mid y) \dd x} 
    &\le  \sup_{g:\norm{g}_{\mathcal{H}}\le 1} \abs{ \int_{\mathcal{X}}g(x)P(x\mid y) \;\dd x-\int_{\mathcal{X}}g(x)Q(x\mid y)\;\dd x} 
    \\&= \MMD(p(X\mid y), Q(X|y)).
\end{align*}
Taking integral over $\mathcal{Y}$ on both sides and then taking $\sup$ over $f:\norm{f}_{\mathcal{H}}\le 1$ finishes the proof.
\end{proof}

We are now ready to prove Theorem~\ref{thm:pseudo-classical}.
\begin{proof}[of Theorem~\ref{thm:pseudo-classical}]
\label{proof:pseudo-classical}
For clarity we collect the notation that governs the pseudo-sampling scheme.  
Given a pair \((w,y)\) and a data set of size \(n\), the (misspecified) posterior predictive distribution of \(X\) is
\[
   \Psi_{n}(\cdot\mid w,y)
      :=\int_{\Theta}\Pi_{\theta}(\cdot\mid w,y)
                       \Pi_{n}(\dd \theta\mid\D_{1:n}),
\]
and its marginal mixture with the true data law is
\[
   \Psi_{n}^{XY}(\dd x,  \dd y)
        :=\int_{\mathcal W}\Psi_{n}(\dd x\mid w,y)
                 p^{0}_{WY}(\dd w,\dd y).
\]
Fix an integer \(m\ge1\).  Conditional on the observed sample \(\D_{1:n}\) we draw, for each \(i=1,\dots,n\),
\[
   \tilde X_{ij}\stackrel{\text{iid}}\sim\Psi_{n}(\cdot\mid W_{i},Y_{i}),
   \qquad j=1,\dots,m,
\]
and collect all pseudo-draws in \(S:=\{\tilde X_{ij}:1\le i\le n,1\le j\le m\}\).  For later bounds we view each observation as generating two probability measures on \(\mathcal X\times\mathcal Y\),
\[
  \mu_{i}:=\Psi_{n}(\cdot\mid W_{i},Y_{i})\delta_{Y_{i}},
  \qquad
  \widehat\mu_{i}:=\frac1m\sum_{j=1}^{m}\delta_{(\tilde X_{ij},Y_{i})},
\]
respectively, the exact posterior predictive and its empirical counterpart based on \(m\) pseudo-samples.  Averaging over \(i\) produces the reference measure and the empirical (“pseudo”) measure
\[
  Q_{n}:=\frac1n\sum_{i=1}^{n}\mu_{i},
  \qquad
  \P^{\mathrm{pseudo}}_{XY}:=\frac1n\sum_{i=1}^{n}\widehat\mu_{i}.
\]
Finally, we define the joint distribution of $(X,Y)$ in the posterior model implied by $\theta^\ast$ as
\[
\Pi_{\theta^\ast}^{XY} (\dd x, \dd y) = \int_{\mathcal W\times\mathcal{Y}}\Pi_{\theta^\ast}(\dd x\mid w,y)\delta_y(\dd y)
                 p^{0}_{WY}(\dd w,\dd y).
\]
By the triangle inequality:
\[
 \MMD_{k}(\P^{\mathrm{pseudo}}_{XY},\P_{XY}^0)
 \le \MMD_{k}(\P^{\mathrm{pseudo}}_{XY},Q_n)
      +\MMD_{k}(Q_n,\Pi_{\theta^\ast}^{XY})
      +\MMD_{k}(\Pi_{\theta^\ast}^{XY},\P_{XY}^0).
\]
Take \(\E_{\D,S}\) on both sides, noting that the second term does not depend on $S$ and the third term does not depend on $\D$ or $S$, we have
\begin{align*}
    \E_{\D,S}\MMD_{k}(\P^{\mathrm{pseudo}}_{XY},\P_{XY}^0) \le \underbrace{\E_{\D,S} \MMD_{k}(\P^{\mathrm{pseudo}}_{XY},Q_n)}_{\text{term A}} &+ \underbrace{\E_{\D}\MMD_{k}(Q_n,\Pi_{\theta^\ast}^{XY})}_{\text{term B}} 
    \\ &+ \underbrace{\MMD_{k}(\Pi_{\theta^\ast}^{XY},\P_{XY}^0)}_{\text{term C}}.
\end{align*}

\emph{Step 1: Bounding term A.}

For any probability measure \(\nu\) on \(\mathcal{X}\times\mathcal{Y}\) we write
\[
      \varphi(\nu):=
      \E_{(X,Y)\sim\nu}\bigl[k\bigl((X,Y),\cdot\bigr)\bigr]
      \in\Hilb_k,
      \qquad
      \MMD_{k}(\nu_1,\nu_2)
        :=\bigl\|\varphi(\nu_1)-\varphi(\nu_2)\bigr\|_{\Hilb_k}.
\]
For a fixed observation pair \((W_i,Y_i)\) let
\[
      \Q_i
          :=\Psi_n\bigl(\cdot\mid W_i,Y_i\bigr)\delta_{Y_i}
      \quad\Bigl(\text{probability on }\mathcal{X}\times\mathcal{Y}\Bigr),
\]
and write its embedding
\(
  \varphi(\Q_i)\in\Hilb_k.
\)
Given \((\D,S)\equiv\bigl\{(W_i,Y_i);(\tilde X_{ij})_{j=1}^{m}\bigr\}\)
we form the empirical measure based on the $m$ pseudo-samples

\[
      \widehat{\Q}_i
      :=\frac1m\sum_{j=1}^{m}\delta_{(\tilde X_{ij},Y_i)},
      \qquad
      \varphi\bigl(\widehat{\Q}_i\bigr)
        =\frac1m\sum_{j=1}^{m}
           k\bigl((\tilde X_{ij},Y_i),\cdot\bigr)\in\Hilb_k.
\]

Conditioned on the data set $\D_{1:n}$, the random vectors
\[
      \varphi_{ij}
        :=k\bigl((\tilde X_{ij},Y_i),\cdot\bigr)\in\Hilb_k,
      \qquad j=1,\dots,m,
\]
are i.i.d.\ with mean
\(
   \E_{S\mid\D}[\varphi_{ij}]
    =\varphi(\Q_i),
\)
because each \(\tilde X_{ij}\) is drawn from
\(\Psi_n(\cdot\mid W_i,Y_i)\).
Moreover
\(
   \|\varphi_{ij}\|_{\Hilb_k}^{2}
      =k\bigl((\tilde X_{ij},Y_i),(\tilde X_{ij},Y_i)\bigr)
      \le \kappa_X \kappa_Y,
\)
so \(\|\varphi_{ij}\|_{\Hilb_k}\le \sqrt{\kappa_X \kappa_Y}\).

We have \(\E_{S\mid\D}[\varphi_{ij}-\varphi(\Q_i)]=0\) and
\(\|\varphi_{ij}-\varphi(\Q_i)\|_{\Hilb_k}\le\|\varphi_{ij}\|_{\Hilb_k}+\|\varphi(\Q_i)\|_{\Hilb_k}\le 2 \sqrt{\kappa_X \kappa_Y}\).

The difference of embeddings admits the representation
\[
      \varphi\bigl(\widehat{\Q}_i\bigr)
      -\varphi(\Q_i)
      =\frac1m\sum_{j=1}^{m}\left[\varphi_{ij}-\varphi(\Q_i)\right].
\]
Conditional on $\D$ the $\varphi_{ij}-\varphi(\Q_i)$ are independent and centred, so
\[
\begin{aligned}
   \E_{S\mid\D}
      \left[
         \left\|
            \varphi(\widehat{\Q}_i)
            -\varphi(\Q_i)
         \right\|_{\Hilb_k}^{2}
      \right]
   &=\E_{S\mid\D}
      \left[
        \frac1{m^{2}}
        \left\|
           \sum_{j=1}^{m}\varphi_{ij}-\varphi(\Q_i)
        \right\|^{2}_{\Hilb_k}
      \right]                                                \\
   &=\frac1{m^{2}}
      \sum_{j=1}^{m}
        \E_{S\mid\D}\left[\|\varphi_{ij}-\varphi(\Q_i)\|^{2}_{\Hilb_k}\right]
      \\&\le \frac{4\kappa_X \kappa_Y}{m}.
\end{aligned}
\]
Observe that $\E_{S\mid\D}\left[\varphi(\widehat{\Q}_i)-\varphi(\Q_i)\right] = \E_{S\mid\D}\left[\frac{1}{m}\sum_{j=1}^m \varphi_{ij}-\varphi(\Q_i)\right] = 0$. Since $\varphi(\widehat{\Q}_i)$ are independent across $i$ given $\D$, we have
\begin{equation}
\begin{aligned}
        \E_{S\mid \D}
   \bigl[\MMD_{k}^2(\P^{\mathrm{pseudo}}_{XY},Q_{n})\bigr] &= \E_{S\mid \D}\left[\norm{\varphi\left(\frac{1}{n}\sum_{i=1}^n\widehat{\Q}_i\right) - \varphi\left(\frac{1}{n}\sum_{i=1}^n\Q_i\right)}_{\Hilb_k}^2\right] 
   \\&= \E_{S\mid \D}\left[\norm{
        \frac1n\sum_{i=1}^{n}
          \Bigl(
            \varphi(\widehat{\Q}_i)
            -\varphi(\Q_i)
          \Bigr)
     }_{\Hilb_k}^2\right]
     \\&=\frac{1}{n^2}\sum_{i=1}^{n}\E_{S\mid \D}\left[
          \norm{
            \varphi(\widehat{\Q}_i)
            -\varphi(\Q_i)}_{\Hilb_k}^2\right]
    \\&\le \frac{4\kappa_X \kappa_Y}{nm}.
\end{aligned}
\end{equation}
Jensen’s inequality implies
\[
    \E_{S\mid\D}
        \Bigl[
           \MMD_{k}(\P^{\mathrm{pseudo}}_{XY},Q_{n})
        \Bigr]
      \le \sqrt{   \E_{S\mid\D}
      \left[
         \MMD_{k}^2(\P^{\mathrm{pseudo}}_{XY},Q_{n})
      \right]} \le  \frac{2\sqrt{\kappa_X \kappa_Y}}{\sqrt{nm}}.
\]
The RHS is deterministic. Taking expectation with respect to $\D$ yields
\begin{equation}
    \label{eq:term-A}
    \E_{\D,S}
   \bigl[\MMD_{k}(\P^{\mathrm{pseudo}}_{XY},Q_{n})\bigr] = \E_{\D}\E_{S\mid \D}
   \bigl[\MMD_{k}(\P^{\mathrm{pseudo}}_{XY},Q_{n})\bigr] \le \frac{2\sqrt{\kappa_X \kappa_Y}}{\sqrt{nm}} \le\frac{2\sqrt{\kappa_X \kappa_Y}}{\sqrt{n}}.
\end{equation}
\emph{Step 2: Bounding term B.}

By triangle inequality:
\begin{align*}
        \E_{\D}\left[\MMD_k(Q_n,\Pi_{\theta^{\ast}}^{XY})\right]&\le  \E_{\D}\left[\MMD_{k}\left(Q_n,
           \frac{1}{n}\sum_{i=1}^n\Pi_{\theta^{\ast}}(\cdot\mid W_i,Y_i) \delta_{Y_i}
         \right)\right]  \qquad \text{term B.1}
    \\&+  \E_{\D}
      \left[
        \MMD_{k}\left(
          \frac1n\sum_{i=1}^{n}
             \Pi_{\theta^{\ast}}(\cdot\mid W_i,Y_i)\delta_{Y_i},
          \Pi_{\theta^{\ast}}^{XY}
        \right)
      \right] \qquad \text{term B.2}.
\end{align*}

We will first bound term B.1 by the misspecified Bernstein von-Mises theorem established by \citet{kleijn2012bernstein}. For each $(W_i,Y_i)$
\begin{equation*}
    \begin{aligned}
         &\quad \MMD_{k_X}\left(
           \Psi_{n}(\cdot\mid W_i,Y_i),
           \Pi_{\theta^{\ast}}(\cdot\mid W_i,Y_i)
         \right) \\&= \MMD_{k_X}\left(\int_\Theta \Pi_{\theta}(\cdot\mid W_i, Y_i)\Pi_n(\dd\theta\mid\D), \int_\Theta \Pi_{\theta^{\ast}}(\cdot\mid W_i, Y_i) \Pi_n(\dd\theta\mid\D)\right)
         \\ &\le \int_{\Theta} \MMD_{k_X}\left(\Pi_{\theta}(\cdot\mid W_i, Y_i), \Pi_{\theta^{\ast}}(\cdot\mid W_i, Y_i) \right) \Pi_n(\dd\theta\mid\D)  \qquad\text{(Lemma~\ref{lem:MMD-Jensen})}
         \\&= \int_{B_n+B_n^C}\MMD_{k_X}\left(\Pi_{\theta}(\cdot\mid W_i, Y_i), \Pi_{\theta^{\ast}}(\cdot\mid W_i, Y_i) \right) \Pi_n(\dd\theta\mid\D),
    \end{aligned}
\end{equation*}
where $B_n:=\{\theta:\|\theta-\theta^\ast\|\le M_n/\sqrt n\}$, and $M_n$ is any sequence such that $M_n \to \infty$. We first choose $M_n$ such that $M_n\le \sqrt{n} \sup_{\Theta_\rho} \|\theta-\theta^\ast\|$ for all $n\ge 1$. Then $B_n \subset \Theta_\rho$ for all $n$.
By the MMD Lipschitz condition~\ref{A2}:
\begin{equation}
\begin{aligned}
        \MMD_{k_X}\left(\Pi_{\theta}(\cdot\mid W_i, Y_i), \Pi_{\theta^{\ast}}(\cdot\mid W_i, Y_i) \right) &\le L(W_i,Y_i)\norm{\theta-\theta^\ast}
        \\&\le \frac{M_n}{\sqrt{n}}L(W_i, Y_i), \qquad \forall \theta\in B_n.
\end{aligned}
\end{equation}
Therefore, 
\begin{equation}
\label{eq:Bn}
    \int_{B_n}\MMD_{k_X}\left(\Pi_{\theta}(\cdot\mid W_i, Y_i), \Pi_{\theta^{\ast}}(\cdot\mid W_i, Y_i) \right) \Pi_n(\dd\theta\mid\D) \le \frac{M_n}{\sqrt{n}}L(W_i, Y_i)\Pi_n(B_n) \le \frac{M_n}{\sqrt{n}}L(W_i, Y_i).
\end{equation}
We denote $\tilde r_n: =\Pi_n(B_n^C) \le 1$, then by Assumption~\ref{A1}, $\Pi_n$ satisfies posterior contraction \citep[Theorem 3.1]{kleijn2012bernstein}:
\begin{equation}
\label{eq:BnC}
    \int_{B_n^C}\MMD_{k_X}\left(\Pi_{\theta}(\cdot\mid W_i, Y_i), \Pi_{\theta^{\ast}}(\cdot\mid W_i, Y_i) \right) \Pi_n(\dd\theta\mid\D) \le \sqrt{\kappa_X} \Pi_n(B_n^C) = \sqrt{\kappa_X} \tilde r_n.
\end{equation}
Combining \eqref{eq:Bn} and \eqref{eq:BnC}, we have
\begin{equation*}
    \MMD_{k_X}\left(
           \Psi_{n}(\cdot\mid W_i,Y_i),
           \Pi_{\theta^{\ast}}(\cdot\mid W_i,Y_i)
         \right) \le \frac{M_n}{\sqrt{n}}L(W_i, Y_i) + \sqrt{\kappa_X} \tilde  r_n.
\end{equation*}
By Lemma~\ref{lem:MMD-conditional}:
\begin{equation*}
    \MMD_{k}\left(
           \Psi_{n}(\cdot\mid W_i,Y_i)\delta_{Y_i},
           \Pi_{\theta^{\ast}}(\cdot\mid W_i,Y_i)\delta_{Y_i}
         \right) \le \frac{\sqrt{\kappa_Y} M_n}{\sqrt{n}}L(W_i, Y_i) + \sqrt{\kappa_X\kappa_Y}\tilde r_n.
\end{equation*}
By the triangle inequality in $\Hilb_k$ and the linearity of kernel mean embedding $\varphi(\cdot)$:
\begin{equation*}
\begin{aligned}
        \MMD_{k}\left(\frac{1}{n}\sum_{i=1}^n P_i, \frac{1}{n}\sum_{i=1}^n Q_i\right) &= \norm{\varphi\left(\frac{1}{n}\sum_{i=1}^n P_i\right) - \varphi\left(\frac{1}{n}\sum_{i=1}^n Q_i\right)}_{\Hilb_k} 
        \\&\le \frac{1}{n}\sum_{i=1}^n \norm{\varphi(P_i) - \varphi(Q_i)}_{\Hilb_k} 
        \\&= \frac{1}{n}\sum_{i=1}^n \MMD_{k}(P_i, Q_i).
\end{aligned}
\end{equation*}
Therefore,
\begin{equation*}
        \MMD_{k}\left(\underbrace{\frac{1}{n}\sum_{i=1}^n
           \Psi_{n}(\cdot\mid W_i,Y_i) \delta_{Y_i}}_{=Q_n},
           \frac{1}{n}\sum_{i=1}^n\Pi_{\theta^{\ast}}(\cdot\mid W_i,Y_i) \delta_{Y_i}
         \right) \le \frac{\sqrt{\kappa_Y} M_n}{\sqrt{n}}\bar{L}(W, Y) +\sqrt{\kappa_X\kappa_Y} \tilde r_n,
\end{equation*}
where $\bar{L}(W, Y) = \frac{1}{n}\sum_{i=1}^n L(W_i, Y_i)$. Taking expectation over $\D = \{(W_i,Y_i)\}$ gives
\begin{equation}
\label{eq:Qn-1'}
     \E_{\D}\MMD_{k}\left(Q_n,
           \frac{1}{n}\sum_{i=1}^n\Pi_{\theta^{\ast}}(\cdot\mid W_i,Y_i) \delta_{Y_i}
         \right) \le \frac{\sqrt{\kappa_Y} M_n}{\sqrt{n}}\underbrace{\E_{W,Y\sim P^0}[L(W, Y)]}_{:= C_L <\infty \text{ by } \ref{A2}} +\sqrt{\kappa_X\kappa_Y} \E_{\D}[\tilde r_n].
\end{equation}
Equation~\eqref{eq:Qn-1'} holds for any $M_n$ such that $M_n\le \sqrt{n} \sup_{\Theta_{\rho}}\norm{\theta-\theta^\ast}$ for all $n$. We can scale with $M'_n := M_n/\max\{\sqrt{\kappa_Y}C_L,1\}$, then $M'_n$ is divergent, and $M'_n \le M_n \le \sqrt{n} \sup_{\Theta_{\rho}}\norm{\theta-\theta^\ast}$. Substituting $M_n$ with $M'_n$ in the arguments above yields
\begin{align*}
    \E_{\D}\MMD_{k}\left(Q_n,
           \frac{1}{n}\sum_{i=1}^n\Pi_{\theta^{\ast}}(\cdot\mid W_i,Y_i)\delta_{Y_i}
         \right) &\le \frac{\sqrt{\kappa_Y} C_L M'_n}{\sqrt{n}} +\sqrt{\kappa_X\kappa_Y} \E_{\D}[\tilde r'_n] 
         \\&\le \frac{M_n}{\sqrt{n}}+ \sqrt{\kappa_X\kappa_Y} \E_{\D}[\tilde r'_n] 
\end{align*}
By Assumption~\ref{A1}, $r_n:= \E_{\D}[\tilde r'_n] \le 1$ satisfies \citet[Theorem 3.1]{kleijn2012bernstein}:
\begin{equation*}
    r_n = \E_{\D}[\tilde r'_n] = \E_{\D}\left[\Pi_n(\norm{\theta-\theta^\ast}\ge M'_n/\sqrt{n})\right] \to 0.
\end{equation*}
Therefore,
\begin{equation}
\label{eq:Qn-1}
    \E_{\D}\MMD_{k}\left(Q_n,
           \frac{1}{n}\sum_{i=1}^n\Pi_{\theta^{\ast}}(\cdot\mid W_i,Y_i) \delta_{Y_i}
         \right) \le \frac{M_n}{\sqrt{n}} +\sqrt{\kappa_X \kappa_Y} r_n.
\end{equation}
The condition $M_n \le \sqrt{n} \sup_{\Theta_{\rho}}\norm{\theta-\theta^\ast}$ can be dropped. For any divergent $M_n$, we let $M''_n:=\min\{M_n, \sqrt{n} \sup_{\Theta_{\rho}}\norm{\theta-\theta^\ast}\}$ for each $n$, then $M''_n$ is also divergent, and $M''_n\le M_n$ for all $n$. Applying \eqref{eq:Qn-1} with $M''_n$ shows that \eqref{eq:Qn-1} holds for any divergent sequence $M_n$.

Next, we bound term B.2 via finite sample convergence of iid observations. For each observation \((W_i,Y_i)\) we denote  
\[
      \mu_i:=\Pi_{\theta^{\ast}}\bigl(\cdot\mid W_i,Y_i\bigr)
              \delta_{Y_i},
\]
Then 
\begin{align*}
    \E_{(W_i,Y_i)\sim p_{WY}^{0}}[\mu_i(\dd x,\dd y)] = \int_{\mathcal{W}\times\mathcal{Y}} \Pi_{\theta^{\ast}}\bigl(\dd x\mid w,y\bigr)
              \delta_{y}(\dd y)p^0_{WY}(\dd w, \dd y) = \Pi_{\theta^\ast}^{XY}(\dd x, \dd y)
\end{align*}

Define the RKHS embedding
\[
      \xi_i :=
      \int_{\mathcal{X}\times\mathcal{Y}}k((x,y),\cdot)\mu_i(\dd x,\dd y)
           \in\Hilb_k.
\]
Because the pairs \((W_i,Y_i)\) are i.i.d.\ under the data-generating
law \(p_{WY}^{0}\), the vectors \(\xi_1,\dots,\xi_n\) are i.i.d.\ in \(\Hilb_k\)
with a common mean
\[
      \xi_\ast := \int_{\mathcal{X}\times\mathcal{Y}}k((x,y),\cdot) \Pi_{\theta^{\ast}}^{XY}(\dd x,\dd y) =\int_{\mathcal{X}\times\mathcal{Y}}k((x,y),\cdot)
        \E_{W_i,Y_i\sim p_{WY}^{0}}
                    \bigl[\mu_i(\dd x,\dd y)\bigr] =
      \E_{W_i,Y_i\sim p_{WY}^{0}}\bigl[\xi_i\bigr].
\]
The last equality follows by Fubini's theorem since $k$ is bounded. By the independence of $\xi_1,\dots,\xi_n$, we have
\[
  \E_{\D}
      \norm{\frac1n\sum_{i=1}^n\xi_i-\xi_\ast}_{\Hilb_k}^{2}
   = \frac1{n^{2}}
     \sum_{i=1}^{n}\E_{\D}\|\xi_i-\xi_\ast\|_{\Hilb_k}^{2}
   \le \frac{4\kappa_X \kappa_Y}{n}.
\]
By Jensen's inequality: 
\[
\E_{\D}
      \norm{\frac1n\sum_{i=1}^n\xi_i-\xi_\ast}_{\Hilb_k}\le \frac{2\sqrt{\kappa_X \kappa_Y}}{\sqrt{n}}.
\]
The LHS is exactly the definition of the expectation of the MMD:
\begin{equation}\label{eq:Qn-2}
          \E_{\D}
      \Bigl[
        \MMD_{k}\Bigl(
          \frac1n\sum_{i=1}^{n}
             \Pi_{\theta^{\ast}}(\cdot\mid W_i,Y_i)\delta_{Y_i},
          \Pi_{\theta^{\ast}}^{XY}
        \Bigr)
      \Bigr]
      \le\frac{2\sqrt{\kappa_X \kappa_Y}}{\sqrt{n}}.
\end{equation}
Combining term B.1 \eqref{eq:Qn-1} and term B.2 \eqref{eq:Qn-2} gives
\begin{equation}
\label{eq:term-B}
    \E_{\D}\left[\MMD_k(Q_n,\Pi_{\theta^{\ast}}^{XY})\right]\le \frac{2\sqrt{\kappa_X \kappa_Y}}{\sqrt{n}}+ \frac{M_n}{\sqrt{n}} +\sqrt{\kappa_X\kappa_Y} r_n.
\end{equation}
\emph{Step 3: Bounding term C.}

By the chain rule of KL divergence: 
\[
  \KL\bigl(p_{0}\Vert p_{\theta^\ast}\bigr)
  =
  \KL\bigl(p^{0}_{W,Y}\Vert p^{\theta^\ast}_{W,Y}\bigr)
  +\E_{W,Y\sim p^{0}_{W,Y}}
      \KL\bigl(
        \Pi_0(\cdot\mid W,Y)\Vert\Pi_{\theta^\ast}(\cdot\mid W,Y)
      \bigr).
\]
Since the first term is non-negative, we have
\[
\E_{W,Y\sim p^{0}_{W,Y}}
      \KL\left(
        \Pi_0(\cdot\mid W,Y)\|\Pi_{\theta^\ast}(\cdot\mid W,Y)
      \right) \le \KL\bigl(p_{0}\Vert p_{\theta^\ast}\bigr).
\]
For any $(W,Y)$, we have 
\begin{equation*}
    \MMD_{k_X}(\Pi_0, \Pi_{\theta^\ast})\le 2\sqrt{\kappa_X}\norm{\Pi_0 -\Pi_{\theta^\ast}}_{\TV} \le 2\sqrt{\kappa_X} \sqrt{1-\exp(-\KL(\Pi_0\Vert\Pi_{\theta^\ast}))},
\end{equation*}
by the Bretagnolle--Huber inequality \citep{bretagnolle1979estimation}.
Since $\sqrt{1-\exp{(-x)}}$ is concave, we have, by Jensen's inequality:
\begin{equation*}
\begin{aligned}
        \E_{W,Y\sim p^0_{WY}} \MMD_{k_X}\left(\Pi_0 (X\mid W,Y), \Pi_{\theta^\ast}(X\mid W,Y)\right) &\le 2\sqrt{\kappa_X}\E_{W,Y\sim p^0_{WY}} \left[\sqrt{1-\exp(-\KL(\Pi_0\Vert\Pi_{\theta^\ast}))}\right]
        \\ &\le 2\sqrt{\kappa_X}\sqrt{1-\exp\left(-\E_{W,Y\sim p^0_{WY}}\KL(\Pi_0\Vert\Pi_{\theta^\ast}) \right)} 
        \\ &\le 2\sqrt{\kappa_X}\sqrt{1-\exp\left(-\KL(p_0\Vert p_{\theta^\ast}) \right)}.
\end{aligned}
\end{equation*}
Therefore,
\begin{equation*}
\begin{aligned}
        &\MMD_{k}\bigl(\P_{XY}^0,\Pi_{\theta^\ast}^{XY}\bigr) \\
        &=\MMD_{k}\Bigl(
               \int_{\mathcal W\times\mathcal Y}
                     \Pi_{0}(\dd x\mid w,y)\delta_{y}(\dd y)
                     p_{WY}^{0}(\dd w,\dd y),
               \int_{\mathcal W\times\mathcal Y}
                     \Pi_{\theta^\ast}(\dd x\mid w,y)\delta_{y}(\dd y)
                     p_{WY}^{0}(\dd w,\dd y)
             \Bigr)                                                          \\
        &\le \E_{(W,Y)\sim p_{WY}^{0}}
               \MMD_{k}\bigl(
                   \Pi_{0}(\cdot\mid W,Y)\delta_{Y},
                   \Pi_{\theta^\ast}(\cdot\mid W,Y)\delta_{Y}
               \bigr)    
               \quad\text{(Lemma~\ref{lem:MMD-Jensen})}                \\
        &\le  \sqrt{\kappa_Y}
            \E_{(W,Y)\sim p_{WY}^{0}}
               \MMD_{k_X}\bigl(
                   \Pi_{0}(\cdot\mid W,Y),
                   \Pi_{\theta^\ast}(\cdot\mid W,Y)
               \bigr) \quad\text{(Lemma~\ref{lem:MMD-conditional})}  
        \\&\le 2\sqrt{\kappa_X \kappa_Y}
              \sqrt{1-\exp\bigl(-\KL(p_{0}\Vert p_{\theta^\ast})\bigr)} .
\end{aligned}
\end{equation*}
Finally, we have the identity
\begin{align*}
\KL(p_{0}\Vert p_{\theta^\ast}) = \KL\left(p^0_X(X)f_N^0(W-X)f_E^0(Y-g^0(X)) \Vert p_X(X)f_N(W-X)f_E(Y-g(X,\theta^\ast))\right).    
\end{align*}
Since $W-X = N \ind X$, and $Y\mid X \ind W$, we have, by the chain rule
\begin{align*}
    \KL(p_{0}\Vert p_{\theta^\ast}) = \KL(p^0_X\Vert p_X) + \KL(f_N^0\Vert f_N)+ \E_{X\sim p_X^{0}}
        \KL\Bigl(p^0_{Y\mid X}
           \Big\Vert
           p^{\theta^\ast}_{Y\mid X}
        \Bigr) = \KL_X + \KL_N + \KL_E := \KL_{\ast},
\end{align*}
which gives 
\begin{equation}
    \label{eq:term-C}
    \MMD_{k}\bigl(\P_{XY}^0,\Pi_{\theta^\ast}^{XY}\bigr) \le 2\sqrt{\kappa_X \kappa_Y}
              \sqrt{1-\exp\bigl(-\KL_{\ast})}.
\end{equation}
Combining term A \eqref{eq:term-A}, term B \eqref{eq:term-B}, and term C \eqref{eq:term-C} and recalling that $\kappa := \kappa_X\kappa_Y$ proves Theorem~\ref{thm:pseudo-classical}.
\end{proof}

\subsection{Proof of Proposition~\ref{cor:pseudo-marginalX-classical}}
\label{proof:pseudo-marginalX-classical}
\begin{proof}[of Proposition~\ref{cor:pseudo-marginalX-classical}]
The proof mirrors that of Theorem~\ref{thm:pseudo-classical}, with the
joint kernel $k=k_{X}k_{Y}$ replaced everywhere by $k_{X}^2$ and every step using Lemma~\ref{lem:MMD-conditional} omitted; for brevity we will only record the changes.

\emph{Term A (Monte Carlo error).}
For each $(W_{i},Y_{i})$ let
$\mu_{i}:=\Psi_{n}(\cdot\mid W_{i},Y_{i})$ and
$\hat\mu_{i}:=m^{-1}\sum_{j}\delta_{\tilde X_{ij}}$. For $k_{X}^{2}$ the kernel mean embedding is
$\varphi_{k_{X}^{2}}(\mu)=\E_{X,X'\sim\mu}[k_{X}(X,\cdot)k_{X}(X',\cdot)]$,
whose norm is bounded by
$\kappa_{X}$. The same Hoeffding inequality in the RKHS $\Hilb_{k_{X}^{2}}$ yields
\[
   \E_{\D,S}\MMD_{k_{X}^{2}}\bigl(\P^{\mathrm{pseudo}}_{X},Q_{n}^{X}\bigr)
      \le \frac{2\kappa_{X}}{\sqrt {n}},
      \qquad
      Q_{n}^{X}:=\frac1n\sum_{i}\mu_{i}.
\]

\emph{Term B (posterior contraction).}
The Lipschitz envelope for $\MMD_{k_X^2}$ is
$2\kappa_{X}L(W,Y)$; every appearance of $\sqrt{\kappa_{X}}$ in the
$\MMD_{k_{X}}$ bound is replaced by $\kappa_{X}$.  Since no $Y$-kernel
is used, all factors $\kappa_{Y}$ disappear, which gives
\[
   \E_{\D}\bigl[\MMD_{k_{X}^{2}}(Q_{n}^X,\Pi_{\theta^{\ast}}^{X})\bigr]
      \le \frac{2\kappa_{X}}{\sqrt n}
           +\frac{M_{n}}{\sqrt n}
           +\kappa_{X} r_{n}.
\]

\emph{Term C (model misspecification).}
Applying the same Bretagnolle--Huber inequality \citep{bretagnolle1979estimation} directly to
$\Pi_{0}(X\mid W,Y)$ and $\Pi_{\theta^{\ast}}(X\mid W,Y)$ yields
\[
   \MMD_{k_{X}^2}\bigl(\P_{X}^{0},\Pi_{\theta^{\ast}}^{X}\bigr)
      \le 2\kappa_{X}
             \sqrt{1-\exp(-\KL_{\ast})}.
\]

Combining the three bounds with the triangle
inequality finishes the proof. 
\end{proof}

\subsection{Proof of Proposition~\ref{cor:pseudo-berkson}}
\label{proof:pseudo-berkson}
\begin{proof}[of Proposition~\ref{cor:pseudo-berkson}]
We first prove the joint bound. This proof follows the structure of Theorem~\ref{thm:pseudo-classical},
highlighting the points at which the Berkson design requires additional justification.


\emph{Term A (Monte Carlo error).}
Since $(W_i, Y_i)$ are still i.i.d., $\mu_{i}$ and $\widehat\mu_{i}$ are defined exactly as in the classical proof.
The bound~(\ref{eq:term-A}) is unchanged.

\emph{Term B (posterior contraction).} \textbf{B.1} Contracting the misspecified posterior still relies on the results of \citet{kleijn2012bernstein}, so Inequality~(\ref{eq:Qn-1}) holds. \textbf{B.2} The same variance calculation gives (\ref{eq:Qn-2}). Combining \textbf{B.1} and \textbf{B.2} recovers (\ref{eq:term-B}) verbatim.

\emph{Term C (model misspecification).} Since $X-W = N \ind W$ and $Y\mid X \ind W$, by the chain rule we have \begin{align*}
    \KL(p_0\Vert p_{\theta^\ast}) &= \KL\left(p_W^0(W)f_N^0(X-W)f_E^0(Y-g^0(X))\Vert  p_W^0(W)f_N(X-W)f_E(Y-g(X,\theta^\ast))\right) 
    \\&=\KL(f_N^0\Vert f_N) + \E_{W\sim p_{W}^0}\E_{X\sim f_N^0(X\mid W)}\left[\KL(p^0_{Y\mid X}\Vert p^{\theta^\ast}_{Y\mid X})\right]
    \\&= \KL_N + \KL_E
    \\&= \KL_{\ast}.
\end{align*}
Applying the total-variation-to-KL bound as in
(\ref{eq:term-C}) therefore gives
\[
   \MMD_{k}\bigl(\P_{XY}^0,\Pi^{XY}_{\theta^{\ast}}\bigr)
      \le 2\sqrt{\kappa}\sqrt{1-\exp(-\KL_{\ast})}.
\]
Inserting the bounds for Terms~A--C into the triangle inequality used at the start of the proof of Theorem~\ref{thm:pseudo-classical} yields the stated inequality, with $\KL_{\ast}$ now equal to $\KL_{N}+\KL_{E}$. All other constants and residual terms are identical to those in the classical case. The marginal-$X$ bound follows exactly as in Proposition~\ref{cor:pseudo-marginalX-classical}.
\end{proof}

\subsection{Proof of Lemma~\ref{thm:nopseudo}}
\label{proof:nopseudo}
We split this proof into two parts: the marginal bound \eqref{eq:nopseudo-marginal} and the joint bound \eqref{eq:nopseudo-joint}.

\begin{proof}[of the Joint MMD bound~\eqref{eq:nopseudo-joint}]
We first prove the joint bound in the Berkson ME model, where
\[
X = W + N,\quad Y = g^0(X)+E, \quad N\sim F_{N}^{0}, \quad E\sim F_{E}^{0},\quad  N, E\ind W, \;N\ind E.
\]
By triangle inequality:
\[
\MMD_k(\P_{XY}^0, \hat{\P}^n_{WY}) \leq \MMD_k(\P_{XY}^0, \P_{WY}^0) + \MMD_k(\P_{WY}^0, \hat{\P}^n_{WY}).
\]

For the first term, we write
$$
\MMD_k(\P_{XY}^0, \P_{WY}^0) = \left\| \varphi(\P_{XY}^0) -  \varphi(\P_{WY}^0) \right\|_{\Hilb_k}.
$$
Let $\Phi(x,y) := k((x,y), \cdot) = k_X(x,\cdot)k_Y(y,\cdot) \in \Hilb_k$ denote the feature map of the product kernel $k((x,y), (x',y')) = k_X(x,x')k_Y(y,y')$, and let $\Phi_X:=k_X(x,\cdot)$ and $\Phi_Y:=k_Y(y,\cdot)$ be the respective feature maps of $k_X$ and $k_Y$. Then we can write the kernel mean embeddings $\varphi(\cdot)$ of $\P_{XY}^0$ and $\P_{WY}^0$ as
\begin{equation}
\label{kernel-mean-embeddings}
    \begin{aligned}
         \varphi(\P_{XY}^0) &= \E_{W,N,E}[\Phi(W + N, Y)] = \E_{W,N,E}\left[\Phi_X(W + N)\Phi_Y(Y)\right], \\
 \varphi(\P_{WY}^0) &= \E_{W,N,E}[\Phi(W, Y)] = \E_{W,N,E}\left[\Phi_X(W)\Phi_Y(Y)\right],
    \end{aligned}
\end{equation}
where $Y= g^0(W + N) + E$.
Taking their difference and applying the reproducing property:
\begin{equation}
\label{eq:main-ineq}
\begin{aligned}
        \MMD_k(\P_{XY}^0, \P_{WY}^0) &=  \norm{\E_{W,N,E} \Bigl[\left(\Phi_X(W + N)-\Phi_X(W)\right)\Phi_Y(Y)\Bigr]}_{\Hilb_k} \\
        &\le \E_{W,N,E} \left[\norm{\left(\Phi_X(W + N)-\Phi_X(W)\right)\Phi_Y(Y)}_{\Hilb_k} \right]
        \\&=\E_{W,N,E} \left[\norm{\Phi_X(W + N)-\Phi_X(W)}_{\Hilb_{k_X}}\norm{\Phi_Y(Y)}_{\Hilb_{k_Y}} \right]
        \\&\le \sqrt{\E_{W,N}\left[\norm{\Phi_X(W + N)-\Phi_X(W)}_{\Hilb_{k_X}}^2\right]} \sqrt{\E_{W,N,E}\left[\norm{\Phi_Y(Y)}_{\Hilb_{k_Y}}^2\right]}
        \\&\le \sqrt{\kappa_Y} \sqrt{\E_{W,N}\left[\norm{\Phi_X(W + N)-\Phi_X(W)}_{\Hilb_{k_X}}^2\right]} 
\end{aligned}
\end{equation}
The first inequality follows from Jensen's inequality; the second equality holds because $\Phi_X(\cdot)\Phi_Y(\cdot)$ is a pure tensor in the RKHS $\Hilb_{k_X}\otimes \Hilb_{k_Y}$ with product kernel; the second inequality is Cauchy-Schwarz, and the final inequality follows by $k_Y(y,y')\le \kappa_Y\quad\forall y,y'$.
Recall that $k_X$ is translation-invariant, i.e., $k_X (x,x') = \psi (x-x')$ with $\psi(0) = \kappa_X$, therefore
\begin{equation}
\label{eq:diag}
    \begin{aligned}
        \E_{W,N}\left[\norm{\Phi_X(W + N)-\Phi_X(W)}_{\Hilb_{k_X}}^2\right] &= \E_{W,N}\left[k_X(W,W) + k_X(W+N, W+N)- 2 k_X(W+N,W)\right]
        \\&=2\kappa_X -2\E_N[\psi(N)].
    \end{aligned}
\end{equation}
Now 
\begin{equation}
\label{eq:ME-MMD-derviation}
    \MMD_{k_X}^2(F_N^0,\delta_0) = \E_{N,N'}[\psi(N-N')]+\kappa_X - 2\E_N[\psi(N)].
\end{equation}
We note the following identities:
\begin{equation*}
\begin{aligned}
    \E_{N,N'}[\psi(N-N')] &= \E_{N,N'}[k(N,N')] = \E_{N,N'}[\langle\Phi_X(N),\Phi_X(N')\rangle] 
    \\&= \langle\E_N[\Phi_X(N)],\E_{N'}[\Phi_X(N')]\rangle = \norm{\E_N[\Phi_X(N)]}_{\Hilb_{k_X}}^2,
\end{aligned}
\end{equation*}
and
\begin{equation*}
    \E_N[\psi(N)] =\E_N[k_X(N,0)] = \E_N[ \langle\Phi_X(N),\Phi_X(0)\rangle] = \langle\E_N[\Phi_X(N)],\Phi_X(0)\rangle.
\end{equation*}
By Cauchy-Schwarz:
\begin{equation*}
    \E_N[\psi(N)]^2 \le  \norm{\E_N[\Phi_X(N)]}_{\Hilb_{k_X}}^2 \norm{\Phi_X(0)}_{\Hilb_{k_X}}^2 = \kappa_X \E_{N,N'}[\psi(N-N')]
\end{equation*}
Hence
\begin{equation*}
    \begin{aligned}
        (\kappa_X - \E_N[\psi(N)])^2 &= \kappa_X^2 + \E_N[\psi(N)]^2 -2\kappa_X \E_N[\psi(N)]
        \\&\le \kappa_X^2 +  \kappa_X \E_{N,N'}[\psi(N-N')]-2\kappa_X \E_N[\psi(N)]
        \\&=\kappa_X(\kappa_X+\E_{N,N'}[\psi(N-N')] - 2 \E_N[\psi(N)]) \qquad \text{by }\eqref{eq:ME-MMD-derviation}
        \\&=\kappa_X \MMD_{k_X}^2(F_N^0,\delta_0)
    \end{aligned}
\end{equation*}
Taking square-root and substituting into~\eqref{eq:diag}:
\begin{equation*}
    \E_{W,N}\left[\norm{\Phi_X(W + N)-\Phi_X(W)}_{\Hilb_{k_X}}^2\right] \le 2\sqrt{\kappa_X}\MMD_{k_X}(F_N^0,\delta_0)
\end{equation*}
By \eqref{eq:main-ineq}:
\begin{equation}
\label{eq:first-term}
    \MMD_k(\P_{XY}^0, \P_{WY}^0) \le \sqrt{2}\kappa_Y^{1/2}\kappa_X^{1/4} \sqrt{\MMD_{k_X}(F_N^0,\delta_0)}
\end{equation}
For the second term, since $\D = (W_i,Y_i)_{i=1}^n$ consists of iid samples  $\{(W_i,Y_i)\}_{i=1}^n\overset{i.i.d.}{\sim} \P_{WY}^0$, we have, by \citet[Lemma S6]{alquier2024universal}:
\begin{equation}
\label{eq:second-term}
    \E_{\D}\MMD_k(\hat{\P}^n_{WY}, \P_{WY}^0) \le \frac{\sqrt{\kappa}}{\sqrt{n}}.
\end{equation}
Here we generalized their last inequality in the proof where \citet{alquier2024universal} assumed $k(\cdot,\cdot)\le1$ but we have $k(\cdot,\cdot)\le \kappa$. Combining \eqref{eq:first-term} and \eqref{eq:second-term} finishes the proof for the Berkson case.

For the classical ME case, recall that 
\begin{equation*}
    W=X+N,\quad Y=g^0(X)+E,\quad X\sim \P_{X}^{0}, \quad N\sim F_{N}^{0}, \quad E\sim F_{E}^{0},\quad N,E\ind X,\; N\ind E. 
\end{equation*}
We only need to substitute every $W$ with $X$ in \eqref{kernel-mean-embeddings}, \eqref{eq:main-ineq}, and \eqref{eq:diag}, which gives
\begin{equation*}
\begin{aligned}
        \MMD_k(\P_{XY}^0, \P_{WY}^0) &\le \sqrt{\kappa_Y} \sqrt{\E_{X,N}\left[\norm{\Phi_X(X + N)-\Phi_X(X)}_{\Hilb_{k_X}}^2\right]} 
        \\ &\le  \sqrt{2}\kappa_Y^{1/2}\kappa_X^{1/4} \sqrt{\MMD_{k_X}(F_N^0,\delta_0)}.
\end{aligned}
\end{equation*}
Combining with \eqref{eq:second-term}, which only relies on $\{(W_i,Y_i)\}_{i=1}^n\overset{i.i.d.}{\sim} \P_{WY}^0$, completes the proof.
\end{proof}

\begin{proof}[of the Marginal MMD bound~\eqref{eq:nopseudo-marginal}]
By the triangle inequality and then taking expectation over $\D$,
\[
\E_{\D}\MMD_{k_X^2}(\P_X^0,\hat{\P}_W^n)
\le
\MMD_{k_X^2}(\P_X^0,\P_W^0)
+
\E_{\D}\MMD_{k_X^2}(\P_W^0,\hat{\P}_W^n).
\]
For the second term, since $W_{1:n}\simiid \P_W^0$ and $\sup_x k_X^2(x,x)=\kappa_X^2$, by \citet[Lemma~S6]{alquier2024universal}
\begin{equation}
\label{eq:marginal-sampling}
\E_{\D}\MMD_{k_X^2}(\P_W^0,\hat{\P}_W^n)\le\frac{\kappa_X}{\sqrt{n}}.
\end{equation}
It remains to bound the population term $\MMD_{k_X^2}(\P_X^0,\P_W^0)$. Let $\Phi_X^{(2)}$ be the feature map of $k_X^2(\cdot,\cdot)=k_X\otimes k_X$. As in \eqref{eq:main-ineq} (with $k$ replaced by $k_X^2$ and no $Y$ factor),
\[
\MMD_{k_X^2}(\P_X^0,\P_W^0)
=\norm{\E_{X,W}[\Phi_X^{(2)}(X)-\Phi_X^{(2)}(W)]}_{\Hilb_{k_X^2}}
\le
\sqrt{\E_{X,W}\|\Phi_X^{(2)}(X)-\Phi_X^{(2)}(W)\|^2_{\Hilb_{k_X^2}}}.
\]
By the same expansion as \eqref{eq:diag} and bounded translation-invariance of $k_X$,
\begin{equation}
\label{eq:diag-marginal}
\begin{aligned}
\E_{X,W}\norm{\Phi_X^{(2)}(X)-\Phi_X^{(2)}(W)}_{\Hilb_{k_X^2}}^2
&=\E_{X,W}[k_X^2(X,X)+k_X^2(W,W)-2k_X^2(X,W)] \\
&=2\kappa_X^2-2\E_{X,W}[k_X(X,W)^2].
\end{aligned}
\end{equation}
Under either Berkson ($X=W+N$) or classical ($W=X+N$) error, $k_X(X,W)=\psi(W-X)=\psi(\pm N)$, hence $k_X^2(X,W)=\psi(N)^2$ and
\[
\E_{X,W}\|\Phi_X^{(2)}(X)-\Phi_X^{(2)}(W)\|^2
=2\kappa_X^2-2\E_N[\psi(N)^2].
\]
Next, exactly as in \eqref{eq:ME-MMD-derviation}-\eqref{eq:first-term} but with $k_X$ replaced by $k_X^2$, we have
\[
\MMD_{k_X^2}^2(F_N^0,\delta_0)
=\E_{N,N'}[\psi(N-N')^2]+\kappa_X^2-2\E_N[\psi(N)^2],
\]
Using the same Cauchy-Schwarz argument as before, we have
\[
\E\|\Phi_X^{(2)}(X)-\Phi_X^{(2)}(W)\|^2
\le2\kappa_X\MMD_{k_X^2}(F_N^0,\delta_0),
\]
and hence
\begin{equation}
\label{eq:marginal-pop}
\MMD_{k_X^2}(\P_X^0,\P_W^0)
\le
\sqrt{\E\|\Phi_X^{(2)}(X)-\Phi_X^{(2)}(W)\|^2}
\le
\sqrt{2\kappa_X\MMD_{k_X^2}(F_N^0,\delta_0)}.
\end{equation}
Finally, combining \eqref{eq:marginal-pop} and \eqref{eq:marginal-sampling} completes the proof of \eqref{eq:nopseudo-marginal}.
\end{proof}

\subsection{Proof of Theorem~\ref{thm:consistency-main}}
\label{proof:consistency-main}
Before proving Theorem~\ref{thm:consistency-main}, we first state and prove the following lemma, which is an application of the classical concentration theory for martingales in Banach spaces by \citet{pinelis1994optimum}.
\begin{lemma}[Conditional \& unconditional Bernstein bound in Hilbert space]
\label{lem:PinelisConditional}
Let $\mathcal H$ be a real separable Hilbert space. Let $V_1,\dots,V_n:\Omega\to\mathcal H$ satisfy
\begin{enumerate}[label=(\roman*)]
\item\label{it:indep}
      $V_1,\dots,V_n$ are \emph{independent} under
      $\Pr_{\DP}(\cdot\mid\D,S)$;
\item\label{it:mean0}
      $\E_{\DP}[V_i\mid\D,S]=0$ for every $i$;
\item\label{it:bound}
      $\|V_i\|_{\mathcal H}\le B$ for a deterministic
      constant $B<\infty$.
\end{enumerate}
Then for every $\varepsilon>0$
\begin{equation}\label{eq:PinelisCond}
   \Pr_{\DP}\Bigl(
      \Bigl\|
        \frac1n\sum_{i=1}^{n}V_i
      \Bigr\|_{\mathcal H}>\varepsilon
      \Bigm|\D,S
   \Bigr)
   \le
   2\exp\Bigl(
        -\frac{n\varepsilon^{2}}{4B^{2}}
     \Bigr).
\end{equation}
A simple consequence is that the same exponential bound holds under the \emph{full} law:
\begin{equation}\label{eq:PinelisUncond}
   \Pr\Bigl(
      \Bigl\|
        \frac1n\sum_{i=1}^{n}V_i
      \Bigr\|_{\mathcal H}>\varepsilon
   \Bigr)
   \le
   2\exp\Bigl(
        -\frac{n\varepsilon^{2}}{4B^{2}}
     \Bigr),
\end{equation}
which means that
\[
\Bigl\|\frac1n\sum_{i=1}^{n}V_i\Bigr\|_{\Hilb} \overset{\Pr}{\rightarrow}0.
\]
\end{lemma}

\begin{proof}
Conditioning on $(\D,S)$, we define a $\mathcal H$-valued martingale by
\begin{equation}
    f_j = \begin{cases}
        0, \qquad &j=0 \\
        \frac{1}{n}\sum_{i=1}^{j}V_i, \qquad &1\le j\le n \\
        \frac{1}{n}\sum_{i=1}^{n}V_i, \qquad & j>n.
    \end{cases}
\end{equation}
Let $\mathcal{F}_j = \sigma(V_1,\dots,V_j,\D,S)$, then the sequence $\{f_j,\mathcal{F}_j\}_{j=1}^\infty$ forms a martingale under $\Pr_{\DP}(\cdot\mid\D,S)$ by the independence of $V_i$. Its differences are
$d_j:=f_j-f_{j-1}=n^{-1}V_j$,
which satisfy $\norm{d_j}_{\Hilb}\le B/n$ for $j\le n$ and $\norm{d_j}_{\Hilb}\equiv 0$ for $j>n$. Let $f^\ast := \sup_{j\ge 0}\norm{f_j}_{\Hilb}$, we have $f^\ast  \ge \norm{f_n}_{\Hilb} =\norm{\frac{1}{n}\sum_{i=1}^n V_i}_{\Hilb}$ almost surely. Because $\mathcal H$ is a Hilbert space, it is a $(2,1)$-smooth Banach space. Here $(2,1)$-smooth means that the parallelogram identity holds: for every $x,y\in\Hilb$, we have $\|x+y\|_{\Hilb}^2+\|x-y\|_{\Hilb}^2\le 2\|x\|_{\Hilb}^2 + 2\|y\|_{\Hilb}^2$. Applying \citet[Theorem 3.1]{pinelis1994optimum} gives
\begin{equation}
\begin{aligned}
         \Pr_{\DP}\bigl(f^{\ast}\ge\varepsilon\mid\D,S\bigr) &\le 2 \exp \left\{-\lambda \varepsilon+\left\|\sum_{j=1}^{\infty} \E_{j-1}\left(\exp \left(\lambda\|d_j\|_{\Hilb}\right)-1-\lambda\|d_j\|_{\Hilb}\right)\right\|_{\infty}\right\} \\
         &= 2 \exp \left\{-\lambda \varepsilon+\left\|\sum_{j=1}^{n} \E_{j-1}\left(\exp \left(\lambda\|d_j\|_{\Hilb}\right)-1-\lambda\|d_j\|_{\Hilb}\right)\right\|_{\infty}\right\} ,
\end{aligned}
\end{equation}
where $\norm{\cdot}_{\infty}$ represents the essential supremum of the enclosed random variable, and it is required that $\E\bigl[e^{\lambda\norm{d_j}_{\Hilb}}\bigr]<\infty$.

Fix $\varepsilon>0$ with $\varepsilon<2B$. We set
$$
   \lambda:=\frac{n\varepsilon}{2B^{2}}
   \quad\Longrightarrow\quad
   \lambda\norm{d_j}_{\Hilb}\le\frac{\varepsilon}{2B}<1
   \quad\text{for all }j,
$$
so the exponential moments
$\E\bigl[e^{\lambda\norm{d_j}_{\Hilb}}\bigr]<\infty$.

Using $\exp(u)-1-u\le u^{2}$ for $0\le u<1$ and $\norm{d_j}_{\Hilb}\le B/n$, we have
$$
   \sum_{j=1}^{n}
   \E_{j-1}\bigl[\exp(\lambda\norm{d_j}_{\Hilb})-1-\lambda\norm{d_j}_{\Hilb}\bigr]
   \le
   \lambda^{2}
   \sum_{j=1}^{n}\E_{j-1}\norm{d_j}_{\Hilb}^{2}
   \le
   \lambda^{2}
   n\Bigl(\frac{B}{n}\Bigr)^{2}
   =\frac{\lambda^{2}B^{2}}{n}.
$$
The RHS is deterministic, so taking the essential supremum $\norm{\cdot}_{\infty}$ does not change the bound. Since
$$
   \exp\left\{-\lambda \varepsilon+\frac{\lambda^{2}B^{2}}{n}\right\}
   =\exp\left\{-\frac{n\varepsilon^{2}}{2B^{2}}
     +\frac{1}{n}\Bigl(\frac{n\varepsilon}{2B^{2}}\Bigr)^{2}B^{2}\right\}
   =\exp\left(-\frac{n\varepsilon^{2}}{4B^{2}}\right),
$$
we have
\begin{equation*}
       \Pr_{\DP}\Bigl(
      \Bigl\|
        \frac1n\sum_{i=1}^{n}V_i
      \Bigr\|_{\mathcal H}>\varepsilon
      \Bigm|\D,S
   \Bigr)
   \le\Pr_{\DP}\bigl(f^{\ast}\ge\varepsilon\mid\D,S\bigr) \le
   2\exp\Bigl(
        -\frac{n\varepsilon^{2}}{4B^{2}}
     \Bigr).
\end{equation*}
This is exactly \eqref{eq:PinelisCond}. Since the RHS is deterministic once $\varepsilon$ is chosen, taking the expectation with respect to
$\Pr_{\D,S}$ preserves the right-hand side, yielding
\eqref{eq:PinelisUncond}.
\end{proof}

We are now ready to prove Theorem~\ref{thm:consistency-main}.
\begin{proof}[of Theorem~\ref{thm:consistency-main}]
\emph{Step (a)}

We first recall the following notation:
\begin{equation*}
   \mathcal{P}_{XY}^{\DP}=\frac1n\sum_{i=1}^{n}\P_{XY\mid W_i}^{\DP}, \qquad
   \mathcal{P}_{X}^{\DP}=\frac1n\sum_{i=1}^{n}\P_{X\mid W_i}^{\DP}; 
\end{equation*}
\begin{equation*}
    \P^{\mathrm{base}}_{XY,i}=\frac{c}{c+m}\Q_{XY,i}+\frac{1}{c+m}\sum_{k=1}^m \delta_{\bigl(\tilde{X}_{ik},Y_i\bigr)}, \qquad \P^{\mathrm{base}}_{X,i}=\frac{c}{c+m}\Q_{X,i}+\frac{1}{c+m}\sum_{k=1}^m \delta_{\tilde{X}_{ik}},   
\end{equation*}
Then $\P^{\mathrm{base}}_{XY,i} = \E_{\DP,i}\left[\P_{XY\mid W_i}^{\DP}\mid \D,S\right]$, and $\P^{\mathrm{base}}_{X,i} = \E_{\DP,i}\left[\P_{X\mid W_i}^{\DP}\mid \D,S\right]$.
Furthermore, we define
\begin{equation*}
   \P^{\mathrm{base}}_{XY}:=\frac1n\sum_{i=1}^{n}\P^{\mathrm{base}}_{XY,i}, \qquad
   \P^{\mathrm{base}}_{X}:=\frac1n\sum_{i=1}^{n}\P^{\mathrm{base}}_{X,i}.
\end{equation*}
We also define the MMD quantities
\begin{align}
   M_n(\theta)
      &:=\MMD_k\bigl(\mathcal{P}_{XY}^{\DP},
          \mathcal{P}_{X}^{\DP} \P_{g(X, \theta)}\bigr),\label{eq:MnNEW}\\[3pt]
  \tilde M_n(\theta)
      &:=\MMD_k\bigl(\mathcal{P}_{XY}^{\DP},
          \P_X^{\mathrm{base}} \P_{g(X, \theta)}\bigr),\label{eq:MncondNEW}\\[3pt]
   M_n^\ast(\theta)
      &:=\MMD_k\bigl(\P^{\mathrm{base}}_{XY},
          \P_X^{\mathrm{base}} \P_{g(X, \theta)}\bigr).\label{eq:phinNEW}
\end{align}
By triangle inequality and \citet[Lemma 2]{alquier2024universal}, we have, for all $\theta\in\Theta$:
\begin{equation*}
    \abs{M_n(\theta) - \tilde M_n(\theta)}\le \MMD_{k}(\mathcal{P}_X^{\DP}\P_{g(X, \theta)},\P_X^{\mathrm{base}} \P_{g(X, \theta)})\le \Lambda\MMD_{k_X^{2}}\bigl(\mathcal{P}_{X}^{\DP},\P_X^{\mathrm{base}}\bigr)
\end{equation*}
The RHS does not depend on $\theta$, so 
\begin{equation*}
    \sup_{\theta\in\Theta}\abs{M_n(\theta) - \tilde M_n(\theta)}\le\Lambda\MMD_{k_X^{2}}\bigl(\mathcal{P}_{X}^{\DP},\P_X^{\mathrm{base}}\bigr)
\end{equation*}
Let $V_i:=\varphi_{k_X^{2}}(\P_{X\mid W_i}^{\DP}) - \varphi_{k_X^{2}}(\P^{\mathrm{base}}_{X,i})$, then
\begin{equation*}
    \MMD_{k_X^{2}}\bigl(\mathcal{P}_{X}^{\DP},\P_X^{\mathrm{base}}\bigr) =\norm{\varphi\left(\frac1n\sum_{i=1}^n \P_{X\mid W_i}^{\DP}\right) - \varphi\left(\frac1n\sum_{i=1}^n \P^{\mathrm{base}}_{X,i}\right)}_{\Hilb_{k_X^2}} = \norm{\frac1n\sum_{i=1}^{n}V_i}_{\Hilb_{k_X^2}}
\end{equation*}
Since $V_1,\dots,V_n$ are independent conditional on $(\D,S)$, $\E_{\DP}[V_i\mid \D,S] = 0$, and $\norm{V_i}_{\Hilb_{k_X^2}} \le 2\kappa_X$, we get, by Lemma~\ref{lem:PinelisConditional}:
\[
\Bigl\|\frac1n\sum_{i=1}^{n}V_i\Bigr\|_{\Hilb_{k_X^2}} \overset{\Pr}{\rightarrow}0.
\]
Therefore,
\begin{equation}
    \sup_{\theta\in\Theta}\abs{M_n(\theta) - \tilde M_n(\theta)}\overset{\Pr}{\rightarrow}0.
\end{equation}
Again, by triangle inequality, we have
\begin{equation*}
    \sup_{\theta\in\Theta}\abs{\tilde M_n(\theta) -  M_n^\ast(\theta)}\le \MMD_{k}(\mathcal{P}_{XY}^{\DP},\P^{\mathrm{base}}_{XY})
\end{equation*}
Similarly, we let $U_i: =\varphi_{k}(\P_{XY\mid W_i}^{\DP}) - \varphi_{k}(\P^{\mathrm{base}}_{XY,i})$ and apply Lemma~\ref{lem:PinelisConditional} with $\E_{\DP}[U_i\mid\D,S] = 0$ and $\norm{U_i}_{\Hilb_k}\le2\sqrt{\kappa}$, we have
\begin{equation}
\label{eq:step-a}
    \sup_{\theta\in\Theta}\abs{\tilde M_n(\theta) -  M_n^\ast(\theta)}\overset{\Pr}{\rightarrow}0.
\end{equation}
Since 
\begin{equation*}
    \sup_{\theta\in\Theta}\abs{M_n(\theta)-M_n^*(\theta)}\le \sup_{\theta\in\Theta}\abs{M_n(\theta) - \tilde M_n(\theta)} + \sup_{\theta\in\Theta}\abs{\tilde M_n(\theta) -  M_n^\ast(\theta)},
\end{equation*}
We have 
\begin{equation}
    \sup_{\theta\in\Theta}\abs{M_n(\theta) -  M_n^\ast(\theta)}\overset{\Pr}{\rightarrow}0.
\end{equation}
\emph{Step (b)} By the same argument as that in step (a), we have
\begin{equation}
\label{eq:step-b}
\begin{aligned}
    \sup_{\theta\in\Theta}\abs{M_n^\ast(\theta) - M(\theta)} &\le  \Lambda \MMD_{k_X^2}(\P^{\mathrm{base}}_{X}, \P_{X}^{\infty}) + \MMD_{k}(\P^{\mathrm{base}}_{XY}, \P_{XY}^{\infty})
\end{aligned}
\end{equation}
By condition~\eqref{B4}, the RHS converges to zero in $\Pr_{\D,S}$-probability, so the LHS also converges to zero in $\Pr_{\D,S}$-probability (neither $M_n^\ast$ nor $M$ depends on the random DP realizations under $\Pr_{\DP}$). Combining \eqref{eq:step-a} and \eqref{eq:step-b} and using the triangle inequality, we have
\begin{equation}
\label{eq:M_n-uniform}
    \sup_{\theta\in\Theta}\abs{M_n(\theta) - M(\theta)} \overset{\Pr}{\rightarrow}0.
\end{equation}
By \citet[Theorem~2.1]{newey1994large}, we obtain $\hat\theta_n\xrightarrow{\Pr}\theta^{\dagger}$.
\end{proof}
\subsection{Proof of Proposition~\ref{prop:cond}}
\label{proof:prop-cond}
\begin{proof}[of Proposition~\ref{prop:cond}]
    We first state the following inequality on the $\MMD$. For any $0\le\alpha,\beta\le 1$ with $\alpha+\beta=1$ and probability measures $P_1,P_2,Q_1,Q_2$:
    \begin{equation}
    \label{eq:master-decomp}
        \begin{aligned}
            \MMD_k (\alpha P_1 + \beta P_2, \alpha Q_1 + \beta Q_2) &= \norm{\varphi(\alpha P_1 + \beta P_2) - \varphi(\alpha Q_1 + \beta Q_2)}_{\Hilb_k} 
            \\ &= \norm{\alpha\left[\varphi(P_1)-\varphi(Q_1)\right]+\beta\left[\varphi(P_2)-\varphi(Q_2)\right]}_{\Hilb_k} 
            \\&\le\alpha\norm{\varphi(P_1)-\varphi(Q_1)}_{\Hilb_k} + \beta \norm{\varphi(P_2)-\varphi(Q_2)}_{\Hilb_k}
            \\&=\alpha \MMD_k(P_1, Q_1) +\beta\MMD_k(P_2,Q_2)
        \end{aligned}
    \end{equation}
Now we show part~\ref{prop:1.1}.

For equation~\eqref{eq:pseudo-with-prior} with non-vanishing prior effect, we recall         
\begin{equation*}
    \P_{XY}^{\infty} = \frac{c}{c+m}\Q_{XY}^\infty+\frac{m}{c+m}\Pi_{\theta^\ast}^{XY}, \qquad  \P_{X}^{\infty} = \frac{c}{c+m}\Q_{X}^\infty+\frac{m}{c+m}\Pi_{\theta^\ast}^{X}.
\end{equation*}
By \eqref{eq:master-decomp} we can decompose
\begin{equation}
\label{eq:base-decomp}
    \MMD_{k}(\P^{\mathrm{base}}_{XY},\P_{XY}^\infty)\le \frac{c}{c+m} \MMD_k\left(\frac{1}{n}\sum_{i=1}^n \Q_{XY,i}, \Q_{XY}^\infty\right) + \frac{m}{c+m}\MMD_{k}\left(\P^{\mathrm{pseudo}}_{XY}, \Pi_{\theta^\ast}^{XY}\right)   
\end{equation}
The first term converges to zero in probability by condition \ref{prop:D}. For the second term, we recall in the proof of Theorem~\ref{thm:pseudo-classical} for the classical ME model (or Proposition~\ref{cor:pseudo-berkson} for Berkson ME) that, by term A~\eqref{eq:term-A} and term B~\eqref{eq:term-B}, we have, for any $M_n\to\infty$,
\begin{equation}
    \E_{\D,S}\left[\MMD_k\left(\P^{\mathrm{pseudo}}_{XY}, \Pi_{\theta^\ast}^{XY}\right)\right] \le \frac{4\sqrt{\kappa}}{\sqrt{n}}+\frac{M_n}{\sqrt{n}}+\sqrt{\kappa}\,r_n
\end{equation}
where $r_n\to 0$ as $n\to\infty$. Taking $M_n = \log n$, we have $\E_{\D, S}\left[\MMD_k\left(\P^{\mathrm{pseudo}}_{XY}, \Pi_{\theta^\ast}^{XY}\right)\right]\to0$ as $n\to \infty$.
By Markov's inequality, we have
\begin{equation}
    \Pr\left(\MMD_k\left(\P^{\mathrm{pseudo}}_{XY}, \Pi_{\theta^\ast}^{XY}\right)>\varepsilon\right) \le \frac{\E_{\D, S}\left[\MMD_k\left(\P^{\mathrm{pseudo}}_{XY}, \Pi_{\theta^\ast}^{XY}\right)\right]}{\varepsilon}.
\end{equation}
The RHS converges to zero as $n\to\infty$, which proves $\MMD_k\left(\P^{\mathrm{pseudo}}_{XY}, \Pi_{\theta^\ast}^{XY}\right)\overset{\Pr}{\rightarrow}0$. Substituting this into \eqref{eq:base-decomp} gives 
\begin{equation*}
    \MMD_k\left(\P^{\mathrm{base}}_{XY},\P_{XY}^\infty\right)\overset{\Pr}{\rightarrow}0 \qquad \text{as }n\to\infty.
\end{equation*}
The same argument applied to $\P^{\mathrm{base}}_{X}$ and $\Pi_{\theta^\ast}^X$ combined with Proposition~\ref{cor:pseudo-marginalX-classical} (classical) or Proposition~\ref{cor:pseudo-berkson} (Berkson) shows
\begin{equation*}
    \MMD_{k_X^2}\left(\P^{\mathrm{base}}_{X}, \P_{X}^\infty\right)\overset{\Pr}{\rightarrow}0.
\end{equation*}

When $c/m\to 0$ as $n\to\infty$, we can decompose
\begin{equation}
\begin{aligned}
    \MMD_{k}(\P^{\mathrm{base}}_{XY},\Pi_{\theta^\ast}^{XY}) &\le \frac{c}{c+m} \MMD_k\left(\frac{1}{n}\sum_{i=1}^n \Q_{XY,i}, \Pi_{\theta^\ast}^{XY}\right) + \frac{m}{c+m}\MMD_{k}\left(\P^{\mathrm{pseudo}}_{XY}, \Pi_{\theta^\ast}^{XY}\right)   
    \\&\le\underbrace{\frac{2\kappa c}{c+m}}_{\to 0\text{ as }m/c\to\infty}+\MMD_{k}\left(\P^{\mathrm{pseudo}}_{XY}, \Pi_{\theta^\ast}^{XY}\right)  
\end{aligned}
\end{equation}
The convergence of the second term gives $ \MMD_{k}(\P^{\mathrm{base}}_{XY},\Pi_{\theta^\ast}^{XY}) \overset{\Pr}{\rightarrow} 0$. Similarly, we have $ \MMD_{k_X^2}(\P^{\mathrm{base}}_{X},\Pi_{\theta^\ast}^{X}) \overset{\Pr}{\rightarrow} 0$.

Next, we show \ref{prop:1.2}. Since $(W_i,Y_i)\overset{i.i.d.}{\sim} \P_{WY}^0$, we have, by \cite[Lemma S6]{alquier2024universal},
\begin{equation}
    \E_{\D}\left[\MMD_{k}\left(\frac{1}{n}\sum_{i=1}^n\delta_{\left(W_i,Y_i\right)},\P_{WY}^0\right)\right] \le \frac{\sqrt{\kappa}}{\sqrt{n}}.
\end{equation}
Again, by Markov's inequality, 
\begin{equation}
\label{eq:nopseudo-convergence}
\MMD_{k}\left(\frac{1}{n}\sum_{i=1}^n\delta_{\left(W_i,Y_i\right)},\P_{WY}^0\right) \overset{\Pr}{\rightarrow}0
\end{equation}
By the same decomposition as in the proof of part~\ref{prop:1.1}, \eqref{eq:nopseudo-with-prior} follows from Assumption~\ref{prop:D} and \eqref{eq:nopseudo-convergence}; the final statement follows from $c\to 0$ and \eqref{eq:nopseudo-convergence}.

\end{proof}

\section{Sufficient conditions for Assumptions~\ref{A1}--\ref{A2} and verifying them in two scenarios}\label{sec:appendix-A1}

Throughout, $\|\cdot\|$ denotes the Euclidean norm, and for a $p\times p$ matrix $A$ we write $\|A\|$ for the operator norm.

\subsection{Sufficient conditions for Assumptions~\ref{A1}--\ref{A2}}
Recall that the true DGP under classical ME is 
\begin{equation}
    W=X+N,\quad Y=g^0(X)+E,\quad X\sim \P_{X}^{0}, \quad N\sim F_{N}^{0}, \quad E\sim F_{E}^{0},\quad N,E\ind X,\; N\ind E. 
\end{equation}
The joint density of the observed pair \((W,Y)\) is $p^{0}_{WY}(w,y)=\int_{\mathcal X}p^{0}_{X}(x)f^{0}_{N}(w-x)f_E^0\bigl(y-g^0(x)\bigr)\dd x$.

We work under a misspecified model
\[W=X+N,\quad Y=g(X, \theta)+E,\quad X\sim \P_{X}, \quad N\sim F_{N}, \quad E\sim F_{E},\quad N,E\ind X,\; N\ind E, \] 
and \(g^0(\cdot)\notin\{g(\cdot, \theta):\theta\in\Theta\}\); here \(\Theta\subset\R^{p}\) is compact.  For each \(\theta\) the induced density of \((W,Y)\) is \(p^{\theta}_{WY}(w,y)=\int_{\mathcal X}p_{X}(x)f_{N}(w-x)f_{E}\bigl(y-g(x, \theta)\bigr)\dd x.\)

Define the pseudo-true parameter by $\theta^{\ast}:=\arg\min_{\theta\in\Theta}\KL\bigl(p^{0}_{WY}\Vert p^{\theta}_{WY}\bigr)$. 

Write $\ell_\theta(W,Y):=\log p^{\theta}_{WY}(W,Y)$. For i.i.d.\
$(W_i,Y_i)\sim p_{WY}^{0}$ we collect and list sufficient conditions for Assumptions~\ref{A1}--\ref{A2}:

\begin{enumerate}[label=\textbf{Con.\arabic*}]\label{sec:checklist}
\item\label{cond:MC-app} \emph{Likelihood-ratio integrability:}
$\E_{p_{WY}^{0}}\big[p^{\theta}_{WY}/p^{\theta^\ast}_{WY}\big]<\infty$
for all $\theta\in\Theta$.

\item\label{cond:B1-app} \emph{Differentiability in probability:}
$\theta\mapsto\ell_\theta(W_1,Y_1)$ differentiable at $\theta^{\ast}$
in $p_{WY}^{0}$-probability.

\item\label{cond:B2-app} \emph{Local Lipschitz envelope:} there exists
$m_{\theta^{\ast}}\in L^{2}(p_{WY}^{0})$ such that for $\theta_1,\theta_2$
near $\theta^{\ast}$,
$\bigl|\ell_{\theta_1}-\ell_{\theta_2}\bigr|
\le m_{\theta^{\ast}}\|\theta_1-\theta_2\|$.

\item\label{cond:B3-app} \emph{Quadratic KL expansion:}
$-\E_{p_{WY}^{0}}\left[\log(p^{\theta}_{WY}/p^{\theta^\ast}_{WY})\right]
=\tfrac12(\theta-\theta^{\ast})^{\T}V_{\theta^{\ast}}(\theta-\theta^{\ast})
+o(\|\theta-\theta^{\ast}\|^{2})$ with $V_{\theta^{\ast}}\succ0$.

\item\label{cond:B5-app} \emph{$L^{1}(p_{WY}^{0})$ continuity:} for every fixed $\theta_0\in\Theta$, the map $\theta\mapsto p^{\theta}_{WY}/p^{\theta_0}_{WY}$ is $L^{1}(p_{WY}^{0})$-continuous at every $\theta\in\Theta$.

\item\label{cond:B6-app} \emph{Exponential moment of the envelope:}
$\exists s>0$ with $\E_{p_{WY}^{0}}[e^{s m_{\theta^{\ast}}}]<\infty$.

\item\label{cond:ScorePD-app} \emph{Non-singular score covariance:}
$S_{\theta^{\ast}}:=\E_{p_{WY}^{0}}\big[\dot\ell_{\theta^{\ast}}\dot\ell_{\theta^{\ast}}^{\T}\big]$
is invertible.

\item\label{cond:B4-app} \emph{Compact $\Theta$ and uniqueness:}
$\Theta$ compact and $\theta^{\ast}\in\mathrm{int}(\Theta)$ the unique minimizer of
$\theta\mapsto -\E_{p_{WY}^{0}}\log p^{\theta}_{WY}$.

\item\label{cond:Prior-appendix}
The prior $\Pi$ on $\Theta$ admits a density $\pi$ with respect to
Lebesgue measure that is continuous and strictly positive on a neighbourhood of $\theta^{\ast}$.
\item\label{cond:MMD-Lip-app} \emph{Posterior Lipschitz in MMD:}
\[
\MMD_k\left(\Pi_{\theta_1}(\cdot\mid w,y),\Pi_{\theta_2}(\cdot\mid w,y)\right)
\le L(w,y)\,\|\theta_1-\theta_2\|,\qquad L(W,Y)\in L^{2}(p_{WY}^{0}),
\]
for $\theta_1,\theta_2$ in a neighbourhood $\Theta_\rho$ of $\theta^\ast$.
\end{enumerate}

Assumption~\ref{A1} invokes the local asymptotic normality (LAN), smoothness, integrability and regularity requirements of \citet[Theorem~3.1]{kleijn2012bernstein} for $\{p^{\theta}_{WY}:\theta\in\Theta\}$ around $\theta^{\ast}$. We now explain why the conditions stated above are sufficient:

\begin{enumerate}
\item \emph{LAN via Lemma~2.1.}
Our \ref{cond:B1-app} (differentiability in probability), \ref{cond:B2-app} (local Lipschitz envelope) and \ref{cond:B3-app} (quadratic KL expansion with $V_{\theta^\ast}\succ0$) yield LAN with $\delta_n=n^{-1/2}$ and central sequence $V_{\theta^\ast}^{-1}\mathbb G_n\dot\ell_{\theta^\ast}$. 

\item \emph{Existence of tests via Theorem~3.2.}
Compactness and uniqueness (\ref{cond:B4-app}) together with the $L^1$-continuity of likelihood ratios (\ref{cond:B5-app}) satisfy the sufficient conditions in \citet[Theorem~3.2]{kleijn2012bernstein}, guaranteeing tests $(\phi_n)$ with the properties required in \citet[Theorem~3.1]{kleijn2012bernstein}.

\item \emph{Moment and prior conditions for Theorem~3.1.}
Our \ref{cond:MC-app} ensures $\E_{p_{WY}^{0}}[p_\theta/p_{\theta^\ast}]<\infty$ for all $\theta\in\Theta$. Condition~\ref{cond:B6-app} provides $\E_{p_{WY}^{0}}[e^{s m_{\theta^\ast}}]<\infty$ for some $s>0$.
Condition~\ref{cond:Prior-appendix} supplies a prior density continuous and strictly positive near $\theta^\ast$. Invertibility of $\E_{p_{WY}^{0}}[\dot\ell_{\theta^\ast}\dot\ell_{\theta^\ast}^{\T}]$
is \ref{cond:ScorePD-app}.
\end{enumerate}

Consequently, Theorem~3.1 of \citet{kleijn2012bernstein} applies under conditions~\ref{cond:MC-app}-\ref{cond:Prior-appendix}. Condition~\ref{cond:MMD-Lip-app} implies Assumption~\ref{A2}. Conditions~\ref{cond:MC-app}--\ref{cond:Prior-appendix} are a compilation of \emph{sufficient} conditions for Assumption~\ref{A1}, and many of them can be relaxed depending on the specific model and true distribution. We refer the reader to \citet{kleijn2012bernstein} for a detailed discussion on weaker forms of these conditions and circumstances where they can be relaxed.

\subsection{Two scenarios where Assumptions~\ref{A1}--\ref{A2} hold}
Next, we demonstrate two scenarios under Gaussian noise models where conditions~\ref{cond:MC-app}--\ref{cond:MMD-Lip-app} are satisfied. In this section, ME and outcome noise are modelled as independent centred Gaussians
\(
  N\sim\mathcal N_{d}(0,\Sigma_{N}),
  E\sim\mathcal N(0,\sigma_{E}^{2}),
\)
with known $\Sigma_{N}\succ0$, $\sigma_{E}^{2}>0$. A working prior density for $X$ is $p_{X}$. Write
\[
  \varphi_{\Sigma}(z):=(2\pi)^{-d/2}(\det\Sigma)^{-1/2}
  e^{-\frac12 z^{\T}\Sigma^{-1}z},\qquad
  \varphi_{\sigma}(t):=(2\pi\sigma^{2})^{-1/2}e^{-t^{2}/(2\sigma^{2})}.
\]
Then
\begin{equation}\label{eq:working-density-app}
  p^{\theta}_{WY}(w,y)=\int_{\R^{d}}
      p_{X}(x)\varphi_{\Sigma_{N}}(w-x)
      \varphi_{\sigma_{E}}\bigl(y-g(x, \theta)\bigr)\dd x.
\end{equation}
We collect here the standing assumptions used throughout. They are invoked in both scenarios below.

\begin{enumerate}[label=(M\arabic*)]

\item\label{ass:G4-appendix} \emph{Local smoothness and curvature near $\theta^\ast$:}
there exists a neighbourhood $U$ of $\theta^\ast$ such that
$g(\cdot, \theta)$ is $C^{2}$ in $\theta$ on $U$; for $\theta\in U$, $\dot\ell_\theta$ and $\ddot\ell_\theta$ admit an $L^{1}(p_{WY}^{0})$ envelope uniform in $\theta\in U$; and
$H(\theta):=-\E_{p_{WY}^{0}}\big[\partial^{2}_{\theta}\ell_{\theta}(W,Y)\big]$ exists for $\theta\in U$, is continuous at $\theta^\ast$, and $H^{\ast}:=H(\theta^{\ast})\succ0$.

\item\label{ass:G5-appendix} \emph{Outcome tail:}
$\E\big[e^{\tau_0 Y^{2}}\big]<\infty$ for some $\tau_0>0$.

\item\label{ass:G3-appendix} \emph{Information non-singularity at $\theta^\ast$:}
$S_{\theta^\ast}:=\E_{p_{WY}^{0}}\big[\dot\ell_{\theta^\ast}\dot\ell_{\theta^\ast}^{\T}\big]$ is invertible. This is \ref{cond:ScorePD-app}.

\item\label{ass:G6-appendix} \emph{Compactness and uniqueness:}
$\Theta$ is compact and $\theta^{\ast}\in\mathrm{int}(\Theta)$ is the unique minimizer of $R(\theta)=-\E_{p_{WY}^{0}}[\ell_{\theta}(W,Y)]$. This is \ref{cond:B4-app}.

\item\label{ass:Prior-appendix} \emph{Continuous and positive prior:}
The prior $\Pi$ on $\Theta$ admits a density $\pi$ with respect to
Lebesgue measure that is continuous and strictly positive on a neighbourhood of $\theta^{\ast}$: this is \ref{cond:Prior-appendix}.
\end{enumerate}

For smooth finite-dimensional nonlinear regressions, these requirements are mild: 
\ref{ass:G4-appendix} is a local $C^2$ and integrability condition that justifies differentiation under the integral and yields positive-definite local curvature $H^\ast$. \ref{ass:G5-appendix} is used for algebraic convenience to control exponential envelopes; it can be weakened to a sub-exponential tail (e.g. $\E[e^{\tau|Y|}]<\infty$) while adjusting the envelope arguments. \ref{ass:G3-appendix} is standard: finiteness of $S_{\theta^\ast}$ follows in both scenarios below from the score envelopes and tail conditions, so it effectively asks only for non-degeneracy of the score covariance at $\theta^\ast$. \ref{ass:G6-appendix} is the standard well-separated pseudo-true parameter assumption. The prior condition \ref{ass:Prior-appendix} is the usual prior thickness requirement.

We consider two concrete scenarios that ensure all conditions~\ref{cond:MC-app}--\ref{cond:MMD-Lip-app} are satisfied. In both cases we assume \ref{ass:G4-appendix}--\ref{ass:Prior-appendix}.
\noindent\textit{Scenario 1 (bounded regression).}\label{ass:S1-appendix}
This covers models where the regression surface and its first two derivatives are uniformly bounded over $\Theta\times\R^d$.
\begin{enumerate}[label=(S1.\arabic*)]
\item\label{ass:S1.1-appendix} $C_g:=\sup_{\theta\in\Theta,x\in\R^{d}}|g(x, \theta)|<\infty$.
\item\label{ass:S1.2-appendix} There exist finite constants $C_{\partial g},C_{\partial^{2} g}$ such that for all $(\theta,x)\in\Theta\times\mathcal{X}$,
\(
\|\partial_\theta g(x, \theta)\|\le C_{\partial g},
\|\partial^2_{\theta\theta} g(x, \theta)\|\le C_{\partial^{2} g}.
\)
\end{enumerate}
\noindent\textit{Scenario 2 (compact latent support and working prior).}\label{ass:S2-appendix}
This covers settings where $X$ is confined to a compact region and the working prior respects that support; boundedness of $g$ and its derivatives is then only needed on that region.
\begin{enumerate}[label=(S2.\arabic*)]
\item\label{ass:S2.1-appendix} $\operatorname{supp}p_X\subseteq B_M:=\{x:\|x\|\le M\}$.
\item\label{ass:S2.2-appendix} $\|X\|\le M$ $\P_{X}^{0}$-a.s.
\item\label{ass:S2.3-appendix}
$\widetilde C_{\partial g}:=\sup_{\theta\in\Theta,x\in B_M}\|\partial_\theta g(x, \theta)\|<\infty$ and
$\widetilde C_{\partial^{2} g}:=\sup_{\theta\in\Theta,x\in B_M}\|\partial^2_{\theta\theta} g(x, \theta)\|<\infty$.
\item\label{ass:S2.4-appendix}
$g$ is continuous on $\Theta\times B_M$, hence
$C_{g,M}:=\sup_{\theta\in\Theta,\|x\|\le M}|g(x, \theta)|<\infty$.
\end{enumerate}

Now we verify \ref{cond:MC-app}--\ref{cond:MMD-Lip-app} for both scenarios. Since \ref{cond:ScorePD-app}-\ref{cond:Prior-appendix} are already assumed by \ref{ass:G3-appendix}-\ref{ass:Prior-appendix}, we only need to verify \ref{cond:MC-app}-\ref{cond:B6-app}, and \ref{cond:MMD-Lip-app}.

\subsection{Verification of~\ref{cond:MC-app} likelihood-ratio integrability}\label{subsec:MC-app}

\emph{Goal:} $\E_{p_{WY}^{0}}\big[p^{\theta}_{WY}/p^{\theta^\ast}_{WY}\big]<\infty$ for all $\theta\in\Theta$.

For $R_\theta(W,Y):=p_{WY}^{\theta}(W,Y)/p_{WY}^{\theta^{\ast}}(W,Y)$ write
$a(x):=p_{X}(x)\varphi_{\Sigma_N}(W-x)$ and
$b_\theta(x):=\varphi_{\sigma_E}(Y-g(x, \theta))$. Then
\[
R_\theta=\frac{\int a(x)b_\theta(x)\dd x}{\int a(x)b_{\theta^{\ast}}(x)\dd x}
\le \sup_{x}\frac{b_\theta(x)}{b_{\theta^{\ast}}(x)}.
\]
Since $b_\theta(x)=c\exp\{-(Y-g(x, \theta))^{2}/(2\sigma_E^{2})\}$,
\[
\log\frac{b_\theta(x)}{b_{\theta^{\ast}}(x)}
=-\frac{(Y-g(x, \theta))^{2}-(Y-g(x,\theta^{\ast}))^{2}}{2\sigma_E^{2}}
=\frac{(g(x, \theta)-g(x,\theta^{\ast}))}{\sigma_E^{2}}Y
+\frac{g(x,\theta^{\ast})^{2}-g(x, \theta)^{2}}{2\sigma_E^{2}}.
\]

\noindent\textit{Scenario 1.}
Using $|g(x, \theta)-g(x,\theta^{\ast})|\le 2C_g$ and
$|g(x,\theta^{\ast})^{2}-g(x, \theta)^{2}|\le 4C_g^{2}$,
\[
R_\theta\le
\exp\Bigl\{\frac{2C_g}{\sigma_E^{2}}|Y|+\frac{2C_g^{2}}{\sigma_E^{2}}\Bigr\}.
\]
By the bound $e^{a|Y|}\le e^{a^{2}/(4\varepsilon)}e^{\varepsilon Y^{2}}$
(valid for all $a,\varepsilon>0$) and \ref{ass:G5-appendix}
(with any $\varepsilon<\tau_0$), $\E R_\theta<\infty$.

\noindent\textit{Scenario 2.}
Because $\operatorname{supp}p_X\subseteq B_M$ and $\|X\|\le M$ a.s., the integrals in \eqref{eq:working-density-app} are over $B_M$. With
$\Delta_\theta(M):=\sup_{\|x\|\le M}|g(x, \theta)-g(x,\theta^{\ast})|$,
\[
R_\theta\le
\exp\Bigl\{\frac{\Delta_\theta(M)}{\sigma_E^{2}}|Y|
+\frac{C_{g,M}\Delta_\theta(M)}{\sigma_E^{2}}\Bigr\}.
\]
Since $\Delta_\theta(M)\le 2C_{g,M}$ by \ref{ass:S2.4-appendix}, \ref{ass:G5-appendix} implies $\E R_\theta<\infty$.

\subsection{Verification of~\ref{cond:B1-app} differentiability}\label{subsec:B1-app}
\emph{Goal:} $\theta\mapsto\ell_\theta(W,Y)$ differentiable at $\theta^{\ast}$ in $p_{WY}^{0}$-probability. 

For $(w,y)\in\R^{d}\times\R$,
\[
p_\theta(w,y)=\int p_X(x)\varphi_{\Sigma_N}(w-x)
\varphi_{\sigma_E}\bigl(y-g(x, \theta)\bigr)\dd x.
\]
Under \ref{ass:S1.2-appendix} or \ref{ass:S2.3-appendix} there exists an integrable envelope $\Gamma$ (constant in Scenario~1, bounded on $B_M$ in Scenario~2) such that $\|\partial_\theta g(x, \theta)\|\le\Gamma(x)$ and $\int p_X(x)\Gamma(x)\dd x<\infty$. Differentiation under the integral (Leibniz rule) gives
\[
\nabla_\theta p_\theta(w,y)=\int p_X(x)\varphi_{\Sigma_N}(w-x)
\nabla_\theta\varphi_{\sigma_E}\bigl(y-g(x, \theta)\bigr)\dd x,
\]
and
\[
\nabla_\theta\varphi_{\sigma_E}(y-g(x, \theta))
=\varphi_{\sigma_E}(y-g(x, \theta))
\frac{y-g(x, \theta)}{\sigma_E^{2}}
\nabla_\theta g(x, \theta).
\]
Hence
\[
\nabla_\theta p_\theta(w,y)=\frac1{\sigma_E^{2}}
\int p_X(x)\varphi_{\Sigma_N}(w-x)\varphi_{\sigma_E}(y-g(x, \theta))
(y-g(x, \theta))\nabla_\theta g(x, \theta)\dd x.
\]
Introduce the $\theta$-posterior density
\[
q_\theta(x\mid w,y):=
\frac{p_X(x)\varphi_{\Sigma_N}(w-x)\varphi_{\sigma_E}(y-g(x, \theta))}
{p_\theta(w,y)} .
\]
Dividing by $p_\theta(w,y)$ yields the score
\begin{equation}\label{eq:def-score-app}
\dot\ell_\theta(W,Y)=\frac1{\sigma_E^{2}}
\E_{q_{\theta}(\cdot\mid W,Y)}\Bigl[(Y-g(X, \theta))\partial_\theta g(X, \theta)\mid W,Y\Bigr].
\end{equation}
Let $f_\theta(x):=p_X(x)\varphi_{\Sigma_N}(W-x)\varphi_{\sigma_E}(Y-g(x, \theta))$,
so $\ell_\theta(W,Y)=\log\int f_\theta(x)\dd x$. For $h\in\R^{p}$ and
$\theta_t:=\theta+th$,
\[
\ell_{\theta+h}-\ell_{\theta}
=\int_{0}^{1}
\frac{\int \langle\partial_\theta f_{\theta_t}(x),h\rangle \dd x}
{\int f_{\theta_t}(x) \dd x}\dd t
=\int_{0}^{1}\langle\dot\ell_{\theta_t},h\rangle \dd t.
\]
A second differentiation gives
\begin{equation}\label{eq:def-hess-app}
\begin{aligned}
\ddot\ell_{\theta}(W,Y)
&:=\partial_{\theta\theta}^{2}\ell_\theta(W,Y)\\
&=\E_{q_{\theta}(\cdot\mid W,Y)}\Bigl[
-\tfrac{1}{\sigma_E^{2}}\partial_\theta g\partial_\theta g^{\T}
+\tfrac{Y-g}{\sigma_E^{2}}\partial_{\theta\theta}^{2}g\ \Bigm|\ W,Y\Bigr]
+\Var_{q_{\theta}(\cdot\mid W,Y)}\Bigl( \tfrac{Y-g}{\sigma_E^{2}}\partial_\theta g\ \Bigm|\ W,Y\Bigr).
\end{aligned}
\end{equation}
Therefore
\[
\ell_{\theta+h}-\ell_{\theta}
=\langle\dot\ell_{\theta},h\rangle
+\int_{0}^{1}(1-t)\langle\ddot\ell_{\theta_t}h,h\rangle \dd t,
\]
and hence
\begin{equation}\label{eq:rem-bnd-app}
\frac{|\ell_{\theta+h}-\ell_{\theta}-\langle\dot\ell_{\theta},h\rangle|}
{\|h\|}\le \frac{\|h\|}{2}\sup_{t\in[0,1]}\|\ddot\ell_{\theta_t}(W,Y)\|.
\end{equation}
From \eqref{eq:def-hess-app} and \ref{ass:S1.2-appendix} (Scenario~1) or
\ref{ass:S2.3-appendix}-\ref{ass:S2.4-appendix} (Scenario~2),
there exist finite $B_0,B_1,B_2$ (scenario-dependent, $\theta$-uniform) such that
\[
\|\ddot\ell_{\theta}(W,Y)\|
\le \frac{B_0+B_1|Y|}{\sigma_E^{2}}
+\frac{B_2(|Y|+C)^{2}}{\sigma_E^{4}},
\qquad C=C_g\ \text{or}\ C_{g,M}.
\]
Thus $\|\ddot\ell_{\theta}(W,Y)\|\le a_0+a_1|Y|+a_2 Y^{2}$ for some finite $(a_0,a_1,a_2)$, which is integrable by \ref{ass:G5-appendix}. Denote $M(W,Y):=a_0+a_1|Y|+a_2 Y^{2}\in L^{1}(p_{WY}^{0})$.

Fix $\theta=\theta^{\ast}$. For every $\varepsilon>0$,
\[
p_{WY}^{0}\Bigl(
\frac{|\ell_{\theta^{\ast}+h}-\ell_{\theta^{\ast}}-\langle\dot\ell_{\theta^{\ast}},h\rangle|}
{\|h\|}>\varepsilon\Bigr)
\le p_{WY}^{0}\bigl(\tfrac{\|h\|}{2}M(W,Y)>\varepsilon\bigr)\xrightarrow[h\to0]{}0.
\]
Hence \ref{cond:B1-app} holds. The bounds also give
$\dot\ell_{\theta^{\ast}}\in L^{2}(p_{WY}^{0})$, used below. The same domination shows differentiability in $p_{WY}^{0}$-probability holds uniformly along compact line segments in $\Theta$, as used in Section~\ref{subsec:B5-app}.

\subsection{Verification of~\ref{cond:B2-app} local Lipschitz envelope}\label{subsec:B2-app}
By the fundamental theorem of calculus,
\[
\ell_{\theta_1}-\ell_{\theta_2}
=\int_{0}^{1}\bigl\langle\dot\ell_{\theta_t},\theta_1-\theta_2\bigr\rangle \dd t
\]
and hence
\[
|\ell_{\theta_1}-\ell_{\theta_2}|
\le \|\theta_1-\theta_2\| \sup_{t\in[0,1]}\|\dot\ell_{\theta_t}\|.
\]
From \eqref{eq:def-score-app} and \ref{ass:S1.2-appendix} (Scenario~1),
\[
\|\dot\ell_\theta\|\le \frac{C_{\partial g}}{\sigma_E^{2}}(|Y|+C_g)
=:m_{\theta^\ast}^{(1)}(W,Y),
\]
and from \ref{ass:S2.3-appendix}-\ref{ass:S2.4-appendix} (Scenario~2),
\[
\|\dot\ell_\theta\|\le \frac{\widetilde C_{\partial g}}{\sigma_E^{2}}(|Y|+C_{g,M})
=:m_{\theta^\ast}^{(2)}(W,Y).
\]
By \ref{ass:G5-appendix}, $m_{\theta^\ast}^{(j)}\in L^{2}(p_{WY}^{0})$ ($j=1,2$), proving \ref{cond:B2-app}.

\subsection{Verification of~\ref{cond:B3-app} curvature of the KL risk}\label{subsec:B3-app}
Let $H^{\ast}:=-\E_{p_{WY}^{0}}[\ddot\ell_{\theta^{\ast}}]$, which is finite by the moment bounds used in \ref{subsec:B1-app}. Differentiation under the integral (per \ref{ass:G4-appendix}) yields
\[
\nabla R(\theta)=-\E_{p_{WY}^{0}}[\dot\ell_{\theta}(W,Y)].
\]
Since $\theta^{\ast}$ minimizes $R$ and $\theta^{\ast}\in\operatorname{int}(\Theta)$, the first-order condition gives $\E_{p_{WY}^{0}}[\dot\ell_{\theta^{\ast}}]=0$. A Taylor expansion of $R$ at $\theta^{\ast}$ then implies
\[
-\E_{p_{WY}^{0}}\bigl[\ell_\theta-\ell_{\theta^{\ast}}\bigr]
=\tfrac12(\theta-\theta^{\ast})^{\T} H^{\ast}(\theta-\theta^{\ast})
+o(\|\theta-\theta^{\ast}\|^{2}),
\]
so $V_{\theta^{\ast}}=H^{\ast}$. By \ref{ass:G4-appendix}, $H^{\ast}\succ0$, proving \ref{cond:B3-app}.

\subsection{Verification of~\ref{cond:B5-app} \texorpdfstring{$L^{1}$}{L1}-continuity of the likelihood ratio}\label{subsec:B5-app}
Fix $\theta_{0}\in\Theta$ and let $\theta,\theta_1\in\Theta$ be arbitrary. Set $h:=\theta-\theta_{1}$, $\theta_t:=\theta_{1}+t h$. Define
\[
r_{t}(W,Y):=\frac{p_{WY}^{\theta_{t}}(W,Y)}{p_{WY}^{\theta_{0}}(W,Y)}
=\exp\bigl\{\ell_{\theta_{t}}(W,Y)-\ell_{\theta_{0}}(W,Y)\bigr\}.
\]
Since $t\mapsto r_t$ is absolutely continuous and
$\partial_{t}r_t=r_t\langle h,\dot\ell_{\theta_t}\rangle$, we obtain the identity
\begin{equation}\label{eq:B5-ftc-app}
r_{\theta,\theta_{0}}(W,Y)-r_{\theta_{1},\theta_{0}}(W,Y)
=\int_{0}^{1}\langle h,\dot\ell_{\theta_t}(W,Y)\rangle\, r_t(W,Y)\dd t.
\end{equation}
Taking absolute values and expectations, by the triangle inequality, Tonelli and Cauchy-Schwarz,
\begin{align}
\norm{r_{\theta,\theta_{0}}-r_{\theta_{1},\theta_{0}}}_{L^{1}(p^{0}_{WY})}
&\le \norm{h}\int_{0}^{1}
\E_{p^{0}_{WY}}\left[\norm{\dot\ell_{\theta_t}} r_t\right]\dd t
\notag\\
&\le \norm{h}\int_{0}^{1}
\left(\E_{p^{0}_{WY}}\norm{\dot\ell_{\theta_t}}^{2}\right)^{1/2}
\left(\E_{p^{0}_{WY}} r_t^{2}\right)^{1/2} \dd t.
\label{eq:B5-L1bound-app}
\end{align}
Using \ref{ass:S1.2-appendix} or \ref{ass:S2.3-appendix}-\ref{ass:S2.4-appendix} together with \ref{ass:G5-appendix},
\[
M_{1}:=\sup_{\vartheta\in\Theta}\E_{p^{0}_{WY}}\|\dot\ell_{\vartheta}\|^{2}<\infty.
\]
From \ref{subsec:MC-app} we have the envelopes
\[
r_{\vartheta,\theta_0}\le \exp\Bigl\{\frac{2C_g}{\sigma_E^{2}}|Y|
+\frac{2C_g^{2}}{\sigma_E^{2}}\Bigr\}
\quad\text{(Scenario 1)},\;
r_{\vartheta,\theta_0}\le \exp\Bigl\{\frac{\Delta_{\vartheta,\theta_0}(M)}{\sigma_E^{2}}|Y|
+\frac{C_{g,M}\Delta_{\vartheta,\theta_0}(M)}{\sigma_E^{2}}\Bigr\}
\quad\text{(Scenario 2)},
\]
where $\Delta_{\vartheta,\theta_0}(M):=\sup_{\|x\|\le M}|g(x,\vartheta)-g(x,\theta_0)|\le 2C_{g,M}$. Hence by \ref{ass:G5-appendix},
\[
M_{2}:=\sup_{\vartheta\in\Theta}\E r_{\vartheta,\theta_0}^{2}<\infty.
\]
Therefore \eqref{eq:B5-L1bound-app} gives
\[
\bigl\|r_{\theta,\theta_{0}}-r_{\theta_1,\theta_{0}}\bigr\|_{L^{1}(P^{0}_{WY})}
\le \|h\|\sqrt{M_{1}M_{2}}\xrightarrow[\theta\to\theta_{1}]{}0,
\]
proving \ref{cond:B5-app}.

\subsection{Verification of~\ref{cond:B6-app} exponential moment}\label{subsec:B6-app}
Recall
\[
m_{\theta^{\ast}}=
\begin{cases}
\dfrac{C_{\partial g}}{\sigma_E^{2}}(|Y|+C_g), & \text{Scenario 1},\\[4pt]
\dfrac{\widetilde C_{\partial g}}{\sigma_E^{2}}(|Y|+C_{g,M}), & \text{Scenario 2}.
\end{cases}
\]
In Scenario~1, for any $s>0$,
\[
\E e^{s m_{\theta^{\ast}}}
=e^{s C_{\partial g}C_g/\sigma_E^{2}}
\E \exp\Bigl\{\frac{s C_{\partial g}}{\sigma_E^{2}}|Y|\Bigr\}
\le e^{s C_{\partial g}C_g/\sigma_E^{2}}
e^{\frac{(s C_{\partial g}/\sigma_E^{2})^{2}}{4\varepsilon}}
\E e^{\varepsilon Y^{2}}<\infty
\]
for any $\varepsilon<\tau_0$ by \ref{ass:G5-appendix}. Thus \ref{cond:B6-app} holds for all $s>0$. The same argument applies in Scenario~2:
\[
\E e^{s m_{\theta^{\ast}}}
\le \exp\Bigl\{\frac{s \widetilde C_{\partial g} C_{g,M}}{\sigma_E^{2}}
+\frac{(s \widetilde C_{\partial g}/\sigma_E^{2})^{2}}{4\varepsilon}\Bigr\}
\E e^{\varepsilon Y^{2}}<\infty,
\]
for any $s>0$ and $\varepsilon<\tau_0$. Hence \ref{cond:B6-app} holds.

\subsection{Verification of~\ref{cond:MMD-Lip-app} posterior Lipschitz in total variation}\label{subsec:TVD-Lip-app}
\begin{lemma}[Pathwise total-variation bound]\label{lem:TV-path-app}
Fix $(w,y)\in\R^{d}\times\R$ and define
\[
f_\theta(x):=p_X(x)\varphi_{\Sigma_N}(w-x)\varphi_{\sigma_{E}}\bigl(y-g(x, \theta)\bigr),\quad
Z_\theta:=\int f_\theta(u)du,\quad
\pi_\theta(x):=\frac{f_\theta(x)}{Z_\theta}.
\]
For $\theta_t:=\theta_2+t(\theta_1-\theta_2)$,
\begin{equation}\label{eq:TVD-path-final-app}
\bigl\|\Pi_{\theta_1}(\cdot\mid w,y)-\Pi_{\theta_2}(\cdot\mid w,y)\bigr\|_{\TV}
\le \|\theta_1-\theta_2\|
\int_0^1 \E_{\Pi_{\theta_t}^{w,y}}
\bigl[|y-g(X,\theta_t)|\|\partial_\theta g(X,\theta_t)\|\bigr]
\frac{\dd t}{\sigma_E^{2}}.
\end{equation}
\end{lemma}
\begin{proof}
Total variation is $\|P-Q\|_{\TV}=\tfrac12\int|p-q|\dd x$. With $\pi_t:=\pi_{\theta_t}$,
\[
\|\Pi_{\theta_1}(\cdot\mid w,y)-\Pi_{\theta_2}(\cdot\mid w,y)\|_{\TV}
=\tfrac12\int\Bigl|\int_0^1\partial_t\pi_t(x)\dd t\Bigr|\dd x
\le \tfrac12\int_0^1\int |\partial_t\pi_t(x)|\dd x \dd t,
\]
by the triangle inequality and Tonelli. Since $\partial_t\pi_t=(\partial_\theta\pi_\theta)|_{\theta=\theta_t}(\theta_1-\theta_2)$,
\[
|\partial_t\pi_t(x)|\le \|\theta_1-\theta_2\|\|\partial_\theta\pi_{\theta_t}(x)\|.
\]
Using $\pi_\theta=f_\theta/Z_\theta$,
\[
\partial_\theta\pi_\theta(x)
=\pi_\theta(x)\Bigl(\partial_\theta\log f_\theta(x)
-\E_{\Pi_{\theta}(\cdot\mid w,y)}[\partial_\theta\log f_\theta(X)]\Bigr).
\]
Hence
\[
\int\|\partial_\theta\pi_\theta(x)\|\dd x
\le 2\E_{\Pi_{\theta}(\cdot\mid w,y)}\bigl[\|\partial_\theta\log f_\theta(X)\|\bigr].
\]
Cancelling the prefactor $\tfrac12$ gives
\[
\|\Pi_{\theta_1}(\cdot\mid w,y)-\Pi_{\theta_2}(\cdot\mid w,y)\|_{\TV}
\le \|\theta_1-\theta_2\|\int_0^1
\E_{\Pi_{\theta_t}^{w,y}}\bigl[\|\partial_\theta\log f_{\theta_t}(X)\|\bigr]\dd t.
\]
Finally,
$\partial_\theta\log f_\theta(x)=\tfrac{y-g(x, \theta)}{\sigma_E^{2}}\partial_\theta g(x, \theta)$,
which gives \eqref{eq:TVD-path-final-app}.
\end{proof}

Let $A_{1}:=C_{\partial g}/\sigma_E^{2}$ and
$A_{2}:=\widetilde C_{\partial g}/\sigma_E^{2}$. In Scenario~1,
\[
\E_{\Pi_{\theta}(\cdot\mid w,y)}\bigl[|y-g(X, \theta)|\|\partial_\theta g(X, \theta)\|\bigr]
\le C_{\partial g}(|y|+C_g),
\]
so
\[
\bigl\|\Pi_{\theta_1}(\cdot\mid w,y)-\Pi_{\theta_2}(\cdot\mid w,y)\bigr\|_{\TV}
\le A_{1}(|y|+C_g)\|\theta_1-\theta_2\|.
\]
In Scenario~2 (on $B_M$),
\[
\E_{\Pi_{\theta}(\cdot\mid w,y)}\bigl[|y-g(X, \theta)|\|\partial_\theta g(X, \theta)\|\bigr]
\le \widetilde C_{\partial g}(|y|+C_{g,M}),
\]
hence
\[
\bigl\|\Pi_{\theta_1}(\cdot\mid w,y)-\Pi_{\theta_2}(\cdot\mid w,y)\bigr\|_{\TV}
\le A_{2}(|y|+C_{g,M})\|\theta_1-\theta_2\|.
\]
By \ref{ass:G5-appendix}, $L(W,Y)\in L^{2}(p_{WY}^{0})$ for
\[
L(w,y):=
\begin{cases}
A_{1}(|y|+C_g), & \text{Scenario 1},\\[4pt]
A_{2}(|y|+C_{g,M}), & \text{Scenario 2}.
\end{cases}
\]
By Lemma~\ref{lem:MMD-TVD}, we have 
\[
\MMD_k\left(\Pi_{\theta_1}(\cdot\mid w,y),\Pi_{\theta_2}(\cdot\mid w,y)\right)
\le 2\sqrt{\kappa}\|\Pi_{\theta_1}(\cdot\mid w,y)-\Pi_{\theta_2}(\cdot\mid w,y)\|_{\TV}
\le 2\sqrt{\kappa} L(w,y)\|\theta_1-\theta_2\|,
\]
which proves \ref{cond:MMD-Lip-app}.

\section{Implementation and additional experiment set-up details}\label{app:exp-details}
We optimize all MMD objectives using automatic differentiation with the Adam optimizer \citep{kingma2014adam} as implemented in \texttt{JAX}. The squared MMD between two probability measures $\P$ and $\Q$ with kernel $k$ is approximated by the unbiased U-statistic \citep{gretton2012kernel} using independent samples $\{x_i\}_{i=1}^n \sim \P$ and $\{y_j\}_{j=1}^s \sim \Q$:
\[
\widehat{\MMD}^2_k(\P,\Q)
=\frac{1}{n(n-1)}\!\!\sum_{i\neq i'}\! k(x_i,x_{i'})
+\frac{1}{s(s-1)}\!\!\sum_{j\neq j'}\! k(y_j,y_{j'})
-\frac{2}{ns}\sum_{i=1}^{n}\sum_{j=1}^{s} k(x_i,y_j).
\]
In all experiments we set $s=n$ equal to the number of observations in the corresponding datasets.

For the Berkson ME experiments, the observed covariates $W$ are organized into $100$ distinct groups, each repeated $3$ times (group size $=3$). The $100$ distinct group values are drawn i.i.d.\ from $\mathcal{N}(0,2)$. This design reflects common Berkson settings where $W$ are pre-specified targets, categories, or group averages.

For the classical ME experiments, the latent covariates $X$ are i.i.d.\ $\mathcal{N}(0,3)$. We choose the classical-error variance of $X$ to be larger than the variance of $W$ in the Berkson setting so that the marginal scales of $X$ are comparable across the two regimes (recall that in Berkson error $\sigma_X^2=\sigma_W^2+\sigma_N^2$).

We use $B_{\mathrm{boot}}=200$ posterior bootstrap realizations for both Robust--MEM and NPL--HMC in synthetic experiments, and $B_{\mathrm{boot}}=100$ for real-world experiments. Since we do not assume a strong prior, the DP concentration parameter is set to $c=10^{-4}$ for both methods. In all experiments, we use Gaussian (RBF) kernels for $k_X$ and $k_Y$, with bandwidth selected by the median heuristic in every MMD computation \citep{gretton2012kernel}.

All HMC sampling is performed using \texttt{cmdstanpy} (the Python interface to \texttt{CmdStan}). Code to reproduce all results in this paper is available at \url{https://github.com/MengqiChenMC/tot_robust_code}.

\section{HMC diagnostics and sensitivity}
\label{sec:HMC}
\subsection{HMC mixing diagnostics}
\label{subsec:HMC-mixing-diagnostics}
We report diagnostics for $\theta$ in the HMC runs used to produce pseudo-samples under the setting with ME scale \(1.5\), \(10\%\) Huber contamination (contaminated points having \(9\times\) the clean noise scale), and a working ME scale equal to \(0.7\times\) the true scale. Four chains were run with \((T,B)=(10{,}000,5{,}000)\) and \(20{,}000\) post-warm-up draws were retained in total for both the classical and Berkson ME models. Table~\ref{tab:hmc-diag} summarizes the scalar diagnostics for the components of \(\theta\) together with sampler-level checks. All \(\hat R\) values are \(1.0\), bulk and tail effective sample sizes are large, and Monte Carlo standard errors are small relative to posterior standard deviations. There were no divergent transitions, no iterations reached the configured maximum tree depth (10), and the per-chain BFMI values are high in both models ($\ge 0.92$ for all chains), indicating good exploration of energy levels. Fig.~\ref{fig:hmc-trace} shows well-mixed traces with stable marginal densities, consistent with sampling from the stationary distribution and low autocorrelation in the retained states. Fig.~\ref{fig:hmc-energy} overlays the marginal and transition energy densities and reports the BFMI per chain; the close overlap supports the absence of pathologies. These checks justify using the HMC draws of \(\theta\) to implement the independent posterior predictive scheme described in the paper for both ME models.

\begin{table}[t]
\centering
\small
\begin{tabular}{llrrrrrrrrr}
\hline
\textbf{Model} &  & Mean & SD & HDI 3\% & HDI 97\% & MCSE mean & MCSE SD & ESS bulk & ESS tail & \(\hat R\) \\
\hline
\multirow{3}{*}{\textbf{Classical ME}}
& \(\theta_1\) & 4.965&0.158&4.676&5.268&0.002&0.001&7089&11685&1.0 \\
& \(\theta_2\) & 1.322&0.212&0.943&1.716&0.003&0.002&6199&7580&1.0 \\
& \(\theta_3\) & 0.052&0.152&-0.233&0.338&0.002&0.001&5286&8461&1.0 \\
\multicolumn{11}{l}{\textit{Sampler-level:} draws \(=20000\);\; divergences \(=0\);\; max tree depth \(=10\) with 0 hits.} \\
\hline
\multirow{3}{*}{\textbf{Berkson ME}}
& \(\theta_1\) & 4.907&0.164&4.608&5.223&0.002&0.001&6512&9658&1.0 \\
& \(\theta_2\) & 1.814&0.278&1.319&2.337&0.004&0.002&5901&9086&1.0 \\
& \(\theta_3\) & -0.087&0.124&-0.326&0.139&0.002&0.001&5332&8806&1.0 \\
\multicolumn{11}{l}{\textit{Sampler-level:} draws \(=20000\);\; divergences \(=0\);\; max tree depth \(=10\) with 0 hits.} \\
\hline
\end{tabular}
\caption{HMC summary diagnostics for \(\theta\) under classical and Berkson ME.}
\label{tab:hmc-diag}
\end{table}

\begin{figure}[t]
  \centering
  \includegraphics[width=0.8\textwidth]{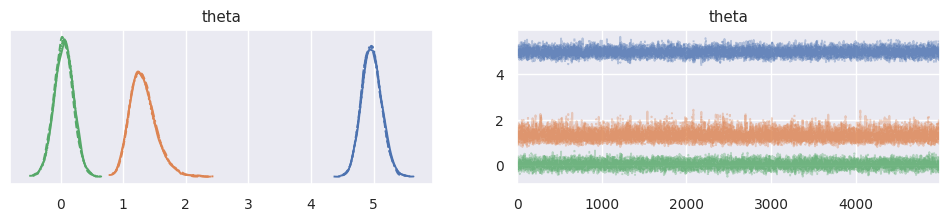}\\[1ex]
  \includegraphics[width=0.8\textwidth]{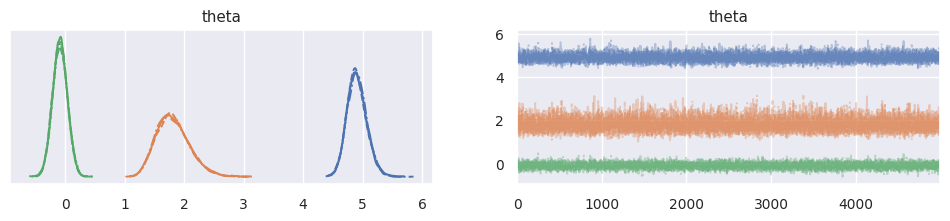}
  \caption{Trace and marginal density for \(\theta\) ($\theta_1$: blue, $\theta_2$: orange, $\theta_3$: green) across four chains: classical (top) and Berkson (bottom).}
  \label{fig:hmc-trace}
\end{figure}

\begin{figure}[t]
  \centering
  \begin{minipage}{0.4\textwidth}
    \centering
    \includegraphics[width=\linewidth]{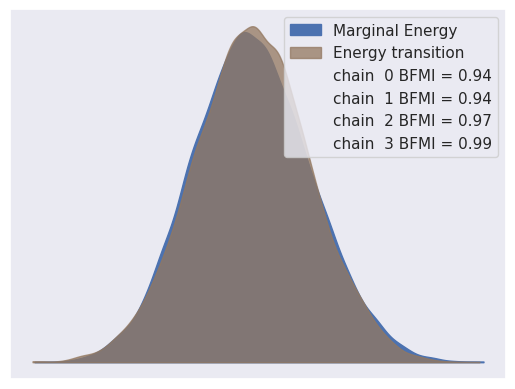}
  \end{minipage}\hfill
  \begin{minipage}{0.4\textwidth}
    \centering
    \includegraphics[width=\linewidth]{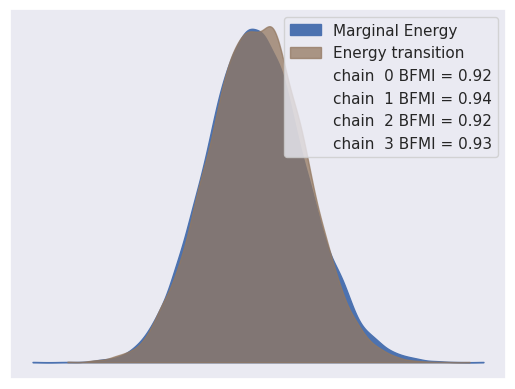}
  \end{minipage}
  \caption{Marginal energy and energy transitions with per-chain BFMI: classical (left) and Berkson (right).}
  \label{fig:hmc-energy}
\end{figure}

\subsection{Sensitivity analysis for the HMC posterior bootstrap}
\label{sec:HMC-sensitivity}
We compare three ways to draw the pseudo-samples used by the NPL update.
\begin{enumerate}
    \item \emph{Regime~A} runs $n\times m$ \emph{independent} HMC chains and retains one post-warm-up draw $\theta_{ij}\sim\Pi_n(\theta\mid W_{1:n},Y_{1:n})$ from each, followed by one latent draw $X_{ij}\sim \Pi(X_i\mid \theta_{ij},W_i,Y_i)$. This aligns exactly with the theoretical construct in Section~\ref{sec:pseudo-sampling} but is computationally expensive.
    \item \emph{Regime~B} fits a single multi-chain HMC run with four parallel chains and takes every 50th state for $\theta$, pairing the corresponding latent draws $X_i$ from the same iterations. This assesses sensitivity to thinning.
    \item \emph{Regime~C} (default) uses the same multi-chain HMC (four parallel chains) but, instead of systematic thinning, selects $m$ spaced-out states to define $\theta_{ij}$ and the corresponding $X_{ij}$ for each $i$. No additional thinning is applied. 
\end{enumerate}

We use the classical or Berkson ME model with the same ME, model misspecification, and HMC settings as in Section~\ref{subsec:HMC-mixing-diagnostics}. We consider \(n\in\{50,100,500\}\) with \(m=3\) (so \(N=n m\) pseudo-samples per regime).

We compare regimes using the unbiased MMD$^2$ with a Gaussian kernel and bandwidth fixed by the median heuristic, applied to the \emph{joint} empirical law of $(X,Y)$. We report (i) the point estimate $\widehat{\MMD}^{2}$, (ii) a bootstrap $95\%$ confidence interval (resampling within each regime), and (iii) a permutation $p$-value based on $2{,}000$ randomizations. We present results for $n\in\{50,100,500\}$ and $m=3$ under both Berkson and classical ME models: see Tables~\ref{tab:mmd-sensitivity-joint-berkson} and~\ref{tab:mmd-sensitivity-joint-classical}. 
The unbiased MMD$^2$ estimator can be slightly negative in finite samples. Under equality of distributions, it is $O_p(1/N)$ with $N=n m$.

\begin{table}[t]
\centering
\small
\begin{tabular}{lccc}
\toprule
Comparison & joint \(\widehat{\MMD}^{2}\) & Bootstrap 95\% CI & Permutation \(p\)-value \\
\midrule
\multicolumn{4}{c}{\(n=50,\; m=3\)}\\
\midrule
A vs B & \(-5.51\times 10^{-3}\) & \([-4.71\times 10^{-3},\; 1.41\times 10^{-2}]\) & \(1.000\) \\
A vs C & \(-4.32\times 10^{-3}\) & \([-4.44\times 10^{-3},\; 1.49\times 10^{-2}]\) & \(0.919\) \\
B vs C & \(-3.82\times 10^{-3}\) & \([-4.18\times 10^{-3},\; 1.52\times 10^{-2}]\) & \(0.846\) \\
\addlinespace
\multicolumn{4}{c}{\(n=100,\; m=3\)}\\
\midrule
A vs B & \(-2.65\times 10^{-3}\) & \([-2.28\times 10^{-3},\; 6.69\times 10^{-3}]\) & \(0.997\) \\
A vs C & \(-2.21\times 10^{-3}\) & \([-2.25\times 10^{-3},\; 7.51\times 10^{-3}]\) & \(0.925\) \\
B vs C & \(-2.57\times 10^{-3}\) & \([-2.21\times 10^{-3},\; 6.72\times 10^{-3}]\) & \(0.990\) \\
\addlinespace
\multicolumn{4}{c}{\(n=500,\; m=3\)}\\
\midrule
A vs B & \(-5.56\times 10^{-4}\) & \([-4.88\times 10^{-4},\; 1.20\times 10^{-3}]\) & \(1.000\) \\
A vs C & \(-5.56\times 10^{-4}\) & \([-4.68\times 10^{-4},\; 1.24\times 10^{-3}]\) & \(1.000\) \\
B vs C & \(-5.46\times 10^{-4}\) & \([-4.73\times 10^{-4},\; 1.25\times 10^{-3}]\) & \(1.000\) \\
\bottomrule
\end{tabular}
\caption{Berkson ME: sensitivity of pseudo-sampling schemes across sample sizes.}
\label{tab:mmd-sensitivity-joint-berkson}
\end{table}

\begin{table}[t]
\centering
\small
\begin{tabular}{lccc}
\toprule
Comparison & joint \(\widehat{\MMD}^{2}\) & Bootstrap 95\% CI & Permutation \(p\)-value \\
\midrule
\multicolumn{4}{c}{\(n=50,\; m=3\)}\\
\midrule
A vs B & \(-4.72\times 10^{-3}\) & \([-4.65\times 10^{-3},\; 1.64\times 10^{-2}]\) & \(0.960\) \\
A vs C & \(-3.94\times 10^{-3}\) & \([-3.94\times 10^{-3},\; 1.72\times 10^{-2}]\) & \(0.859\) \\
B vs C & \(-4.93\times 10^{-3}\) & \([-4.47\times 10^{-3},\; 1.62\times 10^{-2}]\) & \(0.977\) \\
\addlinespace
\multicolumn{4}{c}{\(n=100,\; m=3\)}\\
\midrule
A vs B & \(-2.67\times 10^{-3}\) & \([-2.27\times 10^{-3},\; 7.35\times 10^{-3}]\) & \(1.000\) \\
A vs C & \(-2.86\times 10^{-3}\) & \([-2.32\times 10^{-3},\; 6.72\times 10^{-3}]\) & \(1.000\) \\
B vs C & \(-2.72\times 10^{-3}\) & \([-2.41\times 10^{-3},\; 6.31\times 10^{-3}]\) & \(0.999\) \\
\addlinespace
\multicolumn{4}{c}{\(n=500,\; m=3\)}\\
\midrule
A vs B & \(-5.56\times 10^{-4}\) & \([-4.73\times 10^{-4},\; 1.17\times 10^{-3}]\) & \(0.9995\) \\
A vs C & \(-5.70\times 10^{-4}\) & \([-4.72\times 10^{-4},\; 1.18\times 10^{-3}]\) & \(1.000\) \\
B vs C & \(-5.65\times 10^{-4}\) & \([-4.76\times 10^{-4},\; 1.35\times 10^{-3}]\) & \(1.000\) \\
\bottomrule
\end{tabular}
\caption{Classical ME: sensitivity of pseudo-sampling schemes across sample sizes.}
\label{tab:mmd-sensitivity-joint-classical}
\end{table}

Across $n\in\{50,100,500\}$ the three pairwise \emph{joint} MMD$^2$ estimates are close to zero and decrease in magnitude as $N=n m$ increases, consistent with the $O_p(1/N)$ scale under equality. The bootstrap intervals contain $0$ and the permutation $p$-values are large (Tables~\ref{tab:mmd-sensitivity-joint-berkson}-\ref{tab:mmd-sensitivity-joint-classical}) for all $n$. There is no evidence that the joint distribution of the pseudo-samples $\{(X_{ij},Y_i)\}$ differs across regimes. In particular, using a few well-mixed chains with sparse retention and paired latent draws yields pseudo-samples that are empirically indistinguishable from those obtained by launching $n\times m$ independent chains. Hence our practical implementation gives the same pseudo-sample distribution, within Monte Carlo uncertainty, as the theoretical construction.

\subsection{Discussion: independence requirement of Theorem~\ref{thm:pseudo-classical}}
\label{appendix-independence-discussion}
Our theoretical construct in Section~\ref{sec:pseudo-sampling} imposes independence of the pseudo-samples \(\{X_{ij}\}\) given \(\D\) by drawing \(\theta_{ij}\stackrel{\text{iid}}{\sim}\Pi_n(\theta\mid\D)\) and then \(X_{ij}\sim \Pi(\,\cdot\mid W_i,Y_i,\theta_{ij})\). In this section we relax that requirement across \((i,j)\): we draw \(m\) parameter states \(\theta_j\sim \Pi_n(\theta\mid\D)\) that may be dependent (e.g. states from one or a few HMC chains), and for each fixed \(j\) we set \(X_{ij}\sim \Pi(\,\cdot\mid W_i,Y_i,\theta_j)\) for \(i=1,\dots,n\). Conditional on \((\D,\theta_j)\), the collection \(\{X_{ij}\}_{i=1}^{n}\) is independent across \(i\) by the i.i.d.\ nature of \(\{(W_i,Y_i)\}\).

We show below that, in finite samples with small \(n\), enforcing independence of \(\theta_{ij}\) can reduce the sampling error bound (term~A in the proof of Theorem~\ref{thm:pseudo-classical}) from \(2\sqrt{\kappa}/\sqrt{n}+2\,{M_n}/{\sqrt{n}}+2\sqrt{\kappa}\,r_n\) to \(2\sqrt{\kappa}/\sqrt{nm}\). However, due to posterior contraction, the parameter-mixture contribution associated with the \(\theta\)-mixture in \eqref{eq:predictive-kernel} can be bounded by a quantity of order \({M_n}/{\sqrt{n}}+r_n\) regardless of whether or not the \(\theta_{ij}\) are independent given $\D$. As \(n\) grows (with \(m\) fixed), the gain from enforcing independence across \(\theta_{ij}\) becomes negligible.

This yields the following practical implementation guide:

\begin{enumerate}
    \item \emph{Small \(n\):} the posterior \(\Pi_n(\theta\mid\D)\) may be dispersed; independent (or well-spaced) posterior draws of \(\theta\) can improve exploration of the posterior and reduce Monte Carlo error in the pseudo-samples, at modest cost when \(m\) is small. 
    \item \emph{Large \(n\):} as \(\Pi_n\) concentrates, the bound is driven by \(\frac{M_n}{\sqrt{n}}+r_n\) and near-independence of \(\theta\) is unnecessary. Our sensitivity analysis above (Appendix~\ref{sec:HMC-sensitivity}) empirically confirms that both implementations deliver indistinguishable pseudo-sample distributions as $n$ increases in the regimes considered.
\end{enumerate}

We now derive an alternative bound for term A in the proof of Theorem~\ref{thm:pseudo-classical} without the independence condition on $\{\tilde X_{ij}\}_{i=1,j=1}^{n,\;\;\;\,m}$. All expectations over \(\theta_{1:m}\) are taken with respect to their joint law induced by the sampler. For each fixed \(j\), conditional on \((\D,\theta_j)\) the latent coordinates factorize as
\[
\Pi(X_{1:n,j}\mid \D,\theta_j)=\prod_{i=1}^n \Pi(X_{ij}\mid W_i,Y_i,\theta_j),
\]
because the observations \(\{(W_i,Y_i)\}_{i=1}^n\) are iid. Therefore, the model implies conditional independence of \(X_i\) given \((W_i,Y_i,\theta_j)\). The proof below first exploits this conditional independence to obtain the \(1/\sqrt{n}\) bound, then averages over \(j\), and finally integrates over the (possibly dependent) vector \(\theta_{1:m}\) using Jensen’s inequality and applies posterior contraction.

Recall that term A is
\[
\E_{\D,S} \MMD_{k}(\P^{\mathrm{pseudo}}_{XY},Q_n),
\]
where
\( Q_{i} =\Psi_n\bigl(\cdot\mid W_i,Y_i\bigr)\delta_{Y_i}\)
and its embedding is \( \varphi(Q_i)\in\Hilb_k \).

Given \((\D,S)\equiv\bigl\{(W_i,Y_i);(\tilde X_{ij})\bigr\}_{i=1,j=1}^{n,\;\;\;\,m}\) we form the empirical measures for each $i$
\[ \widehat{\Q}_i :=\frac1m\sum_{j=1}^{m}\delta_{(\tilde X_{ij},Y_i)}, \qquad \varphi\bigl(\widehat{\Q}_i\bigr) =\frac1m\sum_{j=1}^{m} k\bigl((\tilde X_{ij},Y_i),\cdot\bigr)\in\Hilb_k. \]
Rewrite the MMD by re-indexing the double sum:
\begin{equation}\label{eq:A-decomp}
\begin{aligned}
\MMD_{k}(\P^{\mathrm{pseudo}}_{XY},Q_n) &= 
\norm{\varphi\left(\frac{1}{n}\sum_{i=1}^n\widehat{\Q}_i\right) - \varphi(Q_n)}_{_{\Hilb_k}}
\\&\le \frac{1}{m}\norm{\sum_{j=1}^{m}\left\{
      \varphi\left(\frac{1}{n}\sum_{i=1}^{n}\delta_{(\tilde X_{ij},Y_i)}\right)
      -\varphi\left(\frac{1}{n}\sum_{i=1}^{n}\tilde{\Q}_{ij}\right)
   \right\}}_{_{\Hilb_k}}  \\
&\quad+\frac{1}{m}\norm{\sum_{j=1}^{m}\left\{\varphi\left(\frac{1}{n}\sum_{i=1}^{n}\tilde{\Q}_{ij}\right)-\varphi(Q_n)\right\}}_{\Hilb_k},
\end{aligned}
\end{equation}

where, for each \(i,j\), we set
\(
  \tilde{\Q}_{ij}:=\Pi(\cdot\mid W_i,Y_i,\theta_j)\delta_{Y_i}.
\)
Fix \(j\). Conditional on \((\D,\theta_j)\), the vectors
\[
\phi_{ij}:=k\bigl((\tilde X_{ij},Y_i),\cdot\bigr)\in\Hilb_k,\qquad i=1,\dots,n,
\]
are independent with mean
\(
  \E_{S\mid\D,\theta_j}[\phi_{ij}]
  =\varphi(\tilde{\Q}_{ij})
\)
and
\(
  \|\phi_{ij}\|_{\Hilb_k}^2
  \le \kappa_X\kappa_Y.
\)
Therefore,
\begin{equation*}
\begin{aligned}
&\E_{S\mid\D,\theta_j}
\left\|
   \varphi\left(\frac{1}{n}\sum_{i=1}^{n}\delta_{(\tilde X_{ij},Y_i)}\right)
  -\varphi\left(\frac{1}{n}\sum_{i=1}^{n}\tilde{\Q}_{ij}\right)
\right\|_{\Hilb_k}^{2} \\
&\qquad=
\E_{S\mid\D,\theta_j}\left\|
   \frac{1}{n}\sum_{i=1}^{n}\bigl(\phi_{ij}-\E_{S\mid\D,\theta_j}\phi_{ij}\bigr)
\right\|_{\Hilb_k}^{2}
=\frac{1}{n^{2}}\sum_{i=1}^{n}
  \E_{S\mid\D,\theta_j}\left\|
     \phi_{ij}-\E_{S\mid\D,\theta_j}\phi_{ij}
  \right\|_{\Hilb_k}^{2} \\
&\qquad\le \frac{1}{n^{2}}\sum_{i=1}^{n}\bigl(2\sqrt{\kappa_X\kappa_Y}\bigr)^{2}
= \frac{4\kappa_X\kappa_Y}{n}.
\end{aligned}
\end{equation*}
By Jensen’s inequality,
\[
\E_{S\mid\D,\theta_j}
\left\|
   \varphi\left(\frac{1}{n}\sum_{i=1}^{n}\delta_{(\tilde X_{ij},Y_i)}\right)
  -\varphi\left(\frac{1}{n}\sum_{i=1}^{n}\tilde{\Q}_{ij}\right)
\right\|_{\Hilb_k}
\le \frac{2\sqrt{\kappa_X\kappa_Y}}{\sqrt{n}}.
\]
Averaging over \(j=1,\dots,m\) and then over \(\theta_{1:m}\) and \(\D\), we obtain the bound
\begin{equation}
\begin{aligned}
\label{eq:term-no-mix}
&\E_{\D}\E_{\theta_{1:m},S\mid\D}
\left\|
   \frac{1}{m}\sum_{j=1}^{m}
     \Biggl[
       \varphi\left(\frac{1}{n}\sum_{i=1}^{n}\delta_{(\tilde X_{ij},Y_i)}\right)
      -\varphi\left(\frac{1}{n}\sum_{i=1}^{n}\tilde{\Q}_{ij}\right)
     \Biggr]
\right\|_{\Hilb_k}\le \frac{2\sqrt{\kappa_X\kappa_Y}}{\sqrt{n}}.
\end{aligned}
\end{equation}
The extra (parameter-mixture) term can be bounded by posterior contraction. By the triangle inequality in \(\Hilb_k\) and linearity of \(\varphi(\cdot)\),
\begin{equation*}
\left\|
   \varphi\left(\frac{1}{n}\sum_{i=1}^{n}\tilde{\Q}_{ij}\right)
  -\varphi(Q_n)
\right\|_{\Hilb_k}
\le \frac{1}{n}\sum_{i=1}^{n}
  \MMD_{k}\left(\tilde{\Q}_{ij},\Q_i\right),
\qquad
\Q_i:=\Psi_n(\cdot\mid W_i,Y_i)\delta_{Y_i}.
\end{equation*}
By Lemma~\ref{lem:MMD-Jensen} and Lemma~\ref{lem:MMD-conditional},
\begin{equation*}
\begin{aligned}
\MMD_{k}\left(\tilde{\Q}_{ij},\Q_i\right)
&=\MMD_{k}\left(
     \Pi_{\theta_j}(\cdot\mid W_i,Y_i)\delta_{Y_i},
     \int \Pi_{\vartheta}(\cdot\mid W_i,Y_i)\delta_{Y_i}\Pi_n(d\vartheta)
   \right) \\
&\le \int
     \MMD_{k}\left(
       \Pi_{\theta_j}(\cdot\mid W_i,Y_i)\delta_{Y_i},
       \Pi_{\vartheta}(\cdot\mid W_i,Y_i)\delta_{Y_i}
     \right)\Pi_n(d\vartheta) \\
&\le \sqrt{\kappa_Y}\int
     \MMD_{k_X}\left(
       \Pi_{\theta_j}(\cdot\mid W_i,Y_i),
       \Pi_{\vartheta}(\cdot\mid W_i,Y_i)
     \right)\Pi_n(d\vartheta).
\end{aligned}
\end{equation*}
Split the \(\vartheta\)-integral over \(B_n\cup B_n^c\).
On \(B_n\), Assumption~\ref{A2} gives
\[
  \MMD_{k_X}\left(\Pi_{\theta_j}(\cdot\mid W_i,Y_i),\Pi_{\vartheta}(\cdot\mid W_i,Y_i)\right)
  \le L(W_i,Y_i)\|\theta_j-\vartheta\|.
\]
Hence, for any fixed \(\theta_j\),
\begin{equation*}
    \int_{B_n}\MMD_{k_X}\bigl(\Pi_{\theta_j},\Pi_{\vartheta}\bigr)\Pi_n(d\vartheta)
\le
\begin{cases}
L(W_i,Y_i)\int_{B_n}\|\theta_j-\vartheta\|\Pi_n(d\vartheta) \le \frac{2M_n}{n}L(W_i,Y_i), & \theta_j\in B_n,\\
\sqrt{\kappa_X}\Pi_n(B_n), & \theta_j\in B_n^c,
\end{cases}
\end{equation*}
where we used \(\|\theta_j-\vartheta\|\le \|\theta_j-\theta^\ast\|+\|\vartheta-\theta^\ast\|\le 2M_n/\sqrt{n}\) in the first case and \(\MMD_{k_X}\le \sqrt{\kappa_X}\) in the second.
On \(B_n^c\), we have
\[
\int_{B_n^c}\MMD_{k_X}\bigl(\Pi_{\theta_j},\Pi_{\vartheta}\bigr)\Pi_n(d\vartheta)
\le \sqrt{\kappa_X}\Pi_n(B_n^c).
\]
Combining the pieces and averaging over \(i\) gives, for each fixed \(j\) and \(\theta_j\),
\[
\left\|
   \varphi\left(\frac{1}{n}\sum_{i=1}^{n}\tilde{\Q}_{ij}\right)
  -\varphi(Q_n)
\right\|_{\Hilb_k}
\le \sqrt{\kappa_Y}\Biggl[
   \indicator_{\{\theta_j\in B_n\}}
     \cdot \frac{2M_n}{\sqrt{n}}\bar L(W,Y)
 + \indicator_{\{\theta_j\in B_n^c\}}
     \cdot \sqrt{\kappa_X}\Pi_n(B_n)
 + \sqrt{\kappa_X}\Pi_n(B_n^c)
\Biggr],
\]
where \(\bar L(W,Y):=\frac{1}{n}\sum_{i=1}^{n}L(W_i,Y_i)\).
Taking expectation over \(\theta_j\sim\Pi_n(\cdot\mid\D)\) and using
\[
\E_{\theta_j\mid\D}\bigl[\indicator_{\{\theta_j\in B_n\}}\bigr]=\Pi_n(B_n),\qquad
\E_{\theta_j\mid\D}\bigl[\indicator_{\{\theta_j\in B_n^c\}}\bigr]=\Pi_n(B_n^c),
\]
together with
\(
\int_{B_n}\|\vartheta-\theta^\ast\|\Pi_n(d\vartheta)\le M_n/\sqrt{n}
\),
we obtain
\[
\E_{\theta_j\mid\D}
\left\|
   \varphi\left(\frac{1}{n}\sum_{i=1}^{n}\tilde{\Q}_{ij}\right)
  -\varphi(Q_n)
\right\|_{\Hilb_k}
\le \frac{2\sqrt{\kappa_Y}M_n}{\sqrt{n}}\bar L(W,Y)
 +2\sqrt{\kappa_X\kappa_Y}\Pi_n(B_n^c).
\]
By convexity of the norm,
\[
\E_{\theta_{1:m}\mid\D}
\left\|
   \frac{1}{m}\sum_{j=1}^{m}
     \left\{
       \varphi\left(\frac{1}{n}\sum_{i=1}^{n}\tilde{\Q}_{ij}\right)
      -\varphi(Q_n)
     \right\}
\right\|_{\Hilb_k}
\le \frac{2\sqrt{\kappa_Y}M_n}{\sqrt{n}}\bar L(W,Y)
 + 2\sqrt{\kappa_X\kappa_Y}\Pi_n(B_n^c).
\]
Finally, taking expectation over \(\D\) and using
\(C_L:=\E[\bar L(W,Y)]<\infty\) and \(r_n:=\E_{\D}[\Pi_n(B_n^c)]\to 0\),
we obtain
\begin{equation}
\label{eq:term-mix}
\E_{\D}\E_{\theta_{1:m}\mid\D}
\left\|
   \frac{1}{m}\sum_{j=1}^{m}
     \left\{
       \varphi\left(\frac{1}{n}\sum_{i=1}^{n}\tilde{\Q}_{ij}\right)
      -\varphi(Q_n)
     \right\}
\right\|_{\Hilb_k}
\le \frac{2\sqrt{\kappa_Y}C_LM_n}{\sqrt{n}}
  + 2\sqrt{\kappa_X\kappa_Y}r_n.
\end{equation}
As in \eqref{eq:Qn-1'}, rescale \(M_n\) by a fixed constant if desired (replace \(M_n\) with \(M_n/\max\{2\sqrt{\kappa_Y}C_L,1\}\)) to write the right-hand side as \( \frac{M_n}{\sqrt{n}} + \sqrt{\kappa_X\kappa_Y}r_n \).
Combining \eqref{eq:term-no-mix} and \eqref{eq:term-mix} gives
\begin{equation}
    \E_{\D,S} \MMD_{k}(\P^{\mathrm{pseudo}}_{XY},Q_n)\le \frac{2\sqrt{\kappa}}{\sqrt{n}} + \frac{2M_n}{\sqrt{n}}+2\sqrt{\kappa}r_n.
\end{equation}

\end{document}

%% file: preamble_defs_arXiv.tex
\newcommand{\indicator}{\mathbf{1}}
\providecommand{\mathscr}{\mathcal}

\newcommand\norm[1]{\left\lVert#1\right\rVert}
\newcommand\abs[1]{\left\lvert#1\right\rvert}

\newcommand\given[1][]{\:#1\vert\:}

\newcommand\ind{\protect\mathpalette{\protect\independenT}{\perp}}
\def\independenT#1#2{\mathrel{\rlap{$#1#2$}\mkern2mu{#1#2}}}

\newcommand{\D}{\mathcal{D}}
\renewcommand{\P}{\mathbb{P}}

\DeclareMathOperator{\R}{\mathbb{R}}

\newcommand{\Hk}{\mathcal{H}_k}

\DeclareMathOperator{\simiid}{\overset{iid}{\sim}}
\DeclareMathOperator*{\argmin}{arg\,min}

\DeclareMathOperator{\Q}{\mathbb{Q}}
\newcommand{\E}{E}

\DeclareMathOperator{\MMD}{\operatorname{MMD}}

\DeclareMathOperator{\DP}{\operatorname{DP}}

\DeclareMathOperator{\Var}{\operatorname{var}}

\newcommand{\KL}{\operatorname{KL}}
\newcommand{\TV}{\operatorname{TV}}
\newcommand{\Hilb}{\mathcal H}
\newcommand{\dd}{\mathrm d}
\def\T{{ \mathrm{\scriptscriptstyle T} }}
\let\Pr\relax
\DeclareMathOperator{\Pr}{\operatorname{Pr}}
\definecolor{mygreen}{RGB}{9,121,105} 
\definecolor{mypink}{RGB}{231,84,128}
\newif\ifshowcomments
\showcommentstrue

\ifshowcomments
  \newcommand*\td[1]{{\color{blue} [\textbf{Theo}: #1]}}
  \newcommand*\hd[1]{{\color{orange} [\textbf{Harita}: #1]}}
  \newcommand*\tb[1]{{\color{mygreen} [\textbf{Tom}: #1]}}
  
\else
  \newcommand*\td[1]{}
  \newcommand*\hd[1]{}
  \newcommand*\tb[1]{}
  
\fi

%% file: Figures/casual_diagram.tex
\begin{tikzpicture}[x=1cm,y=1cm]

\tikzset{
obs/.style={circle,draw,thick,minimum size=6.5mm,inner sep=0pt},
latent/.style={circle,draw,thick,dashed,minimum size=6.5mm,inner sep=0pt},
noise/.style={circle,draw,thick,dashed,fill=gray!20,minimum size=6.5mm,inner sep=0pt},
method/.style={rectangle,draw,thick,rounded corners=2pt,minimum height=6.5mm,minimum width=18mm,align=center},
loss/.style={rectangle,draw,thick,rounded corners=2pt,minimum height=6.5mm,minimum width=14mm,align=center},
arr/.style={-{Stealth[length=2.2mm]},thick},
darr/.style={-{Stealth[length=2.2mm]},thick,dashed},
lab/.style={font=\scriptsize,inner sep=1pt},
tlab/.style={font=\small}
}

\begin{scope}[xshift=0cm]
\node[tlab] at (0.75,-0.75) {(a) Ignoring ME};
\node[obs]    (Wa) at (0,0) {$W$};
\node[obs]    (Ya) at (1.5,0) {$Y$};
\node[noise]  (Ea) at (1.5,1.25) {$E$};

\draw[arr] (Wa) -- (Ya);
\draw[arr] (Ea) -- (Ya);
\end{scope}

\begin{scope}[xshift=4.2cm]
\node[tlab] at (1.5,-0.75) {(b) Classical ME};

\node[obs]    (Wb) at (0,0) {$W$};
\node[latent] (Xb) at (1.5,0) {$X$};
\node[obs]    (Yb) at (3,0) {$Y$};
\node[noise]  (Nb) at (0,1.25) {$N$};
\node[noise]  (Eb) at (3,1.25) {$E$};

\draw[arr] (Xb) -- (Wb);
\draw[arr] (Xb) -- node[lab,above] {$g^0(\cdot)$} (Yb);
\draw[arr] (Nb) -- (Wb);
\draw[arr] (Eb) -- (Yb);
\end{scope}

\begin{scope}[xshift=10cm]
\node[tlab] at (1.5,-0.75) {(c) Berkson ME};
\node[obs]    (Wa) at (0,0) {$W$};
\node[latent] (Xa) at (1.5,0) {$X$};
\node[obs]    (Ya) at (3,0) {$Y$};
\node[noise]  (Na) at (1.5,1.25) {$N$};
\node[noise]  (Ea) at (3,1.25) {$E$};

\draw[arr] (Wa) -- (Xa);
\draw[arr] (Xa) -- node[lab,above] {$g^0(\cdot)$} (Ya);
\draw[arr] (Na) -- (Xa);
\draw[arr] (Ea) -- (Ya);
\end{scope}

\end{tikzpicture}

%% file: Figures/diagram_comparison.tex
\begin{tikzpicture}[
  >=Stealth,
  obs/.style={circle,draw,thick,minimum size=7mm,inner sep=0pt},
  lat/.style={circle,draw,thick,dashed,minimum size=7mm,inner sep=0pt},
  noise/.style={circle,draw,thick,dashed,fill=gray!20,minimum size=7mm,inner sep=0pt},
  box/.style={rectangle,rounded corners,draw,thick,align=center,inner sep=3pt},
  arr/.style={-Stealth,thick,dashed},
  darr/.style={-Stealth,thick,dotted},
  doarr/.style={-Stealth,thick},
  lab/.style={inner sep=1pt}
]

\begin{scope}[xshift=-0.5cm]
  \node[font=\small] at (0,1.15) {(a) Joint prior + response-informed posterior};

  \node[obs]   (Wb) at (-1.8,0) {$W$};
  \node[lat]   (Xb) at (0,0) {$X$};
  \node[obs]   (Yb) at (1.8,0) {$Y$};

  \draw[doarr] (Wb) -- (Xb);
  \draw[doarr] (Xb) -- node[lab,above] {$g^0$} (Yb);

  \node[obs] (Xtb) at (0,-1.85) {$\tilde X$};
  \node[lab] at (0,-0.85) {$\tilde X \mid (W,Y)$};

  \node[box,minimum width=3.2cm] (Refb) at (0,-3.35)
    {$\P^{\DP}_{XY} \leftarrow (X^\prime,Y^\prime)\oplus(\tilde X,Y)$};
  \draw[darr] (Wb) to[out=-85,in=180] node[lab,near start,below left] {$X^\prime\mid W \times \tilde Y^\prime \mid X^\prime$} (Refb);

  \draw[arr] (Wb) -- (Xtb);
  \draw[arr] (Yb) -- (Xtb);
  \draw[arr]  (Xtb) -- (Refb);
\end{scope}

\begin{scope}[xshift=6.9cm]
  \node[font=\small] at (0,1.15) {(b) Response-agnostic prior  + ME-ignoring posterior}; 

  \node[obs]   (Wa) at (-1.8,0) {$W$};
  \node[lat]   (Xa) at (0,0) {$X$};
  \node[obs]   (Ya) at (1.8,0) {$Y$};

  \draw[doarr] (Wa) -- (Xa);
  \draw[doarr] (Xa) -- node[lab,above] {$g^0$} (Ya);

  \node[obs] (Xta) at (0,-1.85) {$\tilde X$};
  \node[lab] at (0,-1.0) {$X^\prime \mid W$};

  \node[box,minimum width=3.2cm] (Refa) at (0,-3.35)
    {$\P^{\DP}_{XY} \leftarrow (X^\prime,Y)\oplus(W,Y)$};

  \draw[darr] (Wa) -- (Xta);

  \draw[arr] (Wa) to[out=-85,in=180] node[lab,near start,below left] {$W$} (Refa);
  \draw[arr] (Ya) to[out=-95,in=0] node[lab,near start,below right] {$Y$} (Refa.east);

  \draw[darr] (Xta) -- (Refa);
\end{scope}

\end{tikzpicture}